\documentclass[letterpaper,UKenglish]{lipics}
\usepackage{stmaryrd}
\usepackage{amssymb}
\usepackage{empheq}
\usepackage{verbatim}
\usepackage{graphicx}
\usepackage{color}
\usepackage{xcolor,pgf,tikz,pgflibraryarrows,pgffor,pgflibrarysnakes}
\usetikzlibrary{fit} 
\usetikzlibrary{backgrounds} 

\usepackage{blkarray}
\usepackage{slashbox}
\usepgflibrary{shapes}
\usetikzlibrary{snakes,automata}

\tikzstyle{background}=[rectangle,fill=gray!10, inner sep=0.1cm, rounded corners=0mm]

\newcommand\inter[1]{\llbracket #1 \rrbracket}

\newcommand{\matindex}[1]{\mbox{\scriptsize#1}}%

\newcommand{\Ss}{\mathcal{S}}
\newcommand{\Aa}{\mathcal{A}}
\newcommand{\Mm}{\mathcal{M}}

\newcommand{\suc}{\text{succ}}
\newcommand{\seq}[1]{\langle #1 \rangle}
\newcommand{\ktype}[2]{\ensuremath{\langle #1 \rangle_{#2}}}

\newcommand{\N}{\mathbb N}
\newcommand{\Nat}{\mathbb N}

\newcommand{\set}[1]{\left\{ #1 \right\}}

\newcommand{\struc}[1]{\Xi_{#1}}

\newcommand{\wsst}{\textnormal{SST}}

\newcommand{\sstla}{\ensuremath{\textnormal{SST-la}}} 
\newcommand{\sst}{\wsst{}}

\newcommand{\msot}{\textnormal{MSOT}} 
\newcommand{\fot}{\textnormal{FOT}}

\newcommand{\la}{\ensuremath{{\mathrm{la}}}}

\newcommand{\fst}{\mathrm{first}}
\newcommand{\lst}{\mathrm{last}}

\newcommand{\halve}{\mathrm{halve}}

\newcommand{\istring}{\mathrm{is\_string}}

\newcommand{\MSO}{\mathrm{MSO}}
\newcommand{\FO}{\mathrm{FO}}
\newcommand{\dom}{\mathrm{dom}}
\newcommand{\nodes}{\mathrm{pos}}
\newcommand{\positions}{\mathrm{pos}}

\newcommand{\flow}{\mathrm{flow}}

\newcommand{\addr}{\ensuremath{{\alpha}}}
\newcommand{\Addr}{\ensuremath{{\mathcal{A}}}}

\newcommand{\absaddr}{\textsc{hd}}
\newcommand{\tailaddr}{\textsc{tl}}

\newcommand{\flows}{\rightsquigarrow}

\newcommand{\varsst}{\mathcal{X}}

\newcommand{\contribute}{\text{useful}}

\newcommand{\useful}{\text{useful}}

\newtheorem{proposition}{Proposition}

\usepackage{tikz}
\tikzstyle{nloc}=[draw, text badly centered, rectangle, rounded corners, minimum size=2em,inner sep=0.5em]
\tikzstyle{loc}=[draw,rectangle,minimum size=1.4em,inner sep=0em]
\tikzstyle{trans}=[-latex, rounded corners]
\tikzstyle{trans2}=[-latex, dashed, rounded corners]

\def\rmdef{\stackrel{\mbox{\rm {\tiny def}}}{=}} 

\title{First-order definable string transformations}

\author[1]{Emmanuel Filiot}
\author[2]{Shankara Narayanan Krishna}
\author[2]{Ashutosh Trivedi}

\affil[1]{F.N.R.S. Research Associate \\
  Universit\'e Libre de Bruxelles\\
  Bruxelles, Belgium\\
  \texttt{efiliot@ulb.ac.be}}
\affil[2]{Indian Institute of Technology Bombay\\
  Powai, Mumbai, India\\
  \texttt{krishnas,trivedi@cse.iitb.ac.in}}

\authorrunning{Filiot, Krishna, and Trivedi} 

\Copyright{Filiot, Krishna, and Trivedi}
\EventShortName{arXiv}

\begin{document}

\maketitle
\begin{abstract}
  The connection between languages defined by computational models and
  logic for languages is well-studied. Monadic second-order logic and
  finite automata are shown to closely correspond to each-other for the
  languages of strings, trees, and partial-orders. Similar connections
  are shown for first-order logic and finite automata with certain
  aperiodicity restriction. Courcelle in 1994 proposed a way to use
  logic to define functions over structures where the output structure
  is defined using logical formulas interpreted over the input
  structure. Engelfriet and Hoogeboom discovered the corresponding
  "automata connection" by showing that two-way generalised sequential
  machines capture the class of monadic-second order definable
  transformations. Alur and Cerny further refined the result by
  proposing a one-way deterministic transducer model with string
  variables---called the streaming string transducers---to capture the
  same class of transformations. In this paper we establish a
  transducer-logic correspondence for Courcelle's first-order definable
  string transformations. We propose a new notion of transition monoid
  for streaming string transducers that involves structural properties
  of both underlying input automata and variable dependencies. By
  putting an aperiodicity restriction on the transition monoids, we
  define a class of streaming string transducers that captures exactly the class of first-order
  definable transformations.
\end{abstract}

\section{Introduction}
\label{sec:introduction}
The class of regular languages is among one of the most well-studied concept in
the theory of formal languages. 
Regular languages have been precisely characterized widely by differing formalisms like
monadic second-order logic (MSO), finite state automata, regular expressions, and finite monoids.  
In particular, the connection~\cite{Bu60} between finite state automata and 
monadic second-order logic is  one of the celebrated results of formal language theory. 
Over the years, there has been substantial research  to establish similar connections for 
the languages definable using first-order logic (FO)~\cite{dg08SIWT}.
In particular, first-order definable languages have been shown to be precisely captured by, among others, 
aperiodic finite state automata. 
Aperiodic automata are restrictions of finite automata with certain aperiodicity restrictions 
on their transition matrices defined through aperiodicity of their transition monoid. 
Other formalisms capturing first-order definable languages include counter-free automata, 
star-free regular expressions, and very weak alternating automata.

Starting with the work of Courcelle~\cite{Cour94}, logic and automata connections have also been 
established for the theory of string transformations. 
The first result in this direction is by Engelfriet and Hoogeboom~\cite{EH01}, where 
MSO-definable transformations  have been shown to be equivalent to
two-way finite transducers. This result has then been extended to
trees and macro-tree transducers \cite{Engelfriet03}.
Recently, Alur and  {\v C}ern\'y~\cite{AC10,AC11} introduced \emph{streaming string transducers}, 
a one-way finite transducer model extended with variables, and showed that they precisely capture 
MSO-definable transformations not only in finite string-to-string case, but also for 
infinite strings~\cite{FiliotTrivediLics12} and tree~\cite{AD12,ADT13}
transformations. In this paper, we show a logic and transducer connection for first-order definable string 
transformations, by introducing an appropriate notion of aperiodic transition monoid for 
streaming string transducers.   

Streaming string transducers (SSTs) manipulate a finite set of string variables to 
compute their output as they read the input string in one left-to-right pass. 
Instead of appending symbols to the output tape, SSTs concurrently update all
string variables using a concatenation of output symbols and string variables in
a \emph{copyless} fashion, i.e. no variable occurs more than once in each
concurrent variable update.   
The transformation of a string is then defined using an output (partial)
function $F$ that associates states with a copyless concatenation of string
variables, s.t.\ if the state $q$ is reached after reading the string and
$F(q) {=} XY$, then the output string is the final valuation of $X$ concatenated
with that of $Y$. 
It has been shown that SSTs have good algorithmic properties (such as decidable type-checking, 
equivalence)~\cite{AC10,AC11} and naturally generalize to various settings like 
trees and nested words \cite{AD12,ADT13}, infinite strings~\cite{FiliotTrivediLics12}, 
and quantitative languages \cite{ADD+11}. 

\subsection{Aperiodic Streaming String Transducers} 
Let us consider transformation $f_{\halve}$ defined as $a^n \mapsto a^{\lceil \frac{n}{2}\rceil}$. 
Intuitively, it can be shown (see Appendix \ref{halve} for a proof) that $f_\halve$ is not FO-definable since it requires to distinguish based 
on the parity of the input. 
Consider, the following \sst{} $T_1$ with $2$ accepting states and $1$ variable.
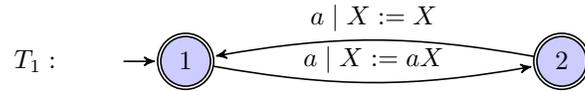
\begin{center}
\begin{tikzpicture}[->,>=stealth',shorten >=1pt,auto,node distance=1cm,
 semithick]
  \tikzstyle{every state}=[fill=blue!20!white,minimum size=2em]
 \node[initial, accepting, initial text={},state,fill=blue!20!white] at (0,0) (A) {1} ;
 \node[state,accepting, fill=blue!20!white] at (5,0) (B) {2} ;
 \path(A) edge  [bend right=10] node[anchor=south,above]{$a \mid X := aX$} (B);
  \path(B) edge [bend right=10] node[anchor=north,above]  {$a \mid X := X$} (A);
 \node[draw=none] at (-2,0) {$T_1:$};
   \end{tikzpicture}
\end{center}
Readers familiar with aperiodic automata may notice that the automata corresponding to $T_1$ is 
not aperiodic, but indeed has period 2. 
Formally such aperiodicity is captured by the notion of automata transition monoid.
The transition monoid of an automaton $A$ is the set of Boolean transition matrices $M_s$, for
all strings $s$, indexed by states of $A$: $M_s[p][q]=1$ iff there
exists a run from $p$ to $q$ on $s$. The set of matrices $M_s$ is a finite
monoid. It is aperiodic if there exists $m\geq 0$ such that for all
$s\in\Sigma^*$, $M_{s^m} = M_{s^{m+1}}$. Aperiodic automata define
exactly first-order languages \cite{Strau94,dg08SIWT}. 
It seems a valid conjecture that SSTs whose transition monoid of underlying automaton is aperiodic characterize 
first-order definable transformations. 
However, unfortunately this is not a sufficient condition as shown by the following \sst{} $T_0$ which 
also implements $f_\halve$ (its output is $F(1) = X$).
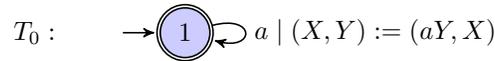
\begin{center}
\begin{tikzpicture}[->,>=stealth',shorten >=1pt,auto,node distance=1cm,
 semithick]
 \tikzstyle{every state}=[fill=purple!20!white,minimum size=2em]
 \node[initial, accepting, initial text={},state,fill=blue!20!white] at (5,0) (A) {1} ;
 \path(A) edge [loop right] node {$a \mid (X,Y) := (aY,X)$} (A); 
 \node[draw=none] at (3,0) {$T_0:$};
\end{tikzpicture}
\end{center}
In this example, although the underlying automaton is aperiodic, variables contribute to certain 
non aperiodicity.  
We capture this idea by introducing the notion of variable flow. 
In this SST, we say that by reading letter $a$, variable $X$ flows to $Y$ (since the update of variable 
$Y$ is based on variable $X$) while $Y$ flows to $X$.
We extend the notion of transition monoid for SSTs to take both state and variable flow into 
account. 
We define transition matrices $M_s$ indexed by pairs $(p,X)$ where $p$ is a state and $X$
is a variable. 
Since in general, for copy-full SSTs, a variable $X$ might be copied in more than one variable,
it could be that $X$ flows into $Y$ several times.  
Our notion of transition monoid also takes into account, the number of times 
a variable flows into another.  
In particular, $M_s[p,X][q,Y] = i$ means that there exists a run from
$p$ to $q$ on $s$ on which $X$ flows to $Y$ $i$ times. 
Hence the transition monoid of an \sst{} may not be finite. 

\subsection{Main results}
In this paper we introduce a new concept of transition monoid for \sst{}, used to
define the notion of aperiodic \sst{}. 
FO transformations, although weaker than MSO transducers, still enjoy a lot of expressive
power: for instance they can still double, reverse, and swap strings, and are closed under FO look-ahead.
We show that FO string transformations are exactly the transformations definable by \sst{}
whose transition monoid is aperiodic with matrix values ranging over
$\set{0,1}$ (called $1$-bounded transition monoid). 
We also show that checking aperiodicity of an \sst{} is \textsc{PSpace-complete}. 
Simple restrictions on \sst{} transition monoids nicely capture restrictions
on variable updates that has been considered in other works. 
For instance, \emph{bounded copy} of \cite{FiliotTrivediLics12} correspond to finiteness of 
the transition monoid, while \emph{restricted copy} of \cite{AD12}
correspond to its $1$-boundedness. 
Finally, unlike \cite{AC10}, our proof is not based on the intermediate model of two-way 
transducers and is more direct. 
We give a logic-based proof that simplifies that of~\cite{ADT13} by restricting it 
to string-to-string transformations. 

\subsection{Related work}
Diekert and Gastin~\cite{dg08SIWT} presented a detailed survey of several automata, logical, 
and algebraic characterisations of  first-order definable languages.
As mentioned earlier the connection between MSO and transducers have been investigated 
in \cite{AC10,EH01}. 
Connection between two-way transducers and FO-transformations has been mentioned 
in~\cite{Carton13} in an oral communication, where they left the SST connection as an open question. 
First-order transformations are considered in \cite{McKenzieEtAl06}, but
not in the sense of \cite{Cour94}. In particular, they are weaker, as they
cannot double strings or mirror them, and are definable by one-way
(variable-free) finite state transducers. 
Finally, \cite{DBLP:journals/corr/Bojanczyk13} considers first-order
definable transformations \emph{with origin information}. The
semantics is different from ours, because these transformations are
not just mapping from string to strings, but they also connect output
symbols with input symbols from where they originate. 

The first-order definability problem for regular languages is known to be decidable. 
In particular, given a deterministic automaton $A$, deciding whether $A$ defines a first-order language can 
be decided in \textsf{PSpace}. 
Although we make an important and necessary step in answering this question in the context of regular string transformation, 
the decidability remains an open problem. 

\section{Preliminaries}
\label{sec:prelims}
\subsection{Alphabets, Strings, and Languages}
An alphabet $\Sigma$ is a finite set of letters. 
A finite string over $\Sigma$ is defined as a finite sequence  of letters from
$\Sigma$.  
We denote by $\epsilon$ the empty string.
We write $\Sigma^*$ for the set of finite strings over $\Sigma$.  
A (string) language over an alphabet $\Sigma$ is defined as a set of finite 
strings. 

For a string $s \in \Sigma^*$ we write $|s|$ for its length and
$\dom(s)$ for the set $\{ 1,\dots,|s|\}$. For all $i\in \dom(s)$ we write $s[i]$ 
for the $i$-th letter of the string $s$. For any $j\in dom(s)$, the
substring starting at position $i$ and ending at position $j$ is
defined as $\epsilon$ if $j<i$ and by the sequence of letters
$s[i]s[i+1]\dots s[j]$ otherwise.
We write $s[i{:}j]$, $s(i{:}j)$, $s[i{:}j)$, and
$s  (i{:}j]$, to denote substrings of $s$ respectively starting at
$i$ and ending at $j$, starting at $i{+}1$ and ending at $j{-}1$, and so
on. 
For instance, $s[1{:}x)$ denotes the prefix ending at $x-1$ (it
is $\epsilon$ if $x = 1$), while $s(x{:}|s|]$ denotes the suffix starting
at $x+1$.

\subsection{First-order logic for strings}
We represent a string $s \in \Sigma^*$ by the relational structure 
$\struc{s} {=} (\dom(s), \preceq^s, (L^s_a)_{a \in \Sigma})$, called the string
model of $s$, where 
\begin{itemize}
\item 
  $\dom(s) = \set{1, 2, \ldots, |s|}$ is the set of positions in $s$, 
\item 
  $\preceq^s$ is a binary relation over the positions in $s$ characterizing the
  natural order, i.e. $(x, y) \in \preceq^s$ if $x \leq y$;
\item
  $L^s_a$, for all $a \in \Sigma$, are the unary predicates that hold for the
  positions in $s$ labeled with the alphabet $a$, i.e., $L^s_a(i)$ iff 
  $s[i]=a$, for all $i\in \dom(s)$.
\end{itemize}
When it is clear from context we will drop the superscript $s$ from the
relations $\preceq^s$ and $L^s_a$. 

Properties of string models over the alphabet $\Sigma$ can be formalized by
first-order logic denoted by $\FO(\Sigma)$ (or $\FO$ when $\Sigma$ is clear from
the context). 
Formulas of $\FO(\Sigma)$ are built up from variables $x, y, \ldots$ ranging
over positions of string models along with \emph{atomic formulas} of the
form   
$x{=} y, x {\preceq} y$, and  $L_a(x)$  for all $a \in \Sigma$  
where formula $x{=}y$ states that variables $x$ and $y$ points to the same
position, the formula $x \preceq y$ states that position corresponding to
variable $x$ is not bigger than that of $y$, 
and the formula $L_a(x)$ states that position $x$ has the label $a \in
\Sigma$.
Atomic formulas are connected with \emph{propositional connectives} $\neg$,
$\wedge$, $\lor$, $\to$, and \emph{quantifiers} $\forall$ and $\exists$ that
range over node variables.  
We say that a variable is \emph{free} in a formula if it does not occur in the
scope of some quantifier.  
A \emph{sentence} is a formula with no free variables.
We write $\phi(x_1, x_2, \ldots, x_k)$ to denote that  at most the variables
$x_1, \ldots, x_k$ occur free in $\phi$.
  For a string $s\in \Sigma^*$ and for positions $n_1, n_2, \ldots, n_k \in
\dom(s)$ we say that $s$ with valuation $\nu = (n_1, n_2, \ldots, n_k)$
satisfies the formula $\phi(x_1, x_2, \ldots, x_k)$ and we write 
$(s, \nu) \models \phi(x_1, x_2, \ldots, x_k)$ or  $s \models
\phi(n_1, n_2, \ldots, n_k)$ if formula $\phi$ with $n_i$ as the
interpretation of $x_i$ is satisfied in the string model $\struc{s}$. 
We define the following useful FO-shorthands. 
\begin{itemize}
\item
  $x \succ y \rmdef \neg (x \preceq y)$ and $x \prec y \rmdef (x \preceq y) \wedge
  \neg (x = y)$,
\item
  $S(x, y) \rmdef (x \prec y) \wedge \forall z ((z \prec y) \to (z \preceq
  x))$   
\item
  $\lst(x) \rmdef \neg \exists y. S(x, y)$ and $\fst(x) \rmdef \neg \exists
  y. S(y, x)$
\item
  The sentence $\istring$ characterizes valid string models and is
  defined as 
  \[
  \istring \rmdef   \forall x,y,z. (\vee_{a\in\Sigma} L_a(x))\wedge\wedge_{a \not = b \in \Sigma} (L_a(x) \to \neg
  L_b(x)) \wedge (S(x, y) \wedge S(x, z) \to y=z) 
  \wedge (\fst(x) \wedge \fst(y) \to x = y).
  \]
\end{itemize}

  It is easy to see that a structure satisfying $\istring$ property uniquely
  characterizes a string.
The language defined by an FO sentence $\phi$ is $L(\phi) \rmdef \set{s \in \Sigma^* \::\:
  \struc{s} \models \phi}$. 
We say that a language $L$ is FO-definable if there is an FO sentence $\phi$
such that $L = L(\phi)$. 
\begin{example}
  Let $\Sigma = \set{a, b}$. 
  Consider the language $L_1 \subseteq \Sigma^*$ of strings ending with $b$
  definable using the following formula  
  $
  \forall x. (\lst(x)\to L_b(x))$.
  The language $L_2 = \set{(ab)^n \::\: n \geq 0}$ is definable using the
  following FO formula:
  \[ 
  \forall x.(\fst(x) \to L_a(x)) \wedge \forall x.(\lst(x) \to L_b(x)) 
  \wedge \forall y. (L_a(x) {\wedge} S(x, y) {\to} L_b(y)) 
  \wedge  \forall y. (L_b(x) {\wedge} S(x, y) {\to} L_a(y))) 
\]
\end{example}

First-order logic can be used, in an analogous manner, to define languages of
trees and graphs by defining appropriate relational structures~\cite{Tho96}. 
Monadic second-order logic extends first-order logic by permitting
variables to range over sets of positions (monadic second-order variables) and
quantification over such variables. 
We say that a language is MSO-definable if it can be characterized by an MSO
sentence. 
\subsection{Properties of first-order logic}
\label{fo:prop}

The \emph{quantifier rank}, $qr(\phi)$, of an FO-formula $\phi$ is defined as
the maximal number of nested quantifiers in $\phi$, formally:
\begin{eqnarray*}
qr(\phi) = 
\begin{cases}
  0 & \text{ if $\phi$ is atomic} \\
  \max \set{qr(\phi_1), qr(\phi_2)} & \text{ if $\phi = \phi_1 \vee \phi_2$ or
    $\phi_1 \wedge \phi_2$}\\
  qr(\phi_1) & \text{ if $\phi = \neg \phi_1$}\\
  1 + qr(\phi_1) & \text{ if $\phi = \exists x \phi_1$ or $\phi = \forall x \phi_1$}\\
\end{cases}
\end{eqnarray*}

A fundamental property~\cite{Strau94} of first-order logic states
that for a given $k \in \Nat$, there are only finitely many sentences---up to
logical equivalence---of quantifier rank lesser than $k$. 
Based on this property one defines the notion of first-order $k$-type for
strings. 
The first-order $k$-type of a string $s$, denoted by $\ktype{s}{k}$, is the set
of FO-sentences of quantifier rank at most $k$ that are satisfied by $s$. 
Formally,
\[
\ktype{s}{k} = \set{ \phi \::\: \phi\text{ is an FO-sentence s.t. } qr(\phi)\leq
  k\ \wedge\ s \models \phi} 
\]
We write $\Theta_k = \set{\ktype{s}{k}\ |\ s \in \Sigma^* }$ for the
set of $k$-types. Since there are only finitely many sentences of
quantifier rank lesser than $k$, $\Theta_k$ is finite.  

We say that two strings $s, s'\in\Sigma^*$ are $k$-equivalent, denoted by
$s \equiv_k s'$, if they have the same $k$-type,
i.e. $\ktype{s}{k}=\ktype{s'}{k}$. 
In other words, $s$ and $s'$ are $k$-equivalent if they satisfy the same
FO-sentences of quantifier rank at most $k$. 
It is also well-known~\cite{Strau94} that $\equiv_k$ is a congruence relation
of finite index.
\begin{figure}
\begin{center}
\begin{tikzpicture}[->,>=stealth',shorten >=1pt,auto]

\tikzstyle{graphnode}=[circle,fill=black,thick,inner
sep=0pt,minimum size=1.5mm]

\draw [|-(] (0,0) -- node[above] {$s_1$} (1.4, 0) ;

\node (a1) at (1.5,0) {$a$} ;
\node (i1) at (1.5,0.3) {$i_1$} ;

\draw [)-(] (1.6,0) -- node[above] {$s_2$} (6.9, 0) ;

\node (b1) at (7,0) {$b$} ;
\node (i2) at (7,0.3) {$i_2$} ;

\draw [)-|] (7.1,0) -- node[below] {$s_3$} (8, 0) ;

\draw [|-(] (0,-1) -- node[below] {$s'_1$} (2.9, -1) ;

\node (a2) at (3,-1) {$a$} ;
\node (i1p) at (3,-1.4) {$i'_1$} ;

\draw [)-(] (3.1,-1) -- node[below] {$s'_2$} (5.9, -1) ;

\node (b2) at (6,-1) {$b$} ;
\node (i2p) at (6,-1.4) {$i'_2$} ;

\draw [)-|] (6.1,-1) -- node[below] {$s'_3$} (8, -1) ;

\draw [-, dotted] (a1) -- (a2) ;
\draw [-, dotted] (b1) -- (b2) ;

\node  at (1.2, -0.5) {$\equiv_{k+2}$};
\node  at (4, -0.5) {$\equiv_{k+2}$};
\node  at (7.3, -0.5) {$\equiv_{k+2}$};

\end{tikzpicture}
\end{center}
\vspace{-4mm}
\caption{\label{fig:indisformulas} String decomposition of Proposition~\ref{prop:ktype}.2}
\vspace{-4mm}
\end{figure}
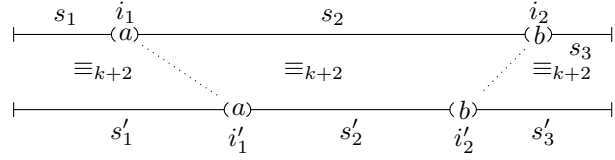

\begin{proposition}[Properties of FO-formulas of bounded
  quantifier-depth~\cite{Strau94}] 
  \label{prop:ktype}
  In this paper we use the following fundamental properties of FO formulas.  
\begin{enumerate}
\item \label{prop:compo}
  For all strings $s_1,s_2,s'_1,s'_2\in\Sigma^*$, if $s_1\equiv_k s'_1$ and
  $s_2\equiv_k s'_2$, then $s_1s_2\equiv_k s'_1s'_2$. 
\item\label{prop:indis}
  For all $k\geq 0$. Let $s_1,s_2,s_3,s'_1,s'_2,s'_3\in\Sigma^*$, and
  $a,b\in\Sigma$, such that $s_i\equiv_{k+2} s'_i$, $i=1,2,3$.  
  Let $i_1 =|s_1|+1$, $i_2=i_1+|s_2|+1$, $i'_1 =|s'_1|+1$, $i'_2=i'_1+|s'_2|+1$
  (see Fig. \ref{fig:indisformulas}).  
  Let $\phi(x,y)$, $\psi(x)$ be two FO formulas of quantifier rank at most
  $k$. We have 
  \begin{eqnarray*}
    s_1as_2\models \psi(i_1) \text{ iff }  s'_1as'_2\models \psi(i'_1) & \text{ and
    } &
    s_1as_2bs_3\models \phi(i_1,i_2) \text{ iff }  s'_1as'_2bs'_3\models \phi(i'_1,i'_2).
  \end{eqnarray*}

\item\label{prop:ape}\cite{Strau94}
  For all $k\geq 0$ and  all $m\geq 2^k$, for all strings $s,s'\in\Sigma^*$, $s^m \equiv_k
  s^{m+1}$, or in other words $\ktype{s^m}{k} = \ktype{s^{m+1}}{k}$. 
\end{enumerate}
\end{proposition}

Thanks to Proposition \ref{prop:ktype}.\ref{prop:compo}, one can extend
the concatenation operator to types: for all
$\tau_1,\tau_2\in\Theta_k$, $\tau_1.\tau_2 = \ktype{s_1.s_2}{k}$ where
$s_1,s_2\in\Sigma^*$ are such that $\tau_i = \ktype{s_i}{k}$,
$i=1,2$. The operator ``$.$'' on $k$-types is called \emph{type composition}. 

The following proposition states that $k$-types can be represented by an FO
sentence of quantifier-depth at most $k$. 
Moreover, the $k$-types of a substring of $s$ between two positions $i_1$ and
$i_2$ such that $i_1<i_2$ can also be characterized by some FO-formula with two
free variables by guarding all quantifications of any variable $z$ in
$\Phi_\tau$ ($\tau$ is a $k$-type) by the predicate $guard(z) = x\preceq z\preceq y$.  

\begin{proposition}[\cite{Strau94}]\label{prop:hintikka}
  Let $\Theta_k$ be the set of all $k$-types. 
\begin{enumerate}
\item
  For all $k$-types $\tau\in\Theta_k$, there exists an FO-sentence $\Phi_\tau$
  of quantifier rank at most $k$, such that for all strings $s \in\Sigma^*$,
  $s \models \Phi_\tau$ iff $\ktype{s}{k} = \tau$.  
\item 
  For all $k$-types $\tau\in\Theta_k$, there exists an FO-formula
  $\Psi_\tau(x,y)$ of quantifier rank at most $k$ such that for all strings
  $s \in\Sigma^*$ and all positions $i_1<i_2$ of $s$, $s\models
  \Psi_\tau(i_1,i_2)$ iff $\ktype{s[i_1:i_2]}{k} = \tau$.  
\end{enumerate}
\end{proposition}

 \subsection{Aperiodic finite automata}
A finite automaton is a tuple $\Aa = (Q, q_0, \Sigma, \delta, F)$ where $Q$ is
a finite set of states, $q_0 \in Q$ is the initial state, $\Sigma$ is an input
alphabet, $\delta: Q \times \Sigma \to Q$ is a transition function, and
$F \subseteq Q$ is the set of accepting states.
For states $q, q' \in Q$ and letter $a \in \Sigma$ we say that $(q, a, q')$ is a
transition of the automaton $\Aa$ if $\delta(q,a ) = q'$ and we write $q
\xrightarrow{a} q'$.  
A run of $\Aa$ over a finite string $s = a_1 a_2 \ldots a_n \in \Sigma^*$  is a
finite sequence of transitions $\seq{(q_0, a_1, q_1), (q_1, a_2, q_2), \ldots,
  (q_{n-1}, a_n, q_n)} \in (Q \times \Sigma \times Q)^*$ starting from the initial state
$q_0$ and we represent such runs as $q_0 \xrightarrow{a_1} q_1 \xrightarrow{a_2}
q_2 \cdots q_n$; also, in this case we say that there is a run of $\Aa$ from $q_0$
to $q_n$ over the string $s$ and we write $q_0 \rightsquigarrow^s_\Aa q_n$ (or
$q_0 \rightsquigarrow^s q_n$ if the automaton is clear from the context).
A string $s$ is accepted by a finite automaton $\Aa$ if there exists $q_n \in F$
such that $ q_0 \rightsquigarrow^{s} q_n$. 
The language defined by a finite automaton $\Aa$ is $L(\Aa) = \set{s \::\:  q_0
  \rightsquigarrow^{s} q_n \text{ and } q_n \in F}$. 

B\"uchi-Elgot-Trakhtenbrot~\cite{Bu60,Elg61,Tra62} first established the
connection between mathematical logic and automata theory  by showing that the
deterministic finite state automata accept the same class  of languages as
monadic second order logic (MSO) interpreted over finite strings.
This class of languages is also known as \emph{regular languages}.
\begin{theorem}[\cite{Bu60,Elg61,Tra62}]
  \label{thm:mso-finite}
  A language $L\subseteq \Sigma^*$ is \ensuremath{\MSO}-definable iff it is
  accepted by some finite automaton.
\end{theorem}

To define a similar automata connection for FO-definable languages, we need to
introduce the concept of aperiodic finite automata. 
Recall that a monoid is an algebraic structure $(M, \cdot, e)$ with a non-empty set
$M$, a binary  operation $\cdot$, and an identity element $e \in M$ such that 
for all $x, y, z \in M$ we have that  $(x \cdot (y \cdot z)) {=} ((x \cdot y)
\cdot z)$, and $x \cdot e = e  \cdot x$ for all $x \in M$.
We say that a  monoid $(M, \cdot, e)$ is \emph{finite} if the set $M$ is finite.
We say that a monoid $(M, ., e)$ is \emph{aperiodic}~\cite{Strau94} if there
exists $n \in \Nat$ such that for all $x \in M$, $x^n = x^{n+1}$. 
Note that for finite monoids, it is equivalent to require that for all $x\in M$,
there exists $n\in \Nat$ such that $x^n= x^{n+1}$.

\begin{example}[Monoids]
  The following three monoids are useful for the development of the results
  presented in the paper.
\begin{itemize}
\item {\bf Free Monoid}. 
  The set of all strings over $\Sigma$ forms a monoid, with string concatenation
  as the operation and the empty string $\epsilon$ as the identity element. 
  This monoid is denoted as $(\Sigma^*, ., \epsilon)$ and known as the free
  monoid.  
\item {\bf $k$-type Monoid}.
  The set of $k$-types form a finite monoid $(\Theta_k, ., \ktype{\epsilon}{k})$
  with type composition as the  operation and the $k$-type of the empty string
  $\ktype{\epsilon}{k}$ as the identity element. 
  For instance, a direct consequence of Proposition~\ref{prop:ktype}.(3) is
  aperiodicity of the  monoid $(\Theta_k, ., \ktype{\epsilon}{k})$. 
\item {\bf Transition Monoid} The set of transition matrices of a finite automaton $\Aa = (Q, q_0, \Sigma,
\delta, F)$ forms a finite monoid with matrix multiplication as the operation and the
unit matrix $\mathbf{1}$ as the identity element.
This monoid is denoted as $\Mm_\Aa = (M_\Aa, \times, \mathbf{1})$ and known as
transition monoid of $\Aa$.  
Formally, the set $M_\Aa$ is the set of $|Q|$-square Boolean matrices
$M_A = \set{ M_s \::\: s \in \Sigma^*}$ where for all strings $s \in
\Sigma^*$, we have that $M_s[p][q] = 1$ iff $p \rightsquigarrow^s q$.  

\end{itemize}
\end{example}

We say that a finite automaton is aperiodic if its transition monoid is
aperiodic.
The following is a key theorem characterizing  FO-definable languages using
automata. 
\begin{theorem}\cite{Strau94}
  \label{thm:fo-aperiodic}
  A language $L\subseteq \Sigma^*$ is \ensuremath{\FO}-definable iff it is
  accepted by some aperiodic finite automaton.
\end{theorem}

Combining Proposition~\ref{prop:hintikka} and Theorem~\ref{thm:fo-aperiodic} it
follows that for every $k$-type $\tau\in \Theta_k$ there is an aperiodic finite
automaton $A_\tau$ that accepts all strings $s$ with $\ktype{s}{k} = \tau$. 
Such automaton $A_\tau$ is defined as the tuple 
$
A_\tau = (Q_\tau = \Theta_k, \ktype{\epsilon}{k}, \Sigma, \delta_\tau, F_\tau =
\{\tau\}),
$
where $\delta_\tau(\tau',a) = \tau'.\ktype{a}{k}$ for all $\tau'\in
\Theta_k$ and $a\in\Sigma$. By definition of $A_\tau$, for all
$k$-types $\tau_1,\tau_2\in\Theta_k$ and all strings $s\in\Sigma^*$, 
$\tau_1\rightsquigarrow^{s}\tau_2$ iff $\tau_1.\ktype{s}{k} =
\tau_2$. Therefore as direct consequence of
Proposition~\ref{prop:ktype}.\ref{prop:ape}, there exists $m\geq 2^k$ such
that $\tau_1\rightsquigarrow^{s^m}\tau_2$ iff
$\tau_1\rightsquigarrow^{s^{m+1}}\tau_2$. In other words, the
transition monoid of $A_\tau$ is aperiodic, and so is $A_\tau$.

\section{Aperiodic String Transducers}
\label{sec:definitions}
For sets $A$ and $B$, we write $[A \to B]$ for the set of functions 
$F: A \to B$, and $[A \rightharpoonup B]$ for the set of partial functions $F: 
A \rightharpoonup B$. 
A string-to-string transformation from an input alphabet $\Sigma$ to an output
alphabet $\Gamma$ is a partial function in ${[\Sigma^* \rightharpoonup
  \Gamma^*]}$.   
We have seen some examples of string-to-string transformations in the
introduction.
For the examples of first-order definable transformations we use the following
representative example. 
\begin{example}
  \label{example0}
  Let $\Sigma {=} \set{a,b}$. 
  For all strings $s\in\Sigma^*$, we denote by $\overline{s}$ its
  mirror image, and for all $\sigma\in\Sigma$, by $s\backslash
  \sigma$ the string obtained by removing all symbols $\sigma$ from
  $s$. 
  The transformation $f_1 : \Sigma^* \rightharpoonup \Sigma^*$
  maps any string $s\in\Sigma^*$ to the output string
  $(s\backslash b)\overline{s} (s\backslash a)$. 
  For example, $f_1(abaa) = aaa.aaba.b$.
\end{example}
\subsection{First-order logic definable Transformations} 
\label{subsec:fotrans}
Courcelle~\cite{Cour94} initiated the study of 
structure transformations using monadic second-order logic. 
In this paper, we restrict this logic-based transformation model to
FO-definable string transformations. 
The main idea of Courcelle's transformations is to define a transformation  
$(w, w') \in R$  by defining the string model of $w'$ using a finite number of
copies of positions of the string model of $w$. 
The existence of positions, various edges, and position labels are then given as
$\FO(\Sigma)$ formulas. 
\begin{definition}[First-order Transducers]
 An \emph{FO string transducer}  is a tuple 
$T {=} (\Sigma, \Gamma, \phi_{\dom}, C, \phi_{\nodes}, \phi_{\preceq})$ where:
\begin{itemize}
\item
  $\Sigma$ and $\Gamma$ are finite sets of input and output alphabets;
\item 
  $\phi_\dom$ is a closed $\FO(\Sigma)$ formula characterizing the domain of
  the transformation;   
\item 
  $C {=} \set{1, 2, \ldots, n}$ is a finite index set;  
\item 
  $\phi_{\nodes} {=} \set{ \phi^c_\gamma(x) : c \in C \text{ and } \gamma \in
    \Gamma}$ is a finite set of $\FO(\Sigma)$ formulas with a free position
  variable $x$; 
\item 
  $\phi_{\preceq} {=} \set{\phi^{c, d}_\preceq(x, y) : c, d \in C}$ is a finite
  set of
  $\FO(\Sigma)$ formulas with two free position variables $x$ and $y$.   
\end{itemize}
\end{definition}

The transformation $\inter{T}$ defined by $T$ is as follows.  
A string $s$ with $\struc{s} = (\dom(s), \preceq, (L_a)_{a \in \Sigma})$ is in
the domain of $\inter{T}$ if $s \models \phi_{\dom}$ and the output is the
relational structure $M = (D, \preceq^M, (L^M_\gamma)_{\gamma\in\Gamma})$ such
that 
\begin{itemize}
\item 
  $D = \set{ v^c \::\: c \in \dom(s), c \in C \text{ and } \phi^c(v)}$ is the
  set of positions where $\phi^c(v) \rmdef
  \lor_{\gamma \in \Gamma} \phi^c_\gamma(v)$; 
\item
  $\preceq^M {\subseteq} D {\times} D$ is the ordering relation between
  positions and it is such that for $v, u \in  dom(s)$ and
  $c, d \in C$ we have that $v^c \preceq^M u^d$ 
  if $w \models \phi^{c, d}_\preceq(v, u)$; and
\item 
  for all $v^c \in D$ we have that $L_\gamma^M(v^c)$ iff $\phi^c_\gamma(v)$.    
\end{itemize}

Observe that the output is unique and therefore FO transducers implement
functions.  
However, note that the output structure may not always be a string.
We say that an FO transducer is a \emph{string-to-string} transducer if its
domain is restricted to string graphs and the output is also a string graph. 
We say that a string-to-string transformation is FO-definable if there exists an FO
string-to-string transducer implementing the transformation. 
We write $\fot{}$ for the set of FO-definable string-to-string transformations. 
  \begin{figure}
\begin{center}
\begin{tikzpicture}[->,>=stealth',shorten >=1pt,auto]

\tikzstyle{alivenode}=[circle,fill=black!80,thick,inner sep=1pt,minimum size=4mm]

\tikzstyle{deadnode}=[circle,fill=black!10,thick,inner sep=0pt,minimum size=4mm]

\node at (-1,1.5)  {{\footnotesize pos.}};
\node at (-1,1) {{\footnotesize input}};
\node at (-1,0) {copy $1$} ;
\node at (-1,-1) {copy $2$} ;
\node at (-1,-2) {copy $3$} ;

\node [alivenode] (a1) at (0,1) {{\color{white} $a$}} ;
\node [alivenode](a2) at (2,1) {{\color{white} $b$}} ;
\node [alivenode](a3) at (4,1) {{\color{white} $a$}} ;
\node [alivenode](a4) at (6,1) {{\color{white} $a$}} ;
\node [alivenode](a5) at (8,1) {{\color{white} $b$}} ;
\node [alivenode](a6) at (10,1) {{\color{white} $b$}} ;
\node [alivenode](a7) at (12,1) {{\color{white} $b$}} ;
\node [alivenode](a8) at (14,1) {{\color{white} $b$}} ;
\node [alivenode](a9) at (16,1) {{\color{white} $a$}} ;

\node (a1) at (0,1.5) {\begin{scriptsize}$1$\end{scriptsize}} ;
\node (a2) at (2,1.5) {\begin{scriptsize}$2$\end{scriptsize}} ;
\node (a3) at (4,1.5) {\begin{scriptsize}$3$\end{scriptsize}} ;
\node (a4) at (6,1.5) {\begin{scriptsize}$4$\end{scriptsize}} ;
\node (a5) at (8,1.5) {\begin{scriptsize}$5$\end{scriptsize}} ;
\node (a6) at (10,1.5) {\begin{scriptsize}$6$\end{scriptsize}} ;
\node (a7) at (12,1.5 ) {\begin{scriptsize}$7$\end{scriptsize}} ;
\node (a8) at (14,1.5 ) {\begin{scriptsize}$8$\end{scriptsize}} ;
\node (a9) at (16,1.5 ) {\begin{scriptsize}$9$\end{scriptsize}} ;

\node [alivenode] (x1) at (0,0) {{\color{white} $a$}} ;
\node [deadnode] (x2) at (2,0) {} ;
\node [alivenode] (x3) at (4,0) {{\color{white} $a$}} ;
\node [alivenode] (x4) at (6,0) {{\color{white} $a$}} ;
\node [deadnode] (x5) at (8,0) {} ;
\node [deadnode] (x6) at (10,0) {} ;
\node [deadnode] (x7) at (12,0) {} ;
\node [deadnode] (x8) at (14,0) {} ;
\node [alivenode] (x9) at (16,0) {{\color{white} $a$}} ;

\node [alivenode] (y1) at (0,-1) {{\color{white} $a$}} ;
\node [alivenode] (y2) at (2,-1) {{\color{white} $b$}} ;
\node [alivenode] (y3) at (4,-1) {{\color{white} $a$}} ;
\node [alivenode] (y4) at (6,-1) {{\color{white} $a$}} ;
\node [alivenode] (y5) at (8,-1) {{\color{white} $b$}} ;
\node [alivenode] (y6) at (10,-1) {{\color{white} $b$}} ;
\node [alivenode] (y7) at (12,-1) {{\color{white} $b$}} ;
\node [alivenode] (y8) at (14,-1) {{\color{white} $b$}} ;
\node [alivenode] (y9) at (16,-1) {{\color{white} $a$}} ;

\node [deadnode] (z1) at (0,-2) {} ;
\node [alivenode] (z2) at (2,-2) {{\color{white} $b$}} ;
\node [deadnode] (z3) at (4,-2) {} ;
\node [deadnode] (z4) at (6,-2) {} ;
\node [alivenode] (z5) at (8,-2) {{\color{white} $b$}} ;
\node [alivenode] (z6) at (10,-2) {{\color{white} $b$}} ;
\node [alivenode] (z7) at (12,-2) {{\color{white} $b$}} ;
\node [alivenode] (z8) at (14,-2) {{\color{white} $b$}} ;
\node [deadnode] (z9) at (16,-2) {} ;

  \draw[->]   (x1) -- node[below] {$\phi_\preceq^{1,1}$} (x3) ;
  \draw[->]   (x3) -- (x4) ;
  \draw[->]   (x4) -- (x9) ;

  \draw[->]   (x9) -- node[left] {$\phi_\preceq^{1,2}$} (y9) ;
  \draw[->]   (y9) -- (y8) ;
  \draw[->]   (y8) -- (y7) ;
  \draw[->]   (y7) -- node[below] {$\phi_\preceq^{2,2}$} (y6) ;
  \draw[->]   (y6) -- (y5) ;
  \draw[->]   (y5) -- (y4) ;
  \draw[->]   (y4) -- (y3) ;
  \draw[->]   (y3) -- (y2) ;
  \draw[->]   (y2) -- (y1) ;

  \draw[->]   (y1) -- node[left] {$\phi_\preceq^{2,3}\ \ $} (z2) ;
  \draw[->]   (z2) -- node[above] {$\phi_\preceq^{3,3}$} (z5) ;
  \draw[->]   (z5) -- (z6) ;
  \draw[->]   (z6) -- (z7) ;
  \draw[->]   (z7) -- (z8) ;


\fill[blue!20,rounded corners, fill opacity=0.2] (-0.3,0.4) rectangle (16.3,-2.3);

\end{tikzpicture}
\end{center}
\caption{First-Order Transduction $w\mapsto
  (w\backslash b)\overline{w}(w\backslash a)$}
\label{fig:foexample} 
\end{figure}

\begin{example}
    The best way, perhaps,  to explain an FO transducers is via an example. 
  Consider the transformation $f_1$ of Example~\ref{example0}.
  It can be defined using an FO transducer that uses three copies of the 
  input domain, as illustrated on Fig.~\ref{fig:foexample}. 
  The domain formula is $\phi_{\dom} = \istring$. 
  Intuitively, the first copy corresponds to $(w\backslash b)$, therefore the
  label formula $\phi_\gamma^1(x)$ is defined by $\text{false}$ if $\gamma = b$
  in order to filter out the input positions labelled $b$,
  and by $\text{true}$ otherwise. 
  For second copy corresponds to $\overline{w}$, hence all positions of the
  input are kept and their labels preserved (however the edge direction will be
  complemented) therefore the label formula is $\phi_\gamma^2(x)= L_\gamma(x)$. 
  Finally, the third copy corresponds to $(w \backslash a)$ and hence
  $\phi_\gamma^3(x)$ is  $\text{true}$ if $\gamma=b$ and $\text{false}$
  otherwise. 
  The transitive closure of the output successor relation is defined by: \\
    $\phi_\preceq^{1,1}(x,y)  =  x\preceq y,~
    \phi_\preceq^{2,2}(x,y)  =  y\preceq x,
    \phi_\preceq^{3,3}(x,y)  =  x\preceq y,  \\
    \phi_\preceq^{c,c'}(x,y)  =  \text{true} \text{ if } c<c', 
    \phi_\preceq^{c,c'}(x,y)  =  \text{false} \text{ if } c'<c.$\\ 
  Note that the transitive closure is not depicted on the figure,
  but only the successor relation. 
  Using first-order logic we define the position successor relation the
  following way: for all copies $c,d$, the existence of a direct edge from a
  position $x^c$ to a position $y^d$ of the output, also called the 
  successor relation $S(x^c, y^d)$, is defined by the formula 
  $  \phi_\suc^{c,d}(x,y) \rmdef \phi_\prec^{c,d}(x,y)\wedge 
  \neg\exists z. \bigvee_{e\in C}
  \phi_\prec^{c,e}(x,z)\wedge \phi_\prec^{e,d}(z,y)
  $
  where $\phi_\prec^{c_1,c_2}(x_1,x_2) \rmdef
  \phi_\preceq^{c_1,c_2}(x_1,x_2)\wedge x_1\neq x_2$ for all 
  $c_1,c_2\in C$. 
\end{example}

We define the \emph{quantifier rank} $qr(T)$ of an FOT $T$ as the maximal
quantifier rank of any formula in $T$, plus $1$. 
We add $1$ for technical reasons, mainly because defining the successor relation
requires one quantifier.  

\subsection{Streaming String Transducers}\label{subsec:sst}
Streaming string transducers~\cite{AC10,AC11} (SSTs) are one-way finite-state
transducers that manipulates a finite set of string variables to compute its
output.
Instead of appending symbols to the output tape, SSTs concurrently update all
string variables using a concatenation of string variables and output symbols.  
The transformation of a string is then defined using an output (partial)
function $F$ that associates states with a concatenation of string
variables, s.t.\ if the state $q$ is reached after reading the string and
$F(q) {=} xy$, then the output string is the final valuation of $x$ concatenated 
with that of $y$. 
In this section we formally introduce SSTs and introduce restrictions on SSTs
that capture FO-definable transformations. 

Let $\varsst$ be a finite set of variables and $\Gamma$ be a finite alphabet. 
A substitution $\sigma$ is defined as a mapping ${\sigma : \varsst \to (\Gamma \cup
  \varsst)^*}$. 
A valuation is defined as a substitution $\sigma: \varsst \to \Gamma^*$.
Let $\Ss_{\varsst, \Gamma}$ be the set of all substitutions $[\varsst \to (\Gamma \cup
\varsst)^*]$.
Any substitution $\sigma$ can be extended to $\hat{\sigma}: (\Gamma \cup \varsst)^*
\to (\Gamma \cup \varsst)^*$ in a straightforward manner.
The composition $\sigma_1 \sigma_2$  of two substitutions $\sigma_1$ and $\sigma_2$
is defined as the standard function composition  $\hat{\sigma_1} \sigma_2$,
i.e. $\hat{\sigma_1}\sigma_2(X) = \hat{\sigma_1}(\sigma_2(X))$ for all $X \in \varsst$. 
We are now in a position to introduce streaming string transducers. 
\begin{definition}
  \label{def:sst}
  A deterministic streaming string transducer (\sst) is a tuple 
  $T = (\Sigma, \Gamma, Q, q_0, Q_f, \delta, \varsst, \rho, F)$ where:
  \begin{itemize}
  \item
    $\Sigma$ and $\Gamma$ are finite sets of input and output alphabets;
  \item 
    $Q$ is a finite set of states with initial state $q_0$;
  \item 
    $\delta : Q \times \Sigma \to Q$ is a transition function;
  \item 
    $\varsst$ is a finite set of variables;
  \item 
    $\rho : \delta \to \Ss_{\varsst, \Gamma}$ is a variable update
    function;
  \item 
    $Q_f$ is a subset of final states;
  \item 
    $F: Q_f \rightharpoonup \varsst^*$ is an output  function.
  \end{itemize}
\end{definition}
The concept of a run of an \sst{} is defined in an analogous manner to that of
a finite state automaton.
The sequence $\seq{\sigma_{r, i}}_{0 \leq i \leq |r|}$ of substitutions induced
by a run $r = q_0 \xrightarrow{a_1} q_1 \xrightarrow{a_2} q_2 \ldots q_{n-1}\xrightarrow{a_n} q_n$ is defined
inductively as the following: $\sigma_{r, i} {=} \sigma_{r, i{-}1} \rho(q_{i-1}, a_{i})$
for $1 < i \leq |r|$ and $\sigma_{r,1} = \rho(q_0,a_1)$. We denote $\sigma_{r,|r|}$ by $\sigma_r$.

If the run $r$ is final, i.e. $q_n\in Q_f$, we can extend the output function $F$ to the run $r$ by 
$F(r) = \sigma_\epsilon\sigma_{r}F(q_n)$, where $\sigma_\epsilon$ substitute all variables by their initial value $\epsilon$.
For all strings $s \in\Sigma^*$, the output of $s$ by $T$ is defined only if
there exists an accepting run $r$ of $T$ on $s$, and in that case the output
is denoted by $T(s) = F(r)$. 
The transformation $\inter{T}$ defined by an SST $T$ is the function
$\set{(s, T(s)) \::\: T(s) \text{ is defined}}$.

\begin{example}
Let us consider the streaming string transducer $T_2$ shown in
Figure~\ref{fig:sst-ex2} implementing the transformation $f_1$ introduced in
Example~\ref{example0}. 
The SST $T_2$ has only one state $q_0$, and three variables $X, Y$, and $Z$. 
The variable update is shown in the figure and the output function is
s.t. $F(q_0) = XYZ$.  
\begin{figure}[t]
  \begin{center}
    \begin{tikzpicture}[->,>=stealth',shorten >=1pt,auto,node distance=1cm,
      semithick,scale=0.9,every node/.style={scale=0.9}]
      \tikzstyle{every state}=[fill=blue!20!white,minimum size=3em]
      \node[initial, accepting, initial text={},state,fill=blue!20!white] at (5,0) (A) {$q_0$} ;
      \path(A) edge [loop right] node {$a \mid (X,Y,Z) := (Xa, aY, Z)$} (A); 
      \path(A) edge [loop below] node {$b \mid (X,Y,Z) := (X, bY, Zb)$} (A);
    \end{tikzpicture}
  \end{center}
  \caption{SST implementing the transformation $s \mapsto (s \backslash
    b)\overline{s}(s \backslash a)$.  Here the output function is $F(1)=XYZ.$}
  \label{fig:sst-ex2}
\end{figure}

The following table shows a run of $T_2$ on the string $s=abaa$. 
\[
\begin{array}{r|cccccccccccccccccccccccccccccccccccc}
  & & \!\!\!\!a\!\!\!\! & & \!\!\!\!b\!\!\!\! & & \!\!\!\!\!a\!\!\!\! & & \!\!\!\!a\!\!\!\! \\
  \hline 
  X & \varepsilon &  & a & & a & & aa & & aaa \\
  Y & \varepsilon & & a & & ba & & aba & & aaba \\
  Z & \varepsilon & & \varepsilon & & b & & b & & b \\
\end{array}
\]
Let $r$ be the run of $T_2$ on $s = abaa$. We have 
$\sigma_{r,1} : (X,Y,Z)\mapsto (Xa,aY,Z)$, $\sigma_{r,2} :
(X,Y,Z)\mapsto \sigma_{r,1}(X,bY,Zb) = (Xa,baY,Zb)$, 
$\sigma_{r,3} : (X,Y,Z)\mapsto \sigma_{r,2}(Xa,aY,Z) = (Xaa,abaY,Zb)$
and $\sigma_{r,4} : (X,Y,Z)\mapsto \sigma_{r,3}(Xa,aY,Z) =
(Xaaa,aabaY,Zb)$. Therefore $T(s) = F(r) = \sigma_\epsilon\sigma_{r,4}F(q_0) = \sigma_\epsilon\sigma_{r,4}(XYZ) =
\sigma_{\epsilon}(XaaaaabaYZb) = aaaaabab$.
\end{example}

\subsection{Transition Monoid of Streaming String Transducers and Aperiodicity}
\label{sec:ap}
We define the notion of aperiodic \sst{}s by introducing an appropriate notion
of transition monoid for transducers. 
The transition monoid of an \sst{} $T$ is based on the effect of a string $s$ on
the states and variables. 
The effect on variables is characterized by, what we call, flow information that
is given as  a relation that describes the number of copies of the
content of a given variable that contribute to another variable after reading a
string $s$.

\vspace{3mm}
\textbf{State and Variable Flow} Let $T = (\Sigma, \Gamma, Q, q_0, Q_f, \delta, \varsst, \rho, F)$ be an \sst{}.
Let $s$ be a string in $\Sigma^*$ and suppose that there exists a run $r$ of $T$
on $s$.
Recall that this run induces a substitution $\sigma_r$ that maps each variable
$X \in \varsst$ to a string $u \in (\Gamma \cup \varsst)^*$.
For string variables $X,Y\in \varsst$, states $p,q\in Q$, and $n \in
\Nat$ we say that  $n$ copies of $Y$ flow to $X$ from $p$ to $q$
if there exists a run $r$ on $s$ from $p$ to $q$, and  $Y$ occurs $n$ times in $\sigma_r(X)$. 
We  denote the flow with respect to a string $s$ as
$(p,Y)\rightsquigarrow^s_n (q,X)$.

\begin{example}
  Consider the run $r$ from $q_0$ to $q_0 $ over the string $aaaa$ in the
  following SST.
  While drawing an SST we often omit the update corresponding to the variables
  that retain their previous value. 
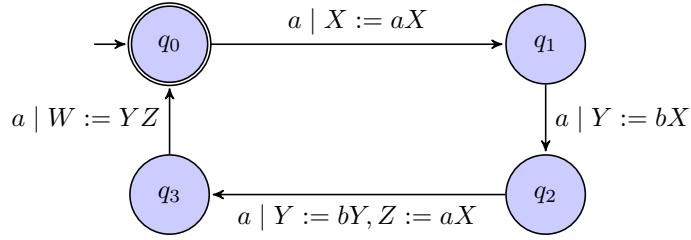
\begin{center}
\begin{tikzpicture}[->,>=stealth',shorten >=1pt,auto,node distance=1cm,
 semithick]
  \tikzstyle{every state}=[fill=blue!20!white,minimum size=3em]
  \node[initial, accepting, initial text={},state,fill=blue!20!white] at (0,0) (A) {$q_0$} ;
  \node[state, fill=blue!20!white] at (5,0) (B) {$q_1$} ;
  \node[state, fill=blue!20!white] at (5,-2) (C) {$q_2$} ;
  \node[state, fill=blue!20!white] at (0,-2) (D) {$q_3$} ;
  \path(A) edge  node {$a \mid X := aX$} (B);
  \path(B) edge node  {$a \mid Y := bX$} (C);
  \path(C) edge node  {$a \mid Y := bY,Z:=aX$} (D);
  \path(D) edge node  {$a \mid W :=YZ$} (A);
\end{tikzpicture}
\end{center}

On the run $r$ on $aaaa$ can be seen that
$\sigma_{r,4}(W)=\sigma_{r,3}[W:=YZ]=\sigma_{r,3}(Y)\sigma_{r,3}(Z)$. 
However, $\sigma_{r,3}(Y)=b\sigma_{r,2}(Y)=b.b.\sigma_{r,1}(X)$ and
$\sigma_{r,3}(Z)=a.\sigma_{r,2}(X)=a.\sigma_{r,1}(X)$, and
$\sigma_{r,1}(X)=a$. 
Thus, on the run from $q_0$ to $q_0$ we have that 
$(q_0,Y) \rightsquigarrow^{aaaa}_1 (q_0,W)$, 
$(q_0,Z) \rightsquigarrow^{aaaa}_1 (q_0,W)$,
$(q_0,X) \rightsquigarrow^{aaaa}_2 (q_0,W)$.
\end{example}

\textbf{Transition Monoid of an SST} In order to define the transition
monoid of an SST $T$, we first extend $\mathbb{N}$ with an extra element
$\perp$, and let $\mathbb{N}_\perp = \mathbb{N}\cup \{\perp\}$. This
new element behaves as $0$: for all $i\in\mathbb{N}_\perp$, $i{.}\bot =
\bot.i = \bot$, $i+\bot = \bot+i=i$. Moreover, we assume that
$\bot < n$ for all $n\in\mathbb{N}$. 
We assume that pairs $(p,X)\in Q\times \varsst$ are
totally ordered. The \emph{transition monoid} of $T$ is the set of square
matrices over $\mathbb{N}_\perp$ indexed (in order) by elements of $Q\times
\varsst$, defined by $M_T = \{ M_s\ |\ s\in \Sigma^*\}$ where
for all strings $s\in \Sigma^*$, $M_s[p,Y][q,X] = n\in\mathbb{N}$ iff $(p,Y)\rightsquigarrow^s_n (q,X)$, and 
 $M_s[p,Y][q,X] = \bot$ iff there is no run from $p$ to $q$ on $s$. 
Note that, by definition,  there is atmost one run
$r$ from $(p,Y)$ to $(q,X)$ on any string $s$.

It is 
easy to see that $(M_T,\times,\textbf{1})$ is a monoid, where $\times$ is defined as matrix
multiplication and the identity element is the unit matrix
$\textbf{1}$. The mapping $M_{\bullet}$, which maps any string $s$ to its
transition matrix $M_s$,  is a morphism from $(\Sigma^*, ., \epsilon)$
to $(M_T, \times, \textbf{1})$. 
We say that the transition monoid $M_T$ of an \sst{} $T$ is 
$n$-bounded if all the coefficients of the matrices of $M_T$ are bounded by $n$.  
Clearly, any $n$-bounded transition monoid is finite.

In \cite{AC11}, SST are required to have \emph{copyless} updates,
i.e., variable updates are defined by linear substitutions. In other
words, the content of a variable can never flow into two different
variables, and cannot flow more than once into another variable. In
\cite{AD12}, this condition was slightly relaxed to the notion of
\emph{restricted copy}. This requirement imposes that a variable
cannot flow more than once into another variable. This allows for a
limited form of copy: for instance, $X$ can flow to $Y$ and $Z$, but 
$Y$ and $Z$ cannot flow to the same variable. Finally,
\emph{bounded copy} SSTs were introduced in
\cite{FiliotTrivediLics12} as a restriction on the variable dependency
graphs. This restriction requires that there exists a bound $K$ such
that any variable flows at most $K$ times in another variable. These
three restrictions were shown to be equivalent, in the sense that SSTs 
with copyless, restricted copy, and bounded copy updates have the same
expressive power.  Given our definition of transition monoid, and the results of Alur, Filiot, and
Trivedi~\cite{FiliotTrivediLics12}, the following result is immediate by
observing that bounded copy restriction of~\cite{FiliotTrivediLics12} for
\sst{}s corresponds to finiteness of transition monoid. 
Also, notice that since the bounded copy assumption generalizes the
copyless~\cite{AC11} and restricted copy~\cite{AD12} assumptions, previous
definitions in the literature of streaming string transducers also correspond to finite transition
monoids.  

\begin{theorem}[\cite{FiliotTrivediLics12}]
  A string transformation is MSO-definable iff it is definable by an
  \sst{} with finite transition monoid. 
\end{theorem}
The main goal of this paper is to present a similar result for FO-definable
transformations.
\begin{definition}[Aperiodic \sst{}s]
  A streaming string transducer is aperiodic if its transition monoid is
  aperiodic. 
\end{definition}
\begin{definition}[1-bounded \sst{}s]
  A streaming string transducer is 1-bounded  if its transition monoid is
  1-bounded. That is, for all strings $s$, and all pairs $(p,Y)$, $(q,X)$, 
  $M_s[p,Y][q,X] \in \{\bot,0,1\}$. 
\end{definition}

\begin{example}(Aperiodic and non-aperiodic SSTs)
Let us consider the transformation $f_{\halve}$ defined as $a^n \mapsto a^{\lceil \frac{n}{2}\rceil}$. 
Consider the SSTs  $T_1$ with $2$ states and $1$ variable,  and 
$T_0$ (its output is $F(1) = X$) both implementing $f_\halve$. 
 
\begin{center}
\begin{tikzpicture}[->,>=stealth',shorten >=1pt,auto,node distance=1cm,
 semithick]
 \tikzstyle{every state}=[fill=purple!20!white,minimum size=2em]
 \node[initial, accepting, initial text={},state,fill=blue!20!white] at (5,0) (A) {1} ;
 \path(A) edge [loop right] node {$a \mid (X,Y) := (aY,X)$} (A); 
 \node[draw=none] at (3,0) {$T_0:$};
\end{tikzpicture}
\end{center}

\begin{center}
\begin{tikzpicture}[->,>=stealth',shorten >=1pt,auto,node distance=1cm,
 semithick]
  \tikzstyle{every state}=[fill=blue!20!white,minimum size=2em]
 \node[initial, accepting, initial text={},state,fill=blue!20!white] at (0,0) (A) {1} ;
 \node[state,accepting, fill=blue!20!white] at (5,0) (B) {2} ;
 \path(A) edge  [bend right=10] node[anchor=south,above]{$a \mid X := aX$} (B);
  \path(B) edge [bend right=10] node[anchor=north,above]  {$a \mid X := X$} (A);
 \node[draw=none] at (-2,0) {$T_1:$};
   \end{tikzpicture}
\end{center}
   It can be seen that the transition monoids 
of both SSTs are $1$-bounded but non aperiodic. In the first case this is caused by the variable flow, while in the second, this is caused by 
the transitions between states. The transition monoid of $T_0$ is a 2 $\times$ 2 matrix. For $k \geq 0$, 
\[
  M_{a^{2k+1}}=\begin{blockarray}{ccc}
    & \matindex{(1,X)} & \matindex{(1,Y)} & \\
    \begin{block}{c(cc)}
      \matindex{(1,X)} & 0 & 1   \\
      \matindex{(1,Y)} & 1 & 0  \\
          \end{block}
  \end{blockarray},~ 
     M_{a^{2k}}=\begin{blockarray}{ccc}
    & \matindex{(1,X)} & \matindex{(1,Y)} & \\
    \begin{block}{c(cc)}
      \matindex{(1,X)} & 1 & 0   \\
      \matindex{(1,Y)} & 0 & 1  \\
          \end{block}
  \end{blockarray}
\]

The transition monoid of $T_1$ is a $2 \times 2$ matrix. 
For $k \geq 0$,
\[
  M_{a^{2k+1}}=\begin{blockarray}{ccc}
    & \matindex{(1,X)} & \matindex{(2,X)} & \\
    \begin{block}{c(cc)}
      \matindex{(1,X)} & 0 & 1   \\
      \matindex{(2,X)} & 1 & 0  \\
          \end{block}
  \end{blockarray},~ 
     M_{a^{2k}}=\begin{blockarray}{ccc}
    & \matindex{(1,X)} & \matindex{(2,X)} & \\
    \begin{block}{c(cc)}
      \matindex{(1,X)} & 1 & 0   \\
      \matindex{(2,X)} & 0 & 1  \\
          \end{block}
  \end{blockarray}
\]

    For both examples,  we can see that there does not exist any $m \in \N$ such that $M_{a^m}=M_{a^{m+1}}$, thereby making both SSTs non aperiodic.
On the other hand, for any string $s$, the  transition monoid of the $\sst{}$  $T_2$ 
in Figure \ref{fig:sst-ex2} is given by 

\[
  M_{s}=\begin{blockarray}{cccc}
    & \matindex{(1,X)} & \matindex{(1,Y)} & \matindex{(1,Z)}\\
    \begin{block}{c(ccc)}
      \matindex{(1,X)} & 1 & 0 & 0  \\
      \matindex{(1,Y)} & 0 & 1  & 0 \\
      \matindex{(1,Z)} & 0 & 0  & 1\\
                   \end{block}
  \end{blockarray} 
    \]
    Clearly, $M_{T_2}$ is aperiodic and 1-bounded. 
\end{example}

The following result states that the domain of an aperiodic, 1-bounded \sst{} is
FO-definable. 
\begin{proposition}\label{prop:fodomain}
  The domain of an aperiodic \sst{} is FO-definable. 
\end{proposition}
\begin{proof}
    Let $T = (\Sigma, \Gamma, Q, q_0, Q_f, \delta, \varsst, \rho, F)$
    be an aperiodic SST and $M_T$ its (aperiodic) transition monoid. Let
    us define a function $\varphi$ which associates with each matrix $M\in
    M_T$, the $|Q|\times |Q|$ Boolean matrix $\varphi(M)$ defined by 
    $\varphi(M)[p][q] = 1$ iff there exist $X,Y\in\varsst$ such that 
    $M_T[p,X][q,Y]\geq 0$. Clearly, $\varphi(M_T)$ is the transition
    monoid of the underlying input automaton of $T$ (ignoring the variable updates).  The result
    follows, since the homomorphic image of an aperiodic monoid is 
    aperiodic. 
\end{proof}

We show that an \sst{} is non-aperiodic iff its transition monoid
contains a non-trivial cycle. Checking the existence of a non-trivial
cycle has been shown to be in \textsc{PSpace} for deterministic
automata \cite{stern}, but this result can be extended to our setting.

\begin{lemma}\label{lem:complexity}
Given an SST $T$, checking whether it is aperiodic and 1-bounded is
\textsc{PSpace-complete}.
\end{lemma}
\begin{proof}
We first prove that given an $\sst{}$ $T$, checking whether its transition monoid 
$M_T$ is 1-bounded is in PSPACE. We then show that  checking whether a 1-bounded 
$\sst{}$ $T$ is aperiodic is  \textsc{PSpace-complete}. The full proof can be seen in Appendix \ref{complexity}. 
\end{proof}

The rest of the paper is devoted to the proof of the following key theorem. 
\begin{theorem}
  \label{thm:main}
  A string transformation is FO-definable iff it is definable by an aperiodic, 1-bounded
  \sst{}. 
\end{theorem}
The proof of this theorem follows from Lemma~\ref{lem:fo2sst}
(Section~\ref{sec:fo2sst}) and Lemma~\ref{lem:sst2fo}
(Section~\ref{sec:sst2fo}).

\section{From aperiodic 1-bounded SST to FOT} 
\label{sec:sst2fo}
In this section we show the following lemma by constructing an equivalent
\fot{} $T'$ for a given \sst{} $T$. 
\begin{lemma}
  \label{lem:sst2fo}
  A string transformation is FO-definable {\bf if it is} definable by an aperiodic,1-bounded SST.
\end{lemma}
 The idea closely follows the \sst{}-to-\msot{} construction of \cite{AC10,FiliotTrivediLics12}. 
The main challenge here is to show that aperiodicity and 1-boundedness on the \sst{} implies FO-definability of 
the output string structure (in particular the predicate $\preceq$). 

\subsection{FO-definability of variable flow}
\label{fo:varflow}
We first show that the variable flow of any aperiodic,1-bounded \sst{} is
FO-definable. 
This will be crucial to show that the output predicate $\preceq$ is
FO-definable.
\begin{proposition}\label{prop:foflow}
    Let $T$ be an aperiodic,1-bounded \sst{} $T$ with set of variables $\varsst$. 
    For all variables $X,Y\in \varsst$, there exists an FO-formula
    $\phi_{X\flows Y}(x,y)$ with two free variables such that,
    for all strings $s\in dom(T)$ and any two positions $i\leq j\in dom(s)$, 
    $s\models \phi_{X\flows Y}(i,j)$ iff $(q_i,X)\flows^{s[i{+}1{:}j]}_1
    (q_j,Y)$, where $q_0\dots q_n$ is the accepting run of $T$ on $s$.
\end{proposition}
Let $X\in\varsst$, $s\in\dom(T)$, $i\in \dom(s)$, and let $n = |s|$. 
We say that the pair $(X,i)$ is \emph{useful} if the content of
variable $X$ before reading $s[i]$ will be part of the output after
reading the whole string $s$. 
Formally, if $r = q_0\dots q_{n}$ is the accepting run of $T$ on $s$, then $(X,i)$ is useful 
for $s$ if $(q_{i-1}, X)\flows^{s[i{:}n] }_1(q_n,Y)$ for some variable $Y\in F(q_n)$.
Thanks to Proposition \ref{prop:foflow}, this property is FO-definable.
\begin{proposition}\label{prop:contribution}
    For all $X\in\varsst$, there exists an FO-formula
    $\contribute_X(i)$ s.t. for all strings $s\in dom(T)$ and all
    positions $i\in \dom(s)$, 
    $s\models \contribute_X(i)$ iff $(X,i)$ is useful for string $s$.
\end{proposition}
Proofs of propositions \ref{prop:foflow} and \ref{prop:contribution} can be found in Appendix \ref{app:varflow}.

\subsection{SST-output relational structure}
In this section, we define the \sst-output structure given an input string structure. 
It is an intermediate representation of the output, and the transformation of any input string into its   
\sst-output structure will be shown to be FO-definable.

For any \sst{} $T$ and string $s\in\dom(T)$,  the \sst-output structure of $s$ is a relational structure 
$G_T(s)$ obtained by taking, for each variable $X\in\varsst$, two copies of $\dom(s)$, respectively denoted by $X^{in}$ and $X^{out}$.  
For notational convenience we assume that these structures are labeled on the edges.
This structure satisfies the following invariants: for all $i\in
dom(s)$, $(1)$ the nodes $(X^{in}, i)$ and $(X^{out}, i)$ exist only if
$(X,i)$ is useful, and $(2)$ there is a directed path from $(X^{in}, i)$ to $(X^{out}, i)$ whose
sequence of labels is equal to the value of the variable $X$ computed by $T$ after reading $s[i]$.

\begin{figure*}
\begin{center}
\begin{tikzpicture}[->,>=stealth',shorten >=1pt,auto,node distance=1.5cm,
                    thick,scale=0.9, every node/.style={scale=0.9}]

\tikzstyle{graphnode}=[circle,fill=black!40,thick,inner
sep=0pt,minimum size=2mm]

\tikzstyle{graphnodeblack}=[circle,fill=black,thick,inner
sep=0pt,minimum size=2mm]

\node [graphnode] (x0in) at (0,5) {} ;
\node [graphnode] (x1in) at (3,5) {} ;
\node [graphnodeblack] (x2in) at (6,5) {} ;
\node [graphnodeblack] (x3in) at (9,5) {} ;
\node [graphnodeblack] (x4in) at (12,5) {} ;
\node [graphnodeblack] (x5in) at (15,5) {} ;

\node [graphnode] (x0out) at (0,4) {} ;
\node [graphnode] (x1out) at (3,4) {} ;
\node [graphnodeblack] (x2out) at (6,4) {} ;
\node [graphnodeblack] (x3out) at (9,4) {} ;
\node [graphnodeblack] (x4out) at (12,4) {} ;
\node [graphnodeblack] (x5out) at (15,4) {} ;

\node [graphnode] (y0in) at (0,3) {} ;
\node [graphnodeblack] (y1in) at (3,3) {} ;
\node [graphnodeblack] (y2in) at (6,3) {} ;
\node [graphnodeblack] (y3in) at (9,3) {} ;
\node [graphnodeblack] (y4in) at (12,3) {} ;
\node [graphnode] (y5in) at (15,3) {} ;

\node [graphnode] (y0out) at (0,2) {} ;
\node [graphnodeblack] (y1out) at (3,2) {} ;
\node [graphnodeblack] (y2out) at (6,2) {} ;
\node [graphnodeblack] (y3out) at (9,2) {} ;
\node [graphnodeblack] (y4out) at (12,2) {} ;
\node [graphnode] (y5out) at (15,2) {} ;

\node [graphnodeblack] (z0in) at (0,1) {} ;
\node [graphnodeblack] (z1in) at (3,1) {} ;
\node [graphnodeblack] (z2in) at (6,1) {} ;
\node [graphnodeblack] (z3in) at (9,1) {} ;
\node [graphnode] (z4in) at (12,1) {} ;
\node [graphnode] (z5in) at (15,1) {} ;

\node [graphnodeblack] (z0out) at (0,0) {} ;
\node [graphnodeblack] (z1out) at (3,0) {} ;
\node [graphnodeblack] (z2out) at (6,0) {} ;
\node [graphnodeblack] (z3out) at (9,0) {} ;
\node [graphnode] (z4out) at (12,0) {} ;
\node [graphnode] (z5out) at (15,0) {} ;

\node (xin) at (-1,5) {$X^{in}$};
\node (xout) at (-1,4) {$X^{out}$};
\node (yin) at (-1,3) {$Y^{in}$};
\node (yout) at (-1,2) {$Y^{out}$};
\node (zin) at (-1,1) {$Z^{in}$};
\node (zout) at (-1,0) {$Z^{out}$};

  \draw[->,dashed]   (x0in) -- node[left] {\textcolor{gray}{$\epsilon$}} (x0out) ;
  \draw[->,dashed]   (y0in) -- node[left] {\textcolor{gray}{$\epsilon$}} (y0out);
  \draw[->]   (z0in) -- node[left] {$\epsilon$} (z0out);


  \draw[->,dashed] (x1in) -- node[above] {\textcolor{gray}{$a$}} (x0in) ;
  \draw[->,dashed] (x0out) -- node[above] {\textcolor{gray}{$b$}} (x1out) ;

  \draw[->] (z1in) -- node[above] {$\epsilon$} (z0in) ;
  \draw[->] (z0out) -- node[above] {$c$} (z1out) ;


  \draw[->] (y1in) -- node[left] {$aaa$} (y1out) ;


  \draw[->] (y2in) -- node[above] {$\epsilon$} (y1in) ;
  \draw[->] (y1out) -- node[above] {$\epsilon$} (y2out) ;

  \draw[->] (z2in) -- node[above] {$d$} (z1in) ;
  \draw[->] (z1out) -- node[above] {$d$} (z2out) ;


  \draw[->] (x2in) -- node[left] {$c$} (x2out) ;
           

  \draw[->] (y3in) -- node[above] {$e$} (y2in) ;
  \draw[->] (y2out) -- node[above] {$f$} (y3out) ;

  \draw[->] (x3in) -- node[above] {$\epsilon$} (x2in) ;
  \draw[->] (x2out) -- node[above] {$\epsilon$} (x3out) ;

  \draw[->] (z3in) -- node[above] {$\epsilon$} (z2in) ;
  \draw[->] (z2out) -- node[above] {$\epsilon$} (z3out) ;


  \draw[->] (y3out) -- node[left] {$b$} (z3in) ;
  

  \draw[->] (x4in) -- node[above] {$\epsilon$} (x3in) ;
  \draw[->] (x3out) -- node[above] {$\epsilon$} (x4out) ;

  \draw[->] (z3out) -- node[above] {$c$} (y4out) ;
  \draw[->] (y4in) -- node[above] {$a$} (y3in) ;
  

  \draw[->] (x4out) -- node[left] {$e$} (y4in) ;
  \draw[->,dashed] (z4in) -- node[left] {$h$} (z4out) ;


  \draw[->] (x5in) -- node[above] {$\epsilon$} (x4in) ;
  \draw[->] (y4out) -- node[above] {$f$} (x5out) ;

  \draw[->,dashed] (z5in) -- node[above] {\textcolor{gray}{$\epsilon$}} (z4in) ;
  \draw[->,dashed] (z4out) -- node[above] {\textcolor{gray}{$\epsilon$}} (z5out) ;
  

  \draw[->,dashed] (y5in) -- node[left] {\textcolor{gray}{$a$}} (y5out) ;


\node (run) at (-1,6) {$run$};

\node (q0) at (0,6) {$q_0$} ;
\node (q1) at (3,6) {$q_1$} ;
\node (q2) at (6,6) {$q_2$} ;
\node (q3) at (9,6) {$q_3$} ;
\node (q4) at (12,6) {$q_4$} ;
\node (q5) at (15,6) {$q_5$} ;


 \draw[->] (q0) -- node[above] {$\begin{array}{lll} X & := & aXb
       \\ Y & := & aaa \\ Z & := & Zc \end{array}$}(q1) ;

 \draw[->] (q1) -- node[above] {$\begin{array}{lll} X & := & c
       \\ Y & := & Y \\ Z & := & dZd \end{array}$}(q2) ;

 \draw[->] (q2) -- node[above] {$\begin{array}{lll} X & := & X
       \\ Y & := & eYf \\ Z & := & Z \end{array}$}(q3) ;

 \draw[->] (q3) -- node[above] {$\begin{array}{lll} X & := & X
       \\ Y & := & aYbZc \\ Z & := & h \end{array}$}(q4) ;

 \draw[->] (q4) -- node[above] {$\begin{array}{lll} X & := & XeYf
       \\ Y & := & a \\ Z & := & Z \end{array}$}(q5) ;


\fill[blue!20,fill opacity=0.2] (-0.3,-0.3) rectangle (15.3,5.3);

\end{tikzpicture}
\end{center}
\vspace{-5mm}
\caption{\label{fig:outputgraph} \sst-output structure}
\vspace{-5mm}
\end{figure*}

The condition on usefulness of nodes implies that \sst-output structures consist of a single directed component, and 
therefore they are edge-labeled string structures. 

As an example of \sst-output structure consider Fig. \ref{fig:outputgraph}. 
We show only the variable updates. 
 Dashed arrows represent variable updates for useless variables, and therefore does not belong the \sst-output
structure. Initially the variable content of $Z$ is equal to $\epsilon$. 
It is represented by the $\epsilon$-edge from $(Z^{in}, 0)$ to $(Z^{out}, 0)$ in the first column. 
Then, variable $Z$ is updated to $Zc$. 
Therefore, the new content of $Z$ starts with $\epsilon$ (represented by the $\epsilon$-edge from $(Z^{in},1)$ to $(Z^{in},0)$,
which is concatenated with the previous content of $Z$, and then concatenated with $c$ (it is represented by the $c$-edge from
$(Z^{out},0)$ to $(Z^{out},1)$). 
Note that the invariant is satisfied. The output is given by the path from $(X^{in},5)$ to $(X^{out},5)$ 
and equals $ceaeaaafbdcdcf$. 
Also note that some edges are labelled by strings with several letters, but there are finitely many possible such strings. 
In particular, we denote by $O_T$ the set of all strings that appear in right-hand side of variable updates. 
\sst-output structures are defined formally in Appendix \ref{app:sst-output}.

\subsection{From SST to FOT}
\label{sst2fot}
It is known from \cite{AC10,FiliotTrivediLics12} that the transformation that maps a string $s$ to its \sst-output 
structure is MSO-definable. 
We show that it is FO-definable as long as the \sst{} is aperiodic and 1-bounded. 
The main challenge is to define the transitive closure of the edge relation in first-order.
We briefly recall the construction of \cite{AC10,FiliotTrivediLics12} in Appendix (in the proof of Lemma \ref{lem:sst2fo}) but rather focus on the
transitive closure in this section. 

Let $T = (Q, q_0, \Sigma, \Gamma, \varsst, \delta, \rho, Q_f)$. 
The \sst-output structure of $T$, as a node-labeled string, can be
seen as logical structures over the signature $S_{O_T} = \{ (E_{\gamma})_{\gamma\in O_T}, \preceq\}$ where the symbols $E_\gamma$  are binary predicates interpreted as edges
labeled by $O_T$. 
We let $E$ denote the edge relation, disregarding the labels. 
To prove that transitive closure is $\FO[\Sigma]$-definable, we use the fact that variable flow is $FO[\Sigma]$-definable. 
The following property is a key result towards FO-definability.
\begin{proposition}\label{prop:path}
    Let $T$ be an \textbf{aperiodic,1-bounded} \sst{} $T$. 
    Let $s\in dom(T)$, $G_T(s)$ its \sst-output structure and $r=q_0\dots q_n$ the accepting run of $T$ on $s$. 
    For all variables $X,Y\in \varsst$, all positions $i,j\in \dom(s)\cup\{0\}$, all $d,d'\in\{in,out\}$, there exists a path from node 
    $(X^{d},i)$ to node $(Y^{d'},j)$ in $G_T(s)$ iff $(X,i)$ and
    $(Y,j)$ are both useful and one of the following conditions hold: either
    \begin{enumerate}
      \item
$(q_j,Y)\flows^{s[j{+}1{:}i]}_1(q_i,X)$ and $d = in$, or
    \item 
$(q_i,X)\flows^{s[i{+}1{:}j]}_1 (q_j,Y)$ and $d' = out$, or
    \item 
there exists $k\geq max(i,j)$ and two variables $X',Y'$ such
      $(q_i,X)\flows^{s[i{{+}1:}k]}_1 (q_k,X')$, $(q_j,Y)\flows^{s[j{+}1{:}k]}_1 (q_k,Y')$
      and $X'$ and $Y'$ are concatenated in this order\footnote{by concatenated we
        mean that there exists a variable update whose rhs is of the
        form $\dots X'\dots Y'\dots$} by $r$ when reading $s[k+1]$. 
    \end{enumerate}
\end{proposition}

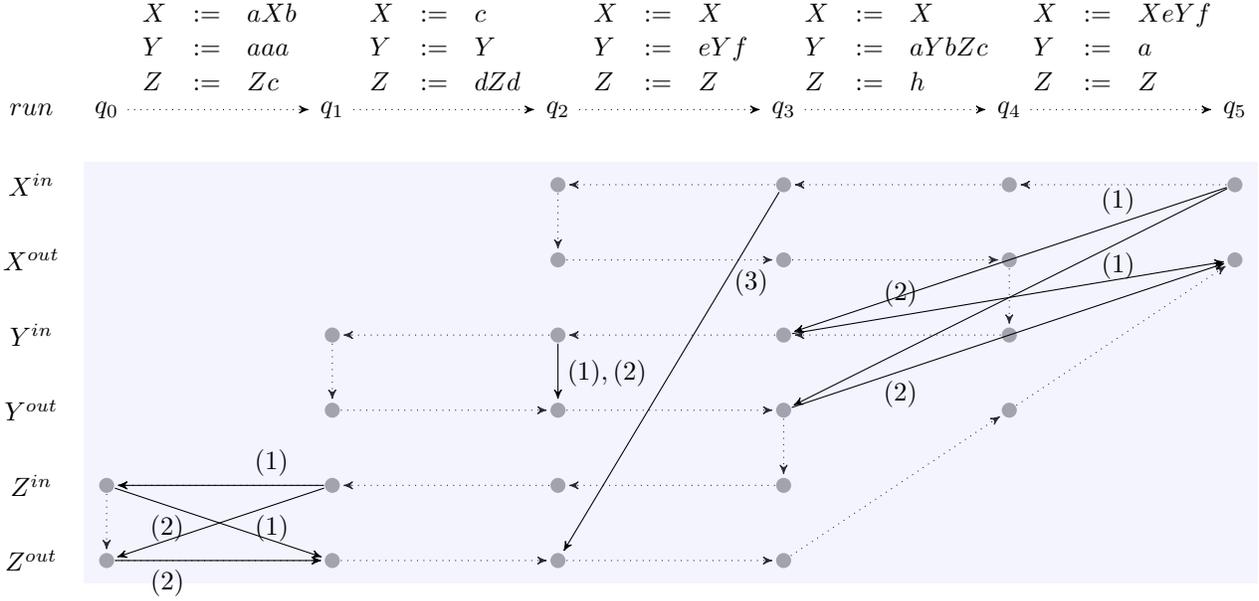
\begin{figure*}
\begin{center}
\begin{tikzpicture}[->,>=stealth',shorten >=1pt,auto]

\tikzstyle{graphnode}=[circle,fill=black!40,thick,inner
sep=0pt,minimum size=2mm]

\node [graphnode] (x2in) at (6,5) {} ;
\node [graphnode] (x3in) at (9,5) {} ;
\node [graphnode] (x4in) at (12,5) {} ;
\node [graphnode] (x5in) at (15,5) {} ;

\node [graphnode] (x2out) at (6,4) {} ;
\node [graphnode] (x3out) at (9,4) {} ;
\node [graphnode] (x4out) at (12,4) {} ;
\node [graphnode] (x5out) at (15,4) {} ;

\node [graphnode] (y1in) at (3,3) {} ;
\node [graphnode] (y2in) at (6,3) {} ;
\node [graphnode] (y3in) at (9,3) {} ;
\node [graphnode] (y4in) at (12,3) {} ;

\node [graphnode] (y1out) at (3,2) {} ;
\node [graphnode] (y2out) at (6,2) {} ;
\node [graphnode] (y3out) at (9,2) {} ;
\node [graphnode] (y4out) at (12,2) {} ;

\node [graphnode] (z0in) at (0,1) {} ;
\node [graphnode] (z1in) at (3,1) {} ;
\node [graphnode] (z2in) at (6,1) {} ;
\node [graphnode] (z3in) at (9,1) {} ;

\node [graphnode] (z0out) at (0,0) {} ;
\node [graphnode] (z1out) at (3,0) {} ;
\node [graphnode] (z2out) at (6,0) {} ;
\node [graphnode] (z3out) at (9,0) {} ;

\node (xin) at (-1,5) {$X^{in}$};
\node (xout) at (-1,4) {$X^{out}$};
\node (yin) at (-1,3) {$Y^{in}$};
\node (yout) at (-1,2) {$Y^{out}$};
\node (zin) at (-1,1) {$Z^{in}$};
\node (zout) at (-1,0) {$Z^{out}$};

\tikzstyle{every path}=[dotted]

  \draw[->]   (z0in) --  (z0out);


  \draw[->] (z1in) --  (z0in) ;
  \draw[->] (z0out) --  (z1out) ;


  \draw[->] (y1in) --  (y1out) ;


  \draw[->] (y2in) --  (y1in) ;
  \draw[->] (y1out) --  (y2out) ;

  \draw[->] (z2in) --  (z1in) ;
  \draw[->] (z1out) --  (z2out) ;


  \draw[->] (x2in) --  (x2out) ;
           

  \draw[->] (y3in) --  (y2in) ;
  \draw[->] (y2out) --  (y3out) ;

  \draw[->] (x3in) --  (x2in) ;
  \draw[->] (x2out) --  (x3out) ;

  \draw[->] (z3in) --  (z2in) ;
  \draw[->] (z2out) --  (z3out) ;


  \draw[->] (y3out) --  (z3in) ;
  

  \draw[->] (x4in) --  (x3in) ;
  \draw[->] (x3out) --  (x4out) ;

  \draw[->] (z3out) --  (y4out) ;
  \draw[->] (y4in) --  (y3in) ;
  

  \draw[->] (x4out) --  (y4in) ;


  \draw[->] (x5in) -- (x4in) ;
  \draw[->] (y4out) --  (x5out) ;



\node (run) at (-1,6) {$run$};

\node (q0) at (0,6) {$q_0$} ;
\node (q1) at (3,6) {$q_1$} ;
\node (q2) at (6,6) {$q_2$} ;
\node (q3) at (9,6) {$q_3$} ;
\node (q4) at (12,6) {$q_4$} ;
\node (q5) at (15,6) {$q_5$} ;


 \draw[->] (q0) -- node[above] {$\begin{array}{lll} X & := & aXb
       \\ Y & := & aaa \\ Z & := & Zc \end{array}$}(q1) ;

 \draw[->] (q1) -- node[above] {$\begin{array}{lll} X & := & c
       \\ Y & := & Y \\ Z & := & dZd \end{array}$}(q2) ;

 \draw[->] (q2) -- node[above] {$\begin{array}{lll} X & := & X
       \\ Y & := & eYf \\ Z & := & Z \end{array}$}(q3) ;

 \draw[->] (q3) -- node[above] {$\begin{array}{lll} X & := & X
       \\ Y & := & aYbZc \\ Z & := & h \end{array}$}(q4) ;

 \draw[->] (q4) -- node[above] {$\begin{array}{lll} X & := & XeYf
       \\ Y & := & a \\ Z & := & Z \end{array}$}(q5) ;


\fill[blue!20,fill opacity=0.2] (-0.3,-0.3) rectangle (15.3,5.3);


\tikzstyle{every path}=[solid]


\draw[->] (y2in) -- node[right] {$(1),(2)$} (y2out);

\draw[->] (z0in) -- node[near start, below] {$(2)$} (z1out);
\draw[->] (z0out) -- node[near start, below] {$(2)$} (z1out);

\draw[->] (z1in) -- node[near start, below] {$(1)$} (z0out);
\draw[->] (z1in) -- node[near start, above] {$(1)$} (z0in);

\draw[->] (y3in) -- node[near start, above] {$(2)$} (x5out);
\draw[->] (y3out) -- node[near start, below] {$(2)$} (x5out);


\draw[->] (x5in) -- node[near start, above] {$(1)$} (y3in);
\draw[->] (x5in) -- node[near start, below] {$(1)$} (y3out);


\draw[->]  (x3in) -- node[near start, right] {$(3)$}  (z2out);

\end{tikzpicture}
\end{center}
\caption{\label{fig:Xoutputgraph} Conditions of Proposition \ref{prop:path}}
\end{figure*}

We illustrate the conditions of this proposition on
Fig.\ref{fig:Xoutputgraph}. We have for instance $(q_2,Y)\flows^{s[3{:}2]=\epsilon}_1 (q_2,Y)$, therefore
by conditions $(1)$ (and $(2)$) by taking $X=Y$ and $i=j=2$, there
exists a path from $(Y^{in},2)$ to $(Y^{out},2)$. Note that none of these
conditions imply the existence of an edge from $(Y^{out},2)$ to
$(Y^{in},2)$, but self-loops on $(Y^{in},2)$ and
$(Y^{out},2)$ are implied by conditions
$(1)$ and $(2)$ respectively. Now consider positions $0$ and $1$ and variable
$Z$. It is the case that $(q_0,Z)\flows^{s[1{:}1]}_1 (q_1,Z)$, therefore by
condition $(1)$ there is a path from $(Z^{in},1)$ to $(Z^{in},0)$ and
to $(Z^{out},0)$. Similarly, by condition $(2)$ there is a path from
$(Z^{in},0)$ to $(Z^{out},1)$ and from $(Z^{out},0)$ to $(Z^{out},1)$.
For positions $3$ and $5$, note that $(q_3,Y)\flows^{s[4{:}5]}_1 (q_5,X)$, hence
there is a path from $(Y^d,3)$ to $(X^{out},5)$ for all
$d\in\{in,out\}$. By condition $(2)$ one also gets edges from
$(X^{in},5)$ to $(Y^d,3)$. Finally consider nodes $(Z^{out},2)$ and 
$(X^{in},3)$. There is no flow relation between variable $Z$ at
position $2$ and variable $X$ at position $3$. However,
$(q_3,X)\flows^{s[4{:}4]}_1 (q_4,X)$ and $(q_2,Z)\flows^{s[3{:}4]} (q_4,Y)$. Then $X$ and $Y$
gets concatenated at position $4$ to define $X$ at position $5$. Therefore there is a path from 
$(X^{in}, 3)$ to $(Z^{out},2)$: this case is covered by condition
$(3)$.

From this result and FO-definability of variable flow, one can show that 
transitive closure  is FO-definable.
\begin{lemma}\label{lem:fopath}
    Let $T$ be an \textbf{aperiodic,1-bounded} \sst{} $T$. 
    For all $X,Y\in \varsst$ and all $d,d'\in
    \{in,out\}$, there exists an FO[$\Sigma$]-formula 
    $\text{path}_{X,Y,d,d'}(x,y)$ with two free variables
    such that for all strings $s\in\dom(T)$ and all positions
    $i,j\in dom(s)$, $s\models \text{path}_{X,Y,d,d'}(i,j)$ iff 
    there exists a path from $(X^d,i)$ to $(Y^{d'},j)$ in
    $G_T(s)$.
\end{lemma}
The proof of Lemma \ref{lem:fopath} can be seen in Appendix \ref{app:sst2fot}.
We are now in a position to sketch the proof of Lemma~\ref{lem:sst2fo} of this section. 
Let $\Gamma$ be the output alphabet.  
The construction presented in \cite{FiliotTrivediLics12,AC10} shows the MSO-definability of strings to \sst-output 
structures. 
We adapt this construction and based on FO-definability of transitive closure, as proved in Lemma \ref{lem:fopath}, we show 
that strings to \sst-output structure transformations are FO-definable whenever the \sst{} is aperiodic and 1-bounded.  
In~\cite{FiliotTrivediLics12,AC10}, \sst-output structures also  contain useless nodes which are later on removed by composing 
another definable transformation. 
Based on Proposition \ref{prop:contribution} which states that usefulness of nodes is FO-definable, we rather directly filter out 
these nodes in the first FO-transformation.  
\sst-output structures are however edge-labeled strings over $O_T$, where $O_T$ is a finite set of strings over $\Gamma$.    
It remains to transform an edge-labeled string over $O_T$ into a  (node-labeled) string over $\Gamma$. This transformation is again
FO-definable by taking a suitable number of copies of the input domain ($max \{ |s|\ |\ s\in O_T\}$).
Then the lemma follows from the closure of FO-transformations under composition \cite{Cour94}.

\section{From FOT to aperiodic 1-bounded SST} 
\label{sec:fo2sst}
\label{sec:sstla2sst}

The goal of this section is to prove the following lemma by showing a reduction
from FO-definable transformations to aperiodic, 1-bounded \sst{}s.
\begin{lemma}
  \label{lem:fo2sst}
  A string transformation is FO-definable {\bf only if} it is definable by an
  aperiodic, 1-bounded SST.
\end{lemma}

We begin this section by introducing the notion of aperiodic,1-bounded SSTs with look-ahead, 
and show that they are equi-expressive to aperiodic,1-bounded SSTs. 
We will then construct an aperiodic, 1-bounded \sst{} with look-ahead implementing the same 
transformation as the given FOT. 
While this construction of the $\sst{}$ with look-ahead closely mimics the construction of
\cite{ADT13}, we show that it preserves aperiodicity and 1-boundedness
(Section~\ref{sec:aperiodicSSTLA}).

\subsection{SSTs with Lookahead}
\label{sec:sst-la}
As an intermediate model we introduce streaming string transducers with
look-ahead ($\sstla{}$), which are \sst{}s that can make transitions based on some
regular property of the current suffix of the input string.   
Such properties can be conveniently specified via a single finite automaton whose
different states characterize various regular properties. 
Intuitively, while processing a symbol $a_i$ of an input $w=a_1a_2 \dots a_n$,
 the SST moves from its current state to some state $q_{i}$
 iff there exists a unique state $p_{i}$ 
 of the look ahead automaton such that $a_ia_{i+1} \dots a_n \in L(p_{i})$.  
As the string is processed, along with the current state of the SST, a set of states 
of the lookahead automaton is also maintained. 
   
Formally, a (deterministic) lookahead automaton is a tuple $A = (Q_A, \Sigma, \delta_A, P_f)$ such
that for all $p \in Q_A$ the tuple $A_p = (Q_A, p, \Sigma, \delta_A, P_f)$
($A$ with initial state $p$) is a
deterministic finite automaton (we write $L(A_p)$ for the language
that it recognizes).

\begin{definition}
\label{sst-LA}
  An \sst{} with lookahead is a tuple $(T, A)$ where $A = (Q_A,
  \Sigma, \delta_A, P_f)$ is a (deterministic) lookahead automaton and $T$ is a tuple 
  $(\Sigma, \Gamma, Q, q_0, Q_f, \delta, \varsst, \rho, F)$ where 
  $\Sigma, \Gamma, Q, q_0, Q_f,
  \varsst, \rho$, and $F$  are defined as for \sst{}s, and $\delta: Q \times
  \Sigma \times P \to Q$ is the transition function.  
 We further require that the look-aheads are mutually exclusive, i.e. for
  all symbols $a\in\Sigma$, all states $q\in Q$, and all transitions
  $q' = \delta (q,a,p)$ and $q'' = \delta(q,a,p')$, we have that  
  $L(A_p)\cap L(A_{p'})=\varnothing$.
 \end{definition}
 The requirement that look-aheads are mutually exclusive
ensures that the $\sstla$ is deterministic: when reading a new symbol,
there is at most one transition that can be triggered. It is obvious
that this requirement can be checked in polynomial time: 
 whenever  $q' = \delta (q,a,p)$ and $q'' = (q,a,p')$,
  we can indeed construct a deterministic automaton ${\cal A}_{mutex}$ which 
  starts from the final states of $A_p$ and $A_{p'}$ and walks backward to $(p,p')$ 
  such that $L({\cal A}_{mutex})=\emptyset$.

A configuration of an $\sstla$ is a pair $(q_i, P_i) \in Q {\times} 2^P$. 
A run $r$ of $T$ over string $s=a_1\ldots a_n \in \Sigma^*$ is a
sequence of configurations and letters $r = (q_0, P_0)\xrightarrow{a_1}(q_1,P_1)\dots (q_{n-1},P_{n-1})\xrightarrow{a_n}
(q_n,P_n)$ such that for all $i\in \{0,\dots,n-1\}$,  $(q_i, P_i)\xrightarrow{a_{i+1}}
(q_{i+1}, P_{i+1})$ if there exists $p\in P_{i+1}$ such that   
$\delta(q_i, a_{i+1}, p) =  q_{i+1}$, and for all $p\in P_i$,
$\delta_A(p,a_{i+1})\in P_{i+1}$. We write
$(q_0,P_0)\flows^s (q_n,P_n)$ if such a sequence exists. We say that
$r$ is accepting if $(q_0,P_0)$ is an initial configuration, i.e. 
$q_0\in Q_0$ and $P_0=\varnothing$, and $(q_n,P_n)$ is an accepting
configuration, i.e. $q_n\in Q_f$ and $P_n\subseteq P_f$. 
Clearly, if $r$ is accepting, then for all $i\in \{1,\dots,n-1\}$, $a_{i+1}\dots
a_n\in L(A_p)$, where $p$ is the look-ahead state of the $i$-th
transition of $r$. A configuration is said to be \emph{accessible} if
it can be reached from an initial configuration, and
\emph{co-accessible} if from it an accepting configuration can be
reached. It is \emph{useful} if it is both accessible and
co-accessible. Note that from the mutual-exclusiveness of
look-aheads and the determinism of $A$, it follows that for any input string, there is at most one
run of the $\sst_\la$ from and to useful configurations, as shown in
Appendix \ref{app:sst-la}. 

The concept of substitutions induced by a run can be naturally extended from
\sst{}s to \sstla{}s.
Also, we can define the transformation implemented by an \sstla{} in a
straightforward manner.
The transition monoid of an \sstla{} is defined by matrices indexed by
configurations $(q_i, P_i)\in Q\times 2^P$, using the notion of run defined before,
and the definition of aperiodicity of \sstla{} follows that of \sst.
Adding look-aheads (in an aperiodic fashion) to \sst{} does not increase
their expressive power, see Appendix \ref{sec:sstla2sst}.

\begin{lemma}\label{lem:aperiodicSSTLA}
    For all aperiodic,1-bounded \sst{} with look-ahead, there exists an
    equivalent aperiodic, 1-bounded \sst{}.
\end{lemma}

\subsection{From FOT to SST with look-ahead} 
\label{fot:sstla}
The main complication in showing this construction is that FO-transducers are
descriptional i.e. they describe the function using logical formulas, while
streaming string transducers are computational as they compute
the output string by reading the input string in one left-to-right pass of the
input string.
Our goal is to construct an SST from an FO-transducer in such a way that after
reading the string till the position $i$ the variables in the SST will store the
substrings of the output corresponding to positions up to index $i$ in different
variables, and to devise an update function for these variables to keep this invariant. 

For instance, consider the FO-transduction shown in Figure~\ref{fig:foexample} 
till position $3$. 
Assume  we omit the positions and edges of the output graph post position $3$.  
Upto position 3, the output
graph consists of two strings: the first string is between the positions $1^1$
and $3^1$ and stores $aa$, while the second string is between  positions
$3^2$  and $2^3$ and stores the string $abab$. 
Let us assume that these strings are stored in variables $X_1$ and $X_2$,
respectively. 
When we read the next letter of the string at position $4$, we need to update
these variables so as to append the letter $a$ in the string stored in variable
$X_1$, while prepend the letter $a$ to the string stored in variable $X_2$ using
perhaps the following updates: $X_1 {:=} X_1 a$ and $X_2 {:=} a X_2$.
The next goal here is to identify the beginning  (``$i$-head'') and the
ending (``$i$-tails'') points of these output sub-string before the position $i$,
and update them as we process the input string.  
In this section we show that these sub-strings can be uniquely identified using
the $k$-types of a suitable decomposition of the input string.

\vspace{3mm}
\noindent {\bf Heads and tails of output substrings}.
We fix an FO transducer $T = (\Sigma, \Gamma, \phi_{\dom}, C, \phi_{\positions},
\phi_{\preceq})$ and let $k$ be its quantifier rank.
Let $s\in\Sigma^*$ and $j\in dom(s)$. For all copies $c\in C$, we denote by 
$j^c$ the $c$th copy of the input $j$ position, and say that $j^c$ is
\emph{alive} if it contributes to the output string, i.e.  there exists some
$\gamma\in\Gamma$ such that $s {\models} \phi_\gamma^c(j)$.  
For instance, on Fig.\ref{fig:foexample}, alive positions are in bold.
This can be defined in FO.

For $j \leq i \in \dom(s)$, we call a position $j^c$ an $i$-head if $j^c$ is
alive and there is no incoming edge to $j^c$ that comes from some position $l^d$
for some position $l \leq i$ and some $d\in C$. 
Formally, $j^c$ is an $i$-head if $s{\models} \text{head}_c(i,j)$ where
$\text{head}_c(x,y)$ is the following FO-formula:
\[
\text{head}_c(x,y) \rmdef 
y\preceq x\wedge \text{alive}_c(y)\wedge  
\neg \exists z\cdot z \preceq x\wedge  \bigvee_{d\in C} 
\text{alive}_d(z) \wedge \phi^{d,c}_\suc(z,y)
\]
where $\phi^{d,c}_\suc(z,y)$ defines the (output) successor relation
(it is FO-definable using $\Phi^{d,c}_\preceq$). The notion of \emph{$i$-tail} can be defined similarly. Formally, $j^c$ is an $i$-tail if 
$s{\models} \text{tail}_c(i,j)$ where $\text{tail}_c(x,y) = y\preceq x\wedge \text{alive}_c(y)\wedge 
\neg \exists z\cdot z\preceq x\wedge \bigvee_{d\in C}
\text{alive}_d(z) \wedge \phi^{c,d}_\suc(y,z)$.

The following lemma (proof in Appendix \ref{app:fot-sstla}) states for all strings $s$, all $i\in dom(s)$, an $i$-tail
or an $i$-head, $j^c$, is uniquely determined by the 
$k$-type of the string $s[1{:}j)$, 
$k$-type of the string $s[j{:}i)$, $k$-type of the string $s(i{:}|s|]$, the symbol $s[j]$, 
and the corresponding copy $c$.  

\begin{lemma}
  \label{lem:addr}
Let $s\in\Sigma^*$, $i\in dom(s)$, $c\in C$, and $a\in\Sigma$. Let
$j_1,j_2\in dom(s)$. Then $j_1=j_2$ if:
$(1)$ $j_1<i$ and $j_2<i$, $(2)$ $s[j_1] = s[j_2] = a$, $(3)$ $s[1{:}j_1) \equiv_{k+2} s[1{:}j_2)$,
$(4)$ $s(j_1{:}i] \equiv_{k+2} s(j_2{:}i]$, and $(5)$
$j_1^c$ and $j_2^c$ are either both $i$-tails or both $i$-heads.
\end{lemma}

As a corollary, the number of $i$-tails
and $i$-heads is bounded by a constant that only depends on the transducer $T$.
\begin{corollary}
For all $s\in\Sigma^*$, all $i\in dom(s)$ and $c\in C$, the number
of $i$-tails and $i$-heads is bounded by $|\Theta_{k+2}|^2.|\Sigma|.|C|$.
\end{corollary}
Lemma~\ref{lem:addr} hints at a unique way to name a sub-string computed till
position $i$ by the unique address of its $i$-head $j^c$, as the tuple
$(\ktype{s[1{:}j)}{k+2}, \ktype{s(j{:}i]}{k+2}, s[j], c)$.
An \emph{address} is defined as a tuple $\addr\in \Theta_{k+2}^2\times
\Sigma \times C$. We denote by $\tau_1(\addr)$,  $\tau_2(\addr)$, $a(\addr)$, and $c(\addr)$ the
projections of $\addr$ on the first, second, third, and fourth components,
respectively.  The set of addresses is denoted by $\Addr_{T}$. 
  
\begin{figure}
\begin{center}
\begin{tikzpicture}[->,>=stealth',shorten >=1pt,auto,scale=0.9, every node/.style={scale=0.9}]

\tikzstyle{graphnode}=[circle,fill=black,thick,inner
sep=0pt,minimum size=1.5mm]

\node [draw=none] (input) at (-1,0.1) {\begin{scriptsize}input string\end{scriptsize}} ;
\node [draw=none] (input) at (-1,-0.2) {\begin{scriptsize}$s$\end{scriptsize}} ;
\node [draw=none] (input) at (-1,-0.7) {\begin{scriptsize}output string\end{scriptsize}} ;
\node [draw=none] (input) at (-1,-1) {\begin{scriptsize}graph\end{scriptsize}} ;

\node [graphnode] (j1) at (2,0) {} ;
\node  (j1) at (2,-0.3) {$j$} ;
\node  (j2) at (3,-0.3) {$j'$} ;
\node  (i) at (6,-0.3) {$i$} ;

\node  (j1) at (2,0.3) {$a(\addr)$} ;

\draw[|-(] (0,0) -- node [above] {$s_1$} (1.8,0) ;
\draw[)-)] (2.2,0) -- node [above] {$s_2$} (6.1,0) ;
\draw[-|] (6,0) -- node [above] {$s_3$} (7,0) ;

\node [graphnode] (j1c) at (2,-2) {} ;
\node [graphnode] (j2c) at (3,-3) {} ;

\draw [->,snake=snake] (j1c) -- (3,-3);

\draw [-,densely dotted] (i) -- (6,-3.5);

\draw [-,dotted] (2,-0.5) --   (j1c);
\draw [-,dotted] (0,-2) --  node [above] {$\tau_1(\addr)$}
node [below] {$(=[s_1]_k)$} (j1c);
\draw [-,dotted] (j1c) --  node [above] {$\tau_2(\addr)$}
node [below] {$(=[s_2]_k)$} (6,-2);

\draw [-,dotted] (3,-0.5) --   (j2c);
\draw [-,dotted] (0,-3) --  (j2c);
\draw [-,dotted] (j2c) --  (6,-3);

\fill[blue!20,fill opacity=0.2] (0,-0.5) rectangle (7,-3.5);

\node  at (-0.8, -2) {\begin{scriptsize}copy $c(\addr)$\end{scriptsize}};
\node  at (-0.8, -3) {\begin{scriptsize}copy $c'$\end{scriptsize}};

\end{tikzpicture}
\end{center}

\caption{\label{fig:address} Head $\absaddr(s,i,\alpha) = (j, c(\alpha))$ and tail $\tailaddr(s, i,
\alpha) = (j', c')$ for an address $\alpha\in\Addr_T$.}

\end{figure}

As a consequence of Lemma \ref{lem:addr}, given a string $s \in\Sigma^*$ and a
position $i\in \dom(s)$, any address $\addr\in\Addr_{T}$ defines at most one
$i$-tail or $i$-head in $s$. 
The head $\absaddr(s, i, \alpha)$ of an address  $\addr\in\Addr_{T}$ at
position $i$ in some input string $s \in \Sigma^*$ is \emph{the} position
$(j, c) \in \dom(s)\times C$ in the output structure s.t. 
$s {\models} \text{head}_c(i,j)$, 
$\tau_1(\addr) =\ktype{s[1{:}j)}{k+2}$,
$a(\addr) = s[j]$,
$\tau_2(\addr) = \ktype{s(j{:}i]}{k+2}$, and
$c(\addr) = c$ (By Lemma \ref{lem:addr}, $(j,c)$ is indeed unique).
If these conditions are not satisfied, then we say that $\absaddr(s,i,\addr)$ is
undefined.  
Similarly, the tail $\tailaddr(s, i, \alpha)$ of an address
$\addr\in\Addr_T$ at position $i$ in $s \in \Sigma^*$ is defined if 
(i) 
there exists 
some $(j',c')$ such that $(j', c')
= \absaddr(s, i, \alpha)$ is defined, and (ii) $\tailaddr(s, i, \alpha)$ is the position $(j, c) \in
\dom(s)\times C$ in the output structure such that  
$s {\models} \text{tail}_c(i,j)\wedge\phi_{\preceq}^{c', c}(j', j)$,
and for all $c''\in C$, all $j''>i$, $s\not\models \phi_{\preceq}^{c',
  c''}(j', j'')\wedge\phi_{\preceq}^{c'', c}(j'',j)$ (i.e. the path
from $(j',c')$ to $(j,c)$ only consists of positions $(j'',c'')$ such
that $j''\leq i$).

Fig.~\ref{fig:address} illustrates the notions of $i$-head and $i$-tail of an 
address. It represents an output position $j^{c(\alpha)}$ which is the head of the
address $\addr\in\Addr_{T}$ at position $i$ in string $s$. 
The input string $s$ is decomposed as $s = s_1 (a(\alpha)) s_2 s_3$
such that $[s_1]_k = \tau_1(\addr)$ and $[s_2]_k = \tau_2(\addr)$. 
From the definition it is clear that the heads and the tails of addresses  are 
FO-definable.  The proof of Lemma \ref{lem:addrfo} can be found in Appendix \ref{app:addrfo}.
\begin{lemma}\label{lem:addrfo}
  The functions $\absaddr$ and $\tailaddr$ are FO-definable, i.e. given $\addr{\in} \Addr_{T}$ and
  a copy $c\in C$, there
  exist two FO-formula $\Phi_{\absaddr(\addr)}^c(x,y)$ (of quantifier
  rank at most $k+2$) and
  $\Phi_{\tailaddr(\addr)}^c(x, y)$ (of quantifier rank at most $k+3$)
  such that for all $s{\in}\Sigma^*$ and $i,j{\in}
  \dom(s)$, $s{\models} \Phi_{\absaddr(\addr)}^c(i,j)$ iff
  $\absaddr(s,i,\addr)$ is defined and $\absaddr(s,i,\addr) {=} (j,
  c)$, and, $s{\models} \Phi_{\tailaddr(\addr)}^c(i,j)$ iff
  $\tailaddr(s,i,\addr)$ is defined and $\tailaddr(s,i,\addr) {=} (j, c)$.
\end{lemma}

\subsection*{SST Construction}
\label{sst-cons}
Given an FOT, to obtain the corresponding $\sst{}$,  
we define the set of SST variables $\varsst = \set{X_\alpha \::\: \alpha \in
  \Addr_{T}}$. 
While reading a string, we will maintain the invariant that after reading the
position $i$ of the input string $s$, the  variable $X_\alpha$ will store the
output substring rooted at position $j^c = \absaddr(s, i, \alpha)$ iff $j^c$ is
an $i$-head, otherwise the variable $X_\alpha$ will contain $\epsilon$. 

The next challenge is to show how to update these string variables.  
There are several cases to consider depending on the new direct edges in the
output graph from some copy in the current position to a head or a tail of a
variable relative to the previous position, or vice-versa. 
In general, for a variable $X_\alpha$ we have an update rule $X_\alpha = \gamma
X_{\alpha_1} \gamma_1 X_{\alpha_2} \ldots X_{\alpha_n} \gamma_n$ such that $|\gamma
X_{\alpha_1} \gamma_1 X_{\alpha_2} \ldots X_{\alpha_n} \gamma_n| \leq |C|$.  Thus,
 there are only a bounded number of updates to consider. 
Given a string and a position $i \in \dom(s)$ we can write an FO-formula 
$\Phi_{upd}[X_\alpha {:=} \gamma X_{\alpha_1} \gamma_1 X_{\alpha_2} \ldots X_{\alpha_n}
\gamma_n](i)$ of quantifier rank at most $k+5$  which characterizes the update
corresponding to the current position.
We briefly sketch some update formulas. For instance,  
\begin{enumerate}
\item 
$s {\models} \Phi_{upd}[X_\alpha{:=} \epsilon](i)$ if $\absaddr(s, i, \alpha)$ is
  not defined;  
  \item  $s {\models} \Phi_{upd}[X_\alpha {:=} X_{\alpha'}](i)$ if both $\absaddr(s, i, \alpha)$ and $\absaddr(s, i{-}1,
  \alpha')$ are defined and are equal to each-other and $\tailaddr(s, i, \alpha)
  = \tailaddr(s, i{-1}, \alpha')$;  
\item  $s {\models} \Phi_{upd}[X_\alpha {:=} \gamma X_{\alpha'} \gamma' X_{\alpha''}
  \gamma''](i)$ if $\absaddr(s, i, \alpha)$ is defined,     
  both $\absaddr(s, i{-}1, \alpha')$ and $\absaddr(s, i{-}1, \alpha'')$ are defined, 
    $\tau_2(\alpha) {=} \ktype{\epsilon}{k}$ and  there is an edge in the output
    structure from $i^{c(\alpha)}$ to $\absaddr(s, i-1, \alpha')$,
    the  label of the node $i^{c(\alpha)}$ is $\gamma$,  
    there is a copy $c'$ such that the position $i^{c'}$ is labeled $\gamma'$
    and $i^{c'}$ has a direct edge from $\tailaddr(s, i-1, \alpha')$ and $i^{c'}$ has a direct edge to
    $\absaddr(s, i-1, \alpha'')$, and 
    there is copy $c''$ such that the position $i^{c''}$ is labeled $\gamma''$
    and has a direct edge from $\tailaddr(s, i-1, \alpha'')$. By reusing
    variable names, we have to use only 2 nested extra quantifiers to
    express this formula, and therefore, since any formula
    $\Phi_{\tailaddr{\beta}}^c$ has quantifier rank at most $k+3$, we
    can express this variable update by a formula of quantifier rank
    at most $k+5$. This variable update easily generalizes to longer
    concatenations of variables, while using formulas of quantifier
    rank at most $k+5$ only. 
\end{enumerate}
We also define the look-around formula $\Phi_{\tau, a, \tau'}(i)$  that holds
for a string $s$ if the substring $s[1{:}i) \models \Phi_\tau$,
the substring $s(i{:}|s|] \models \Phi_{\tau'}$ and $s[i] = a$.

Now we are in a position to construct an equivalent $\sstla$ $(T_\la, A)$ from a
given FOT $T = (\Sigma, \Gamma, \phi_{\dom}, C, \phi_{\positions},
\phi_{\preceq})$.  
Let $T_\la= (\Sigma, \Gamma, Q, q_0, Q_f, \delta, \varsst, \rho, F)$ be a
look-ahead \sst{} with look-ahead $A = (Q_A,\Sigma, \delta_A, P_f)$.
The look-ahead automaton $A = (Q_A,\Sigma, \delta_A, P_f)$ is constructed as
a collection of automata that capture FO sentence $\Phi_\tau$ for all $\tau \in
\Theta_{k+2}$. 
More precisely, $A = \biguplus_{\tau\in \Theta_{k+2}} A_\tau$ where $A_\tau$ is the automaton 
accepting strings of type $\tau$ as introduced in Sec.~\ref{sec:prelims}.
For convenience we assume that the states of $A$ are pairs $(\tau, \tau')$
where $\tau$ corresponds to the FO type that is checked and $\tau'$ is a state
of $A_\tau$, and write $p_\tau\in Q_A$ for the state
$(\tau,\ktype{\epsilon}{k+2})\in Q_A$. In particular, the set of strings
$s$ such that $\ktype{s}{k+2} = \tau$ equals $L(A_{p_\tau})$.
The SST $T_\la$ is a tuple $(\Sigma, \Gamma, Q, q_0, Q_f, \delta, \varsst, \rho, F)$ where  
\begin{itemize}
 \item 
  the set of states is the set of $k+2$ types, i.e. $Q =
  \Theta_{k+2}$;
\item the initial state is $q_0 = \ktype{\epsilon}{k+2}$;
 \item 
  the set of final states are the $(k+2)$-types that implies the domain
  formula $\phi_{\dom}$ (on strings), i.e. $Q_f = \set{\tau \::\: \tau\models \phi_{\dom}}$;
 \item 
  the transition function $\delta: Q \times \Sigma \times Q_A \to Q$ is
  defined such that $\delta(\tau, a, p_{\tau'})=\tau''$ where
  $\tau''=\tau.\ktype{a}{k}$;
 \item 
  the set of variables is defined as $\varsst = \set{X_\alpha \::\: \alpha \in
    \Addr_T}$; 
 \item
  the output function is simply the concatenation of all the variables since 
  after reading the whole string only a unique address is alive,
  i.e. all the variables except the variable corresponding to that address must
  be empty, i.e. $F(q) = \prod_{X \in \varsst} X$; and
 \item
  the update function $\rho: \delta \to \Ss_{\varsst, \Gamma}$ is defined by 
  $\rho(\tau, a, p_{\tau'})(X_\alpha) {:=} \gamma X_{\alpha_1} \gamma_1 X_{\alpha_2}
  \gamma_2 \ldots  X_{\alpha_n} \gamma_{n}$ if $n\leq C$ and the following
  formula is valid (on strings, which is decidable):
 $ \forall x. \left(\Phi_{\tau, a, \tau'}(x) \to
    \Phi_{upd}[X_\alpha{:=}\gamma X_{\alpha_1} \gamma_1 X_{\alpha_2}
  \gamma_2 \ldots  X_{\alpha_n} \gamma_{n}](x)\right).
  $ 
 \end{itemize}

\subsection{Aperiodicity and 1-boundedness of SST-la} 
\label{sec:aperiodicSSTLA}
\begin{figure}
\begin{center}
\begin{tikzpicture}[->,>=stealth',shorten >=1pt,auto,scale=0.9, every node/.style={scale=0.9}]

\tikzstyle{graphnode}=[circle,fill=black,thick,inner
sep=0pt,minimum size=1.5mm]

\node [draw=none] (input) at (-1,0.1) {\begin{scriptsize}input string\end{scriptsize}} ;
\node [draw=none] (input) at (-1,-0.2) {\begin{scriptsize}$s$\end{scriptsize}} ;
\node [draw=none] (input) at (-1,-0.7) {\begin{scriptsize}output string\end{scriptsize}} ;
\node [draw=none] (input) at (-1,-1) {\begin{scriptsize}graph\end{scriptsize}} ;

\node [graphnode] (j1c) at (2,-2) {} ;

\node  [graphnode] (i'd) at (6,0) {} ;
\node  [graphnode] (id) at (4,0) {} ;

\node  (i') at (6,-0.3) {$i'$} ;
\node  (i) at (4,-0.3) {$i$} ;

\draw[|-|] (0,0) -- (7,0) ;

\node  (j1ct) at (2,-1.7) {$\begin{array}{c}\absaddr(s,i,\alpha) = \absaddr(s,i',\alpha')\end{array}$} ;

\node [graphnode] (tail) at (3,-3) {};

\draw [->,snake=snake] (j1c) -- (tail);

\draw [-,densely dotted] (i) -- (4,-3.5);
\draw [-,densely dotted] (i') -- (6,-3.5);

\fill[blue!20,fill opacity=0.2] (0,-0.5) rectangle (7,-3.5);

\node [graphnode] (head3) at (4.2, -2.7) {};
\node [graphnode] (tail3) at (5.7, -3.2) {};

\path[->] (tail) edge [bend right=20]  (head3) ;
\draw [->,snake=snake] (head3) -- (tail3);
\end{tikzpicture}
\begin{tikzpicture}[->,>=stealth',shorten >=1pt,auto]
\tikzstyle{graphnode}=[circle,fill=black,thick,inner
sep=0pt,minimum size=1.5mm]

\node [draw=none] (input) at (-1,0.1) {\begin{scriptsize}input string\end{scriptsize}} ;
\node [draw=none] (input) at (-1,-0.2) {\begin{scriptsize}$s$\end{scriptsize}} ;
\node [draw=none] (input) at (-1,-0.7) {\begin{scriptsize}output string\end{scriptsize}} ;
\node [draw=none] (input) at (-1,-1) {\begin{scriptsize}graph\end{scriptsize}} ;

\node  [graphnode] (i'd) at (6,0) {} ;
\node  [graphnode] (id) at (4,0) {} ;
\node  (i') at (6,-0.3) {$i'$} ;
\node  (i) at (4,-0.3) {$i$} ;
\draw[|-|] (0,0) -- (7,0) ;

\draw [-,densely dotted] (i) -- (4,-5.5);
\draw [-,densely dotted] (i') -- (6,-5.5);



\node [graphnode] (head5) at (4.5, -1) {} ;
\node [graphnode] (tail5) at (5.5, -2) {} ;
\draw [->,snake=snake] (head5) -- (tail5);

\node [graphnode] (head4) at (3.5, -2) {} ;
\node [graphnode] (tail4) at (1.5, -3) {} ;
 \draw [->,snake=snake] (head4) -- (tail4);

\node [graphnode] (head3) at (5.5, -3) {} ;
\node [graphnode] (tail3) at (4.5, -3.8) {} ;
\draw [->,snake=snake] (head3) -- (tail3);

\node [graphnode] (head2) at (2,-4) {} ;
\node [graphnode] (tail2) at (3,-5) {};
\draw [->,snake=snake] (head2) -- (tail2);

\node [graphnode] (head1) at (4.2, -4.7) {};
\node [graphnode] (tail1) at (5.7, -5.2) {};
\draw [->,snake=snake] (head1) -- (tail1);


\path[->] (tail3) edge [bend right=20]  (head2) ;
\path[->] (tail2) edge [bend right=20]  (head1) ;
\path[->] (tail5) edge  (head4) ;
\path[->] (tail4) edge (head3) ;


\node  at (1.2,-4) {$\begin{array}{c}\absaddr(s,i,\alpha)\end{array}$} ;
\node  at (4.5, -4.1) {$\ell^d$} ;
\node  at (5.5, -1) {$\absaddr(s,i',\alpha')$} ;

\fill[blue!20,fill opacity=0.2] (0,-0.5) rectangle (7,-5.5);

/\end{tikzpicture}
\end{center}
\vspace{-5mm}
\caption{Variable flow information is FO-definable.}
\label{combo}
\vspace{-5mm}
\end{figure}
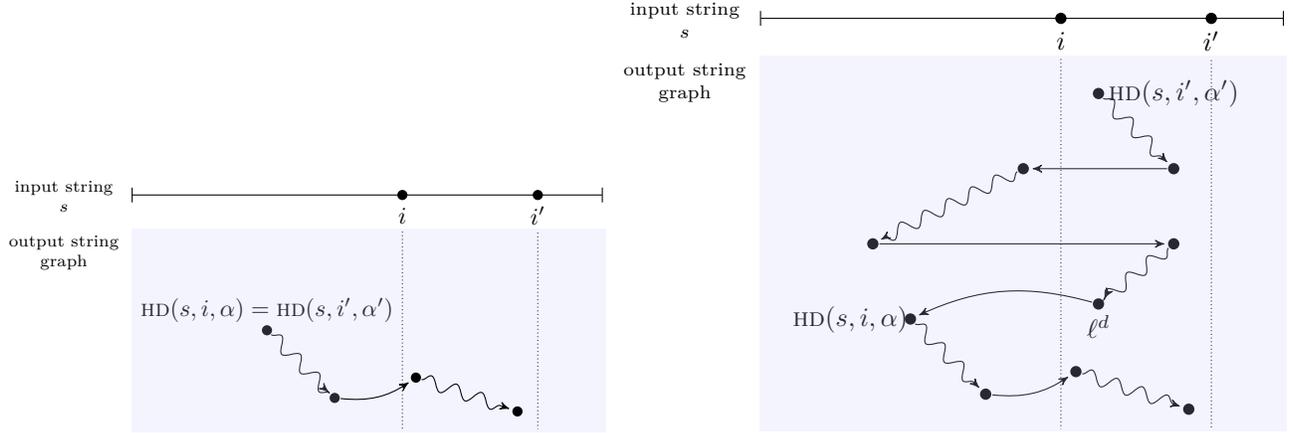

In this section, we first prove that the $\sst{}_\la$ $T_\la= (\Sigma,
\Gamma, Q, q_0, Q_f, \delta, \varsst, \rho, F)$ with look-ahead $A =
(Q_A,\Sigma, \delta_A, P_f)$ constructed in the previous section is
aperiodic and $1$-bounded, i.e., its transition monoid is 
aperiodic and $1$-bounded. Given a tuple $t =
(q,q',R,R',X_\addr,X_{\addr'},m)\in Q^2\times (2^{Q_A})^2\times
\varsst{}^2\times \mathbb{N}$, we show that the flow is FO-definable, i.e. there exists an FO-sentence $\text{flow}_t$ such that for all strings
$s\in\Sigma^*$, $s\models \text{flow}_t$ iff
$((q,R),X_\addr)\flows_{m}^s((q',R'),X_ {\addr'})$.
Then, aperiodicity of the transition monoid of $T$ will follow by Proposition~\ref{prop:ktype}.\ref{prop:ape}. 
Indeed, we know that there exists $n_0$ such that for all strings $s\in\Sigma^*$, $s^{n_0} \equiv_{b} s^{n_0+1}$ and therefore, 
$s^{n_0}\models \text{flow}_t$ iff $s^{n_0+1}\models \text{flow}_t$, where $b$ is quantifier rank of the formulas $\text{flow}_t$.  
We start with the following observation: for all strings $s\in
\dom(T_\la)$, there exists exactly one accepting run of $T_\la$ on $s$
(proved in Appendix \ref{onerun}).
We first prove a result on the state flow of $T_{la}$. 

\begin{lemma}\label{lem:stateflowb}(State Flow)
    Given two states $q,q'\in Q$, and two sets
    $R,R'\in 2^{Q_A}$. There exists an FO-formula
    $\text{sflow}_{q,q',R,R'}(x,y)$ of quantifier rank at most $k+3$ such
    that for all strings $s\in\dom(T_\la)$ of length $n\geq 1$ and any two positions
    $i,i'\in\dom(s)$, $s\models \text{sflow}_{q,q',R,R'}(i,i')$ iff
    $i\leq i'$ and the
    (unique accepting) run $r = (q_0,R_0)\dots (q_n,R_n)$ of $T_\la$ on $s$ satisfies
    $(q_{i-1},R_{i-1}) = (q,R)$ and $(q_{i'},R_{i'}) = (q',R')$.  
\end{lemma}
\begin{proof}[Proof (Sketch).]
By definition of $T_\la$ and its look-ahead automaton, we precisely
characterize the configurations $(q_j,R_j)$ in FO. 
For instance, the fact that the main run of
$T_{\la}$ is in $q$ at position $x$, by definition of $T_{\la}$,
is equivalent to say that the prefix up to $x$ has type $q$
(remind that $Q = \Theta_k$). It is expressible in FO by a formula
with one free variable $x$ obtained by guarding all quantifiers of
any variable $z$ in $\Phi_\tau$ by $z\preceq x$, where $\Phi_\tau$
has been defined in Prop.~\ref{prop:hintikka}. The full proof is in
Appendix \ref{stateflowb}. 
\end{proof}

The flow between variables is characterized by the following property.

\begin{lemma}(Variable Flow)
\label{var-add}
Let $X_{\addr},X_{\addr'}\in\varsst{}$ be two variables, $s\in
\dom(T_\la)$ a string of length $n\geq 1$ and $i\leq i'\in dom(s)$ two positions of $s$. 
Let $r = (q_0,R_0)\dots (q_n,R_n)$ be the accepting run of $T_\la$ on
$s$. Then $(q_{i-1},R_{i-1},X_{\addr})\flows^{s[i{:}i']}_m
(q_{i'},R_{i'},X_{\addr'})$ for some $m\geq 1$ iff %
    (1) $\absaddr(s,i,\addr)$ and $\absaddr(s,i',\addr')$ are both defined;
     (2) there is a path from $\absaddr(s,i',\addr')$ to
       $\absaddr(s,i,\addr)$ such that each node $(j'',c'')$ of this path
       is such that $j''\leq i'$. Formally, 
       if $\absaddr(s,i,\addr) {=} (j,c)$ and $\absaddr(s,i',\addr') {=} (j',c')$, 
       then $s\models \phi^{c',c}_\preceq(j',j)$ and, for all $c''\in C$ and 
       all $j''\leq |s|$, if  
       $s\models \phi^{c',c''}_\preceq(j',j'')\wedge
       \phi^{c'',c}_\preceq(j'',j)$, then $j''\leq i'$. 
Moreover, $(q_{i-1},R_{i-1},X_{\addr})\flows^{s[i{:}i']}_m
(q_{i'},R_{i'},X_{\addr'})$ for some $m\geq 1$ iff
$(q_{i-1},R_{i-1},X_{\addr})\flows^{s[i{:}i']}_1
(q_{i'},R_{i'},X_{\addr'})$.
\end{lemma}
\begin{proof}[Proof (Sketch).]
Suppose that all the conditions are met (the converse is proved similarly). 
Consider first the particular case  where $\absaddr(s,i,\addr) = \absaddr(s,i',\addr')$ depicted in left-side of Figure~\ref{combo}. 
It means that the output node $\absaddr(s,i,\addr)$ is both an $i$-head and an $i'$-head.
The name of this node however has changed to $\absaddr(s,i',\addr')$ at position $i'$, and possibly, the path represented by variable 
$X_{\addr}$ at position $i$ has been extended (as shown on the figure).
By construction of $T_\la$, variable $X_{\addr}$ at position $i$ flows into variable $X_{\addr'}$ at position $i'$ through the sequence of
variable updates $X_{\addr_j} := X_{\addr_{j-1}}$ for all $i\leq j\leq i'$ where $\addr_i = \addr$ and $\addr_{i'} = \addr'$, and for all 
$i\leq j \leq i'$, $a(\addr_j) = a_j$ (the $j$-th symbol of $s$), $\tau_1(\addr_j) = \tau_1(\addr_{j-1}).[a_j]_k$, $\tau_2(\addr_j) =
[s[j..i']]_k$ and $c(\addr_j) = c(\addr)$.

The other case is when the node $\absaddr(s,i,\addr)$ is the target of an edge from some (alive) node $\ell^d$ such that $i < \ell \leq i'$, i.e.,  
$\absaddr(s,i,\addr)$ is an $i$-head but is not an $i'$-head.
This new position  $\ell^d$ belongs to some path that never goes beyond position $i'$, and the $i'$-head of this path is represented, by
construction of $T_\la$, by some variable. 
If this variable is precisely $X_{\addr'}$, then one gets that
$X_\addr$ at position $i$ flows into $X_{\addr'}$ at position $i'$, by construction of variable update in $T_\la$.  It is depicted in right-side of Figure~\ref{combo}.
On the figure, the path from node $\absaddr(s,i',\addr')$ contains node $\absaddr(s,i,\addr)$. 
Therefore the content of variable $X_{\addr'}$ at position $i'$
depends on the content of variable $X_{\addr}$ at position $i$. 

From this characterization of variable flow, it is easy to see that a
variable cannot flow multiple times to another variable, since there
exists only one path from $\absaddr(s,i',\addr')$ to
$\absaddr(s,i,\addr)$.
\end{proof}

Based on the two previous lemmas, we are now able to express the
``relative'' flow of states and variables of $T_\la$ in between two
positions of a string $s\in dom(T_\la)$ in FO.

\begin{lemma}\label{lem:relativeflow}(Relative State-Variable Flow)
Given a tuple $t = (q,q',R,R',X_\addr,X_\addr',m)\in Q\times Q\times
2^{Q_A}\times 2^{Q_A}\times \varsst\times \varsst\times \mathbb{N}$, there exists an FO-formula
$\text{rflow}_t(x,y)$ of quantifier rank at most $k+4$ such that for all strings
$s\in\dom(T_\la)$ of length $n\geq 1$ and any two positions $i\leq i'\in dom(s)$,  if $r =
(q_0,R_0)\dots (q_n,R_n)$ is the accepting run of $T_\la$ on $s$, then 
$s\models \text{rflow}_t(i,i')$ iff
$(q_{i-1},R_{i-1},X_\addr)\flows^{s[i{:}i']}_m (q_{i'},R_{i'},X_{\addr'})$ for
some $m\geq 1$. Moreover, $(q_{i-1},R_{i-1},X_\addr)\flows^{s[i{:}i']}_m (q_{i'},R_{i'},X_{\addr'})$ for
some $m\geq 1$ iff $(q_{i-1},R_{i-1},X_\addr)\flows^{s[i{:}i']}_1 (q_{i'},R_{i'},X_{\addr'})$.
\end{lemma}

\begin{proof}
    We express the conditions of Lemma \ref{var-add} in FO and take
    the resulting formula in conjunction with the formula 
    $\text{sflow}_{q,q',R,R'}(x,y)$ obtained from
    Lemma \ref{lem:stateflowb}. The full proof is in Appendix
    \ref{relativeflow}. 
\end{proof}

The formulas $\text{rflow}_t(x,y)$ for tuples $t =
(q,q',R,R',X_\addr,X_\addr',m)$ describe the flow between two
positions $x$ and $y$ of some string $s\in \dom(T_\la)$, with respect to the unique
run of $T_{la}$ on $s$. However to prove aperiodicity of the
transition monoid of $T_{la}$, one has to express the flow on a whole
string $s$ (which is not necessarily in $\dom(T_\la)$), and this flow
must only depend on the starting and ending configurations $(q,R)$ and
$(q',R')$ resp. In particular, $(q,R,X_\addr)$ flows
to $(q',R',X_\addr')$ on $s$ is not equivalent to $s\models
\text{rflow}_t(1,n)$ where $n=|s|$, because the
run of $T_{la}$ on $s$ may not start with $(q,R)$. However, the flow
of an SST with look-ahead is defined between useful configurations only,
i.e. configurations which are both accessible from an initial state and
co-accessible (a final state is accessible from them). Thanks to this
requirement, we are able to express the flow on a string by using
$\text{rflow}_t(x,y)$. This formula is first transformed into an
aperiodic automaton that runs on strings extended with boolean values
that indicate the positions of $x$ and $y$. Then we take the quotient
of this automaton to define the set of substrings from position $x$ to
position $y$ and project the boolean values away. All these steps
preserve aperiodicity. A proof of Lemma \ref{lem:flowFO} can be found in Appendix \ref{app:flowFO}.

\begin{lemma}\label{lem:flowFO}
Given a tuple $t = (q,q',R,R',X_\addr,X_\addr',m)\in Q\times Q\times
2^{Q_A}\times 2^{Q_A}\times \varsst\times \varsst\times\mathbb{N}$ such that $(q,R)$ and $(q',R')$ are
both useful, there exists an FO-sentence
$\text{flow}_t$ of quantifier rank at most $k+4$ such that for all strings
$s\in\Sigma^*$ and any two positions $i<i'\in dom(s)$,  $s\models
\text{flow}_t$ iff $(q,R,X_\addr)\flows_{m}^s (q',R',X_{\addr'})$ for
some $m\geq 1$. Moreover, $(q,R,X_\addr)\flows_{m}^s (q',R',X_{\addr'})$ for
some $m\geq 1$ iff $(q,R,X_\addr)\flows_{1}^s (q',R',X_{\addr'})$.
\end{lemma}

\begin{proof}[Sketch of Proof]
    The proof of this result is based on automata. The formula
    $\text{rflow}_t(x,y)$ is transformed into an aperiodic automaton $A_1$ that runs
    on strings extended with Boolean values that indicate the
    positions $x$ and $y$. This automaton can be modified into 
    an automaton $A_2$ that accepts only factors of strings $s$ accepted by
    $A_1$ from
    position $x$ to position $y$, while preserving aperiodicity. 
    The automaton $A_2$ is then projected on alphabet $\Sigma$, getting
    an aperiodic automaton $A_3$. Then the sentence $\flow_t$ is
    defined as an FO-sentence equivalent to $A_3$. Usefulness of
    $(q,R)$ and $(q',R')$ is needed to ensure that $(q,R,X_\addr)$
    $s$-flows in $(q',R',X_{\addr'})$ implies $s\models \flow_t$. 
    Indeed, in that case, there exist two strings $s_1,s_2$ and the accepting run of 
    $T_{la}$ on $s_1ss_2$ reaches $(q,R)$ after reading $s_1$ and
    $(q',R')$ after reading $s_1s$ and therefore, $s_1ss_2\models
    \text{rflow}_t(|s_1|+1, |s_1|+|s|)$ by Lemma
    \ref{lem:relativeflow}, from which we can prove that $s\models
    \flow_t$. 
\end{proof}

\begin{corollary}\label{coro:aperiodicity-sstla}
    The $\sst_\la$ $T_\la$ is aperiodic and 1-bounded. 
\end{corollary}

\newpage
\appendix
\section{Proofs from Section  \ref{fo:prop}}

\subsection{Proof of Proposition \ref{prop:ktype}.\ref{prop:indis}}

\begin{proof}
    We prove the proposition for formulas with two free variables. The
    case of one free variable is a particular case. The proof is based on the
    composition result of Proposition \ref{prop:ktype}.\ref{prop:compo}. 
    
    Let $s = s_1as_2bs_3$ and $s' = s'_1as'_2bs'_3$. 
Considering the extended alphabet $\Sigma'=\Sigma \times \{0,1\}^2$,  we define the string 
$u=u_1\tiny{\left( \begin{array}{c} a \\ 1 \\ 0  \end{array} \right)}
u_2 \tiny{\left( \begin{array}{c} b \\ 0 \\ 1  \end{array} \right)} u_3$ where 
$u_i \in 
\{\tiny{\left( \begin{array}{c} c \\ 0 \\ 0  \end{array} \right)}   \mid c \in \Sigma\}^*$. $u$ is an extension of $s$ (hence 
$u_i$ is an extension of $s_i$).    
 The two extra bits serve as the interpretation of first order variables $x,y$ with 
$x=1$ at position $i_1$ and $y=1$ at position $i_2$. In a similar manner, we define 
$u'$ as well as $u'_i$ as  extensions of $s'$ and $s'_i$ respectively.  

Since $s_i\equiv_{k+2} s'_i$ for all $i\in\{1,2,3\}$, 
we obtain $u_i\equiv_{k+2} u'_i$ by extending 
the signature of FO to $\Sigma\times \{0,1\}^2$.  
 Therefore by Proposition~\ref{prop:ktype}.\ref{prop:compo}, we get $u \equiv_{k+2} u'$. 
Replacing every atomic formula $L_\gamma(z)$ of $\phi(x,y)$  by 
$\bigvee_{m,n\in\{0,1\}} L_{\tiny{\left(\begin{array}{c} \gamma\\ m\\n\end{array} \right)}}(z)$, 
we obtain the formula $\phi'(x,y)$. Quantifying $x,y$ we obtain the sentence 
    $\psi_{xy} = \exists x\exists y \bigvee_{c,d\in\Sigma}
    L_{\tiny{\left( \begin{array}{c} c \\ 1 \\ 0 \end{array}\right)}}(x) \wedge L_{\tiny{\left( \begin{array}{c} d\\ 0 \\ 1 \end{array}\right)}}(y)\wedge \phi'(x,y)$.
   It can be easily checked that   $s\models \phi(i_1,i_2)$ iff $u\models
    \psi_{xy}$ and $s'\models \phi(i'_1,i'_2)$ iff $u'\models
    \psi_{xy}$. Since the quantifier rank of $\phi$ is at most $k$, the quantifier rank of 
    $\psi_{xy}$ is at most $k+2$. Since $u\equiv_{k+2} u'$, we get  
    $s\models \phi(i_1,i_2)$ iff $u\models \psi_{xy}$ iff  $u'\models \psi_{xy}$ iff $s'\models
    \phi(i'_1,i'_2)$. 
\end{proof}

\section{Proofs from Section \ref{sec:ap}}

\subsection{Proof of Lemma  \ref{lem:complexity}}
\label{complexity}
\begin{lemma}
Given an SST $T$, checking whether its transition monoid $M_T$ is 1-bounded is in \textsc{PSPACE}.
\end{lemma}
\begin{proof}
Let $T=(\Sigma,\Gamma,Q,q_0,Q_f,\delta,\varsst,\rho,F)$ be an SST. 
To check if $T$ is 1-bounded, we have to check that there does {\it not} exist a string $s=s_1s_2$ 
having a run from some state $p$ to state $q$ 
such that 
\begin{itemize}
\item There is  a run on string $s_1$ from  state $p$ to state $r$, such that variable $X$ 
1-flows into variables $Y,Z$.  
\item There is a run on string $s_2$ from state $r$ to state $q$ such that, variables $Y,Z$ 1-flow into variable $G$
\end{itemize}
Clearly, if the above situation happens, $X$ 2-flows into variable $G$, and 
 $M_{s_1}[(p,X)][(r,Y)]=1=M_{s_1}[(p,X)][(r,Z)]$,
$M_{s_2}[(r,Y)][(q,G)]=1=M_{s_2}[(r,Z)][(q,G)]$, and hence 
  $M_s[(p,X)][(q,G)]=2$, which means  $T$ is not 1-bounded.
 We give below, the algorithm to check if  $T$ is 1-bounded.
 \begin{enumerate}
 \item Successively guess the symbols of two strings $s_1$ and $s_2$
and along the way, keep computing the transition matrices $M_{s_1}$ and $M_{s_2}$. 
 This is possible to be done in PSPACE.
 \item Compute $M_{s_1} \times M_{s_2}$ and check if it contains an integer $i \geq 2$. If so, 
 then as discussed above, there is a variable $X$ that $i$-flows into some variable $G$.  
  \end{enumerate}
  Clearly, the overall complexity of this algorithm is NPSPACE. 
Thanks to Savitch's Theorem, we have a PSPACE algorithm.
    \end{proof}

\begin{lemma}
Checking whether a given 1-bounded SST is aperiodic is \textsc{PSPACE-complete}.
\end{lemma}
\begin{proof}
Given an SST $T=(\Sigma,\Gamma,Q,q_0,Q_f,\delta,\varsst,\rho,F)$, we first construct 
an automaton $A_T$ such that the transition monoids of $T$ and $A_T$ are the same. 
 By definition, $T$ is aperiodic iff 
its transition monoid $M_T$  is aperiodic. 
  It is known \cite{cho-dung} that a deterministic (not necessarily minimal) 
  finite state automaton is non-aperiodic iff there is some string $u \in \Sigma^*$ 
  with the ``non-trivial cycle property''.  
  We cannot directly apply this result of \cite{cho-dung} to $A_T$ since 
  in general, $A_T$ could be non-deterministic. However, we show that 
  $M_T$ is non-aperiodic iff there exists a non-trivial cycle in $A_T$ 
  (note that this result is in general not true for arbitrary automata: for instance, one can have an automaton 
  $A$ accepting the aperiodic language $(ab)^*$; however the  transition monoid of $A$ could be non-aperiodic).

   Given an automaton, we explain what the ``non-trivial cycle property'' means: 
            There is a string $u$ 
    and a state $p$ such $p \notin \delta(p,u)$, and for some positive
    integer $r$, $p \in \delta(p,u^r)$.  In this proof, we show that
    $M_{A_T}$ is non-aperiodic  iff  there is a string  $u \in \Sigma^*$ that has the ``non-trival cycle property''    
  in $A_T$. 

 First, we explain the construction of $A_T$ from $T$.  Given $T$, $A_T$ is constructed as 
 $(Q \times \varsst, \Sigma, \delta_A, q_0 \times \varsst, Q_f \times \varsst)$ where 
 $\delta((p,X),a)= \{(q,Y) \mid$ there is a transition from $p$ to $q$ on $a$, such that  
  on the variable update on this transition, $X$ flows to $Y\}$.       
 Corresponding to one transition from $p$ to $q$ on $a$ in $T$, we have the transitions  
 from $(p,X_i)$ to $(q,X_j)$ on $a$ in $A_T$, whenever variable $X_j$ is updated 
  and $X_i$ flows into $X_j$ on that update. It is easy to see that the transition monoids of $T,A_T$ are same.

  Suppose now that there exists a non-trivial cycle in $A_T$. Then,  there exists a string $u$
and a state $(p,X)$ and $m \geq 0$ such that
 $M_u[(p,X)][(p,X)] = 0$, and 
 $M_{u^m}[(p,X)][(p,X)] = 1$. 
We want to show that $M_T$ is {\it not} aperiodic. That is,  there exists some string $v$ such that, 
for all $k$,  $M_{v^k} \neq M_{v^{k+1}}$. 

Let $k\geq 0$. We show that $M_{u^k}[(p,X)][(p,X)] = 1$ implies
$M_{u^{k+1}}[(p,X)][(p,X)] = 0.$

\begin{enumerate}
\item  
If $k = 0$, this is trivially true, since  $M_{\epsilon}[(p,X)][(p,X)] = 1$
and $M_u[(p,X)][(p,X)] = 0$. 
\item 
If $k > 0$,  then assume that $M_{u^k}[(p,X)][(p,X)] = 1$. We show that
$M_{u^{k+1}}[(p,X)][(p,X)] = 0$. Therefore suppose that $M_{u^{k+1}}[(p,X)][(p,X)] =
1$ and we will arrive at a contradiction. 

By assumption, $M_u[(p,X)][(p,X)] = 0$. Since we also assume 
$M_{u^{k+1}}[(p,X)][(p,X)] =1$, it
is necessarily the case that $M_{u^k}[(p,X)][(q,Y)] = 1$ and $M_u[(q,Y)][(p,X)] = 1$
for some $(q,Y) \neq (p,X)$. Since the underlying SST $T$ is deterministic and $p$ is reachable from
$p$ on $u^k$ (since $M_{u^k}[(p,X)][(p,X)] = 1$), we necessarily have that $q = p$.
Therefore $X \neq Y$. Now, we have the situation depicted in Figure \ref{case}. 
Clearly, this  contradicts the 1-boundedness of the SST. Therefore, we get 
$M_{u^{k+1}}[(p,X)][(p,X)] = 0$. 
\end{enumerate}
\begin{figure}[h]
\begin{center}
\begin{tikzpicture}[->,>=stealth',shorten >=1pt,auto,scale=0.7]

\tikzstyle{alivenode}=[circle,fill=black!80,thick,inner
sep=0pt,minimum size=2mm]

\tikzstyle{deadnode}=[circle,fill=black!10,thick,inner
sep=0pt,minimum size=2mm]

\node  (y1) at (0,-1) {$(p,X)$} ;
\node  (y4) at (4,0) {$(p,X)$} ;
\node  (y5) at (4,-2) {$(p,Y)$} ;
\node  (y6) at (7,-2) {$(p,X)$} ;
\node  (y11) at (11,-1) {$(p,X)$} ;

  \draw[->,snake=snake]   (y1) -- node[above] {$u^k$} (y4) ;
  \draw[->,snake=snake]   (y1) -- node[above] {$u^k$} (y5) ;
  \draw[->,snake=snake]   (y4) -- node[above] {$u^{k+1}$} (y11) ;
  \draw[->,snake=snake]   (y5) -- node[above] {$u$} (y6) ;
  \draw[->,snake=snake]   (y6) -- node[above] {$u^{k}$} (y11) ;


\fill[blue!20,fill opacity=0.2] (-0.7,0.3) rectangle (12.3,-2.3);

\end{tikzpicture}
\end{center}
\caption{\label{case} Multiple paths in $A_T$: $X$ flows into $X$ and $Y$ on $u^k$; 
further on $u^{k+1}$, $X$ flows into $X$, and $Y$ flows into $X$.
This gives $(p,X) \rightsquigarrow^{w}_2 (p,X)$, for $w=u^{2k+1}$, contradicting 1-boundedness. 
 }
\end{figure}
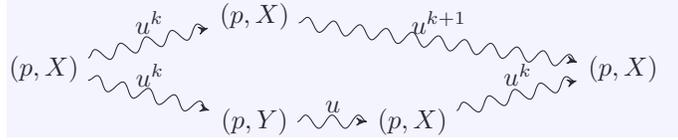

 It cannot be the case that $T$ is aperiodic : If it were, 
then there exists $m_0$ such that
for all $n\geq m_0$ we have $M_{u^n} = M_{u^{n+1}}$. 
We know that $M_{u^m}[(p,X)][(p,X)] = 1$, therefore $M_{u^{i.m}}[(p,X)][(p,X)] = 1$
for all $i$. Take $i$ such that
$i.m \geq  m_0$.  Then $M_{u^{i.m}}[(p,X)][(p,X)] = M_{u^{i.m+1}}[(p,X)][(p,X)] = 1$. 
  This however, contradicts what we just showed i.e,  
   $M_{u^k}[(p,X)][(p,X)] = 1 \Rightarrow 
M_{u^{k+1}}[(p,X)][(p,X)] = 0$.\\
  
  Conversely, assume that $M_T$ is not aperiodic.  
  Then there is a string $u$ such that for all $m$, 
  $M_{u^m} \neq M_{u^{m+1}}$. 
  We show the
existence of a non-trivial cycle in $A_T$.

Assume now that, for all states $(p,X)$, and for all $m \geq 1$, and all strings $u$,   
$M_{u^m}[(p,X)][(p,X)]=1$ iff $M_{u}[(p,X)][(p,X)]=1$.  Note that 
this is the same as saying that all strings $u$ 
give rise only to trivial cycles. We will arrive at a contradiction to this assumption. 
By non-aperiodicity, we can  pick some large $m$ for which 
  $M_{u^m} \neq M_{u^{m+1}}$. Then there are states 
$(p,X)$ and $(q,Y)$ such that 
$M_{u^m}[(p,X)][(q,Y)] \neq M_{u^{m+1}}[(p,X)][(q,Y)] $. 

\begin{enumerate}
\item 
Without loss of generality, assume 
$M_{u^m}[(p,X)][(q,Y)]=1$. If we take $m > |Q|.|\varsst|$, 
then on the run of $u^m$ from $(p,X)$ to $(q,Y)$ in $A_T$, we will revisit a state $(r,Z)$ more than once. Assume that 
the run is such that $(p,X) \rightsquigarrow^{u^{l_1}} (r,Z) \rightsquigarrow^{u^{l_2}} 
(r,Z) \rightsquigarrow^{u^{l_3}} (q,Y)$ where $l_1+l_2+l_3=m$. 
By our assumption on ``only trivial cycles'', we know that 
$M_{u}[(r,Z)][(r,Z)]=1$ since $M_{u^{l_2}}[(r,Z)][(r,Z)]=1$. Hence, we 
also have the run  
$(p,X) \rightsquigarrow^{u^{l_1}} (r,Z) \rightsquigarrow^{u^{l_2+1}} 
(r,Z) \rightsquigarrow^{u^{l_3}} (q,Y)$ in $A_T$. This gives 
$M_{u^{m+1}}[(p,X)][(q,Y)]=1$, contradicting 
our assumption of $M_{u^m}[(p,X)][(q,Y)] \neq M_{u^{m+1}}[(p,X)][(q,Y)] $.
\item 
Consider the case 
$M_{u^m}[(p,X)][(q,Y)]=0$.  We now consider the run 
in $A_T$ from $(p,X)$ to $(q,Y)$ on $u^{m+1}$, where $(r,Z)$ is revisited on 
$u^{l}$ for some $l>0$.  Again, the ``only trivial cycles'' assumption 
then gives us a run on $u^m$ from $(p,X)$ to $(q,Y)$ contradicting  
$M_{u^m}[(p,X)][(q,Y)] \neq M_{u^{m+1}}[(p,X)][(q,Y)]$. 
\end{enumerate}  
  
    Thus, we have shown that $M_T$ is aperiodic iff all strings satisfy the trivial cycle property in $A_T$. It
    remains now to check the existence of a string $u$ having the
    non-trivial cycle property in $A_T$. Adapting Stern's algorithm \cite{stern} 
    to non-deterministic automata,  we show that checking the existence 
    of a string $u$ having the non-trivial cycle property can be done in PSPACE.  
    Briefly, we successively guess the symbols of a string $u$ and compute the transition matrix 
    of $u$. Next, we guess a state $(p,X)$. From the transition matrix of $u$, we can check if 
     $M_{u}[(p,X)][(p,X)]=0$. If so, we guess an integer $r \leq |Q \times \varsst|$  
and compute $M_{u^r}$. If $M_{u^r}[(p,X)][(p,X)]=1$, then we have found a non-trivial cycle. 
Using the PSPACE-hardness of checking non-trivial cycles in \cite{cho-dung}, we conclude that 
checking aperiodicity of SSTs is PSPACE-complete. 
  \end{proof}                                

\section{Proofs from Section \ref{fo:varflow}}
\label{app:varflow}
\subsection{Proof of Proposition \ref{prop:foflow}}

First, we show that states of accepting runs of aperiodic \sst{} are
FO-definable:

\begin{proposition}
\label{prop:fostates}
Let $T$ be an aperiodic \sst{} $T$. For all states $q$, there exists an
FO-formula $\phi_q(x)$ such that for all strings $s\in \Sigma^+$, for
all positions $i$, $s\models \phi_q(i)$ iff $s\in dom(T)$ and the state of the (unique)
accepting run of $T$ before reading the $i$-th symbol of $s$ is $q$. 
There exists an FO-sentence $\phi_q^{last}$ that defines
the last state of the accepting run of $T$ on $s$ (if it exists). 
\end{proposition}

\begin{proof}
Let $A$ be the underlying (deterministic) automaton of $T$. Since $T$
is aperiodic, so is $A$. For all
$q$, let $L_q$ be the set of strings $s$ such that there exists a run of $T$ on
$s$ that ends in $q$. Clearly, $L_q$ can be defined by some aperiodic
automaton $A_q$ obtained by setting the set of final states of $A$ to
$\{q\}$. Therefore $L_q$ is definable by some FO-formula $\psi^L_q$. 
Let $R_q$ be the set of strings $s$ such that there exists a run of
$T$ on $s$ from $q$ to some accepting state. Clearly, $u\in dom(T)$
iff there exists $q\in Q$, $v\in L_q$ and $w\in R_q$ such that $u=vw$.
The language $R_q$ is also definable by the aperiodic automaton
obtained by setting the initial state of $A$ to $q$, and therefore is
definable by some FO-formula $\psi^R_q$.

Then, $\phi_q(x)$ is defined as
$$
\phi_q(x) = [\psi^L_q]_{\prec x} \wedge [\psi^R_q]_{x\preceq}
$$
where $[\psi^L_q]_{\prec x}$ is the formula $\psi^L_q$ in which 
all quantifications of any variable $y$ is guarded by $y\prec x$ and,
similarly, $[\psi^R_q]_{x\preceq}$ is the formula $\psi^R_q$ is which
all quantifications of any variable $y$ is guarded by $x\preceq y$.
Therefore, $s\models
\phi_q(i)$ iff $s[1{:}i)\in L_q$ and $s[i{:}|s|]\in R_q$.  %

The formula $\phi_q^{\text{last}}$ is constructed similarly.
%
\end{proof}
Now we start the proof of Proposition \ref{prop:foflow}.    
\begin{proof}
    For all states $p,q\in Q$, 
    let $L_{(p,X)\flows (q,Y)}$ be the language of strings $u$ such that $(p,X)\flows_1^{u}
    (q,Y)$. We show that $L_{(p,X)\flows (q,Y)}$ is an aperiodic
    language. It is indeed
    definable by an aperiodic non-deterministic automaton $A$ that keeps track of flow information
    when reading $u$. It is constructed from $T$ as follows.
    Its state set $Q'$ are pairs $(r,Z)\in 2^{Q\times \mathcal{X}}$. Its initial state is
    $\{(p,X)\}$ and final states are all states $P$ such that
    $(q,Y)\in P$. There exists a transition $P\xrightarrow{a} P'$ in
    $A$ iff for all $(p_2,X_2)\in P'$, there exists $(p_1,X_1)\in P$
    and a transition $p_1\xrightarrow{a|\rho} p_2$ in $T$ such that
    $\rho(X_2)$ contains an occurrence of $X_1$. Note that by
    definition of $A$, there exists a run from a state $P$ to a state
    $P'$ on some $s\in\Sigma^*$ iff for all $(p_2,X_2)\in P'$, there
    exists $(p_1,X_1)\in P$ such that $(p_1,X_1)\flows^{s}_1
    (p_2,X_2)$  (Remark $\star$).

    Clearly, $L(A) = L_{(p,X)\flows (q,Y)}$. It remains to show that $A$ is
    aperiodic, i.e. its transition monoid $M_A$ is aperiodic.
    Since $T$ is aperiodic, there exists $m\geq 0$ such
    that for all matrices $M\in M_T$, $M^m = M^{m+1}$. 
    For $s\in\Sigma^*$, let $\Phi_A(s) \in M_A$ (resp. $\Phi_T(s)$) 
    the square matrix of dimension
    $|Q'|$ (resp. $|Q|$) associated with $s$ in $M_A$ (resp. in $M_T$). We show that $\Phi_A(s^m) =
    \Phi_A(s^{m+1})$, i.e. $(P,P')\in \Phi_A(s^m)$ iff
    $(P,P')\in\Phi_A(s^{m+1})$, for all $P,P'\in Q'$.

    First, suppose that $(P,P')\in \Phi_A(s^m)$, and let
    $(p_2,X_2)\in P'$.  By definition of $A$, there exists 
    $(p_1,X_1)\in P$ such that $(p_1,X_1)\flows^{s^m}_1 (p_2,X_2)$, and 
    by aperiodicity of $T$, it implies that
    $(p_1,X_1)\flows^{s^{m+1}}_1 (p_2,X_2)$. Since it is true for all
    $(p_2,X_2)\in P'$, it implies by Remark $(\star)$ that there exists a run 
    of $A$ from $P$ to $P'$ on $s^{m+1}$, i.e. $(P,P')\in
    \Phi_A(s^{m+1})$. The converse is proved similarly.

    We have just proved that $L_{(p,X)\flows (q,Y)}$ is aperiodic. Therefore it is
    definable by some FO-formula $\phi_{(p,X)\flows (q,Y)}$. 
    Now, $\phi_{X\flows Y}(x,y)$ is defined by
    $$
    \phi_{X\flows Y}(x,y)\equiv x\preceq y \wedge
    \bigvee_{p,q\in Q}\{ [\phi_{(p,X)\flows (q,Y)}]^{x\preceq
      \cdot\preceq y} \wedge \phi_p(x)\wedge
    ((\text{last}(y)\rightarrow \phi_q^{\text{last}})\wedge
    (\neg\text{last}(y)\rightarrow \bigvee_{r \in Q} \phi_r(y+1)))\},
    $$
    where $\phi_p$, $\phi_r$ and $\phi_q^{\text{last}}$ were defined in Proposition
    \ref{prop:fostates} and $[\phi_{(p,X)\flows (q,Y)}]^{x\preceq \cdot\preceq y}$ is
    obtained from $\phi_{(p,X)\flows (q,Y)}$ by guarding all the
    quantifications of any variable $z$ by $x\preceq z\preceq y$. 
\end{proof}

\subsection{Proof of Proposition \ref{prop:contribution}}
\begin{proof}
The formula $\contribute_X(x)$ is defined by
$$
\begin{array}{llllllll}
\contribute_X(x) & = & \exists y\cdot \text{last}(y)\wedge \bigwedge_{p\in
  Q, q\in
  Q_f} (\Phi_q^{last}\rightarrow 
 \bigvee_{Y\in F(q)} \Phi_p(x)\wedge
\Phi_{X\flows Y}(x,y))
\end{array}
$$
where $\text{last}(y)$ defines the last position of the string,
$\Phi_p(x)$ is defined in proposition \ref{prop:fostates} and 
$\Phi_{X\flows Y}(x,y)$ in proposition \ref{prop:foflow}.
\end{proof}

\subsection{Definition of \sst-output graphs}
\label{app:sst-output}
Let $T = (Q, q_0, \Sigma, \Gamma, \varsst, \delta, \rho, Q_f)$ be an
\sst{}.  Let $u\in (\Gamma\cup X)^*$ and $s\in
\Gamma^*$. The string $s$ is said to \emph{occur} in $u$ if $s$ is a
factor of $u$. In particular, $\epsilon$ occurs in $u$ for all $u$. 
Let $O_T$ be the set of constant strings
occurring in variable updates, i.e. $O_T = \{ s\in\Gamma^*\ |\ \exists
t\in \delta,\ s\text{ occurs in } \rho(t)\}$. Note that $O_T$ is
finite since $\delta$ is finite.

Let $s\in dom(T)$. The \emph{\sst-output graph} of $s$ by $T$, denoted by
$G_T(s)$, is defined as a directed graph whose edges are labelled by
elements of $O_T$. Formally, it is the graph $G_T(s) =
(V,(E_\gamma)_{\gamma\in O_T})$ where 
$V = \{0,1,\dots,|w|\}\times \mathcal{X}\times \{in,out\}$ is the set of
vertices, $E := \bigcup_{\gamma\in O_T} E_\gamma \subseteq V\times V$
is the set of labelled edges defined as follows.

Vertices $(i,X,d)\in V$ are denoted by $(X^d,i)$. 
Let $n=|s|$ and $r = q_0\dots q_n$ the
accepting run of $T$ on $s$. The set $E$ is defined as the smallest
set such that for all $X\in \varsst$, 
\begin{enumerate}
\item $((X^{in},0),(X^{out},0))\in E_\epsilon$ if $(X,0)$ is useful,

\item for all $i<n$ and $X\in X$, if $(X,i)$ is useful and
if $\rho(q_{i},s[i+1],q_{i+1})(X) = \gamma$, then 
$((X^{in},i+1),(X^{out},i+1))\in E_\gamma$,

\item for all $i<n$ and $X\in X$, if $(X,i)$ is useful 
and if $\rho(q_{i},s[i+1],q_{i+1})(X) = \gamma_1X_1\dots
\gamma_kX_{k}\gamma_{k+1}$ (with $k>1$), then 
\begin{itemize}

\item $((X^{in},i+1), (X_1^{in},i))\in E_{\gamma_1}$
\item $((X_k^{out},i), (X^{out},i+1))\in E_{\gamma_{k+1}}$
\item for all $1\leq j< k$, $((X_j^{out},i), (X_{j+1}^{in},i))\in E_{\gamma_{j+1}}$

\end{itemize}      

\end{enumerate}

Note that since the transition monoid of $T$ is $1$-bounded, it is
never the case that two copies of some variable (say $X$) flows into 
some variable (say $Y$), therefore this graph is
well-defined and there is \textbf{no} multiple edges between two
nodes.

\section{Proofs fom Section \ref{sst2fot}}

\begin{proposition}\label{prop:uniquepath}
    $G_T(s)$ consists of a unique directed path. Moreover, the concatenation of edge labels occurring
    along this path equals $T(s)$.
\end{proposition}

\subsection{Proof of Lemma \ref{lem:fopath}}
\label{app:sst2fot}
\begin{proof}
    For all variables $X,Y\in \varsst$, we denote by $C_{X,Y}$ the set of 
    pairs $(p,q,a)\in Q^2\times \Sigma$ such that there exists a transition from $p$
    to $q$ on $a$ whose variable update concatenate $X$ and $Y$ (in
    this order). We first define a formula for condition $(3)$:
$$
\Psi_3^{X,Y}(x,y) \ \equiv \ \exists z\cdot x\preceq z
\wedge y\preceq z 
\wedge \bigvee_{X',Y'\in\varsst, (p,q,a)\in C_{X',Y'}}[ \\
L_a(z)\wedge 
\phi_{X\flows X'}(x,z) \wedge 
\phi_{Y\flows Y'}(y,z)  \wedge \phi_p(z)\wedge \phi_q(z+1)]
$$

Then, formula $\text{path}_{X,Y,d,d'}(x,y)$ is defined by
$$
\begin{array}{llllllllll}
\text{path}_{X,Y,in,in}(x,y) & \equiv & 
    \phi_{Y\flows X}(y,x) \vee \Psi_3^{X,Y} \\
\text{path}_{X,Y,in,out}(x,y) & \equiv & 
    \phi_{Y\flows X}(y,x) \vee \phi_{X\flows Y}(x,y) \vee \Psi_3^{X,Y} \\
\text{path}_{X,Y,out,in}(x,y) & \equiv & 
    \text{false} \\
\text{path}_{X,Y,out,out}(x,y) & \equiv & 
    \phi_{X\flows Y}(x,y) \vee \Psi_3^{X,Y} \\
\end{array}
$$
\end{proof}

\subsection{Proof of Lemma \ref{lem:sst2fo}}

    We show here that the transformation which associates 
    a string $s$ with its \sst-output graph $G_T(s)$ is FO-definable
    whenever $T$ is aperiodic and 1-bounded, based on Lemma \ref{lem:fopath} and 
    the construction of  \cite{FiliotTrivediLics12,AC10}. The
    idea of \cite{FiliotTrivediLics12,AC10} is to define
    the accepting runs of $T$ by using set variables, as for classical 
    automata-to-MSO transformations, and to use state information
    in order to determine which variable updates apply and then define
    the edge relations. There is a copy of the domain for each 
    variable $x$ and each $d\in\{in,out\}$. 
    Since states, variable flow and paths are all
    FO-definable when $T$ is aperiodic and 1-bounded, it follows that $G_T$ is FOT-definable. We refer
    the reader to \cite{FiliotTrivediLics12,AC10} for more
    details, but we recall here that the domain formula $\phi_{\dom}$
    is a sentence defining the domain of $T$, and therefore in our
    case is FO-definable, since $\dom(T)$ is aperiodic. To illustrate
    the construction, we also give
    the formula $\phi_{E_\gamma}^{X^{in},
      X^{in}}(y,x)$ that defines the $\gamma$-labelled edge
    relation for the domain copy $X^{in}$. It is defined by
    $$
    y = x+1\wedge \bigvee_{t := (p,a,q)\in \delta,
      \rho(t)(X) = \gamma X \beta...} L_a(y)\wedge \phi_p(x)\wedge
    \phi_q(y)\wedge \contribute_X(y)
    $$
    where $\phi_p$ and $\phi_q$ are FO-formulas defined in proposition
    \ref{prop:fostates} and $\contribute_X(y)$ has been defined in
    Proposition \ref{prop:contribution}.

    Thanks to Lemma \ref{lem:fopath}, the transitive closure between
    some copy $X^d$ and some copy $Y^{d'}$ is
    defined by the FO-formula
    $$
    \phi_{\preceq}^{X^d, Y^{d'}}(x,y)\ \equiv\ \text{path}_{X,Y,d,d'}(x,y)
    $$
    \qed

\section{Proofs from Section \ref{sec:sst-la}}
\label{app:sst-la}
We first define the transition monoid of an $\sstla$ $(T,A)$ where
$A = (Q_A, \Sigma, \delta_A, P_f)$ is a deterministic lookahead automaton
and $T = (\Sigma, \Gamma, Q, q_0, Q_f, \delta, \varsst, \rho, F)$.

\textbf{Uniqueness of accepting runs} Let $s=s_1\dots s_n\in\Sigma^*$
and $r : (q_0,P_0)\xrightarrow{s_1} (q_1,P_1)\dots
(q_{n-1},P_{n-1})\xrightarrow{s_n} (q_n,P_n)$ be an accepting run of
$(T,A)$ on $s$. We not only show that $r$ is unique, but that the
sequence of transitions associated with $r$ is unique. Given a
sequence of transitions of $T$, it is clear that there exists exactly
one run associated with that sequence, since $A$ is deterministic. 
 
Suppose the
sequence of transitions is not unique, i.e. there exists another
accepting run on $r$ which follows another transition of $T$
eventually. Let $i\geq 1$ be the
smallest index where the $i$-th transitions are different on both
runs. Before taking the $i$-th transition, both runs are
in the configuration $(q_{i-1}, P_{i-1})$. Suppose that the $i$-th
transition on the first run is $(q_{i-1}, a, p, q_i)$ for some
look-ahead state $p$, and is $(q_{i-1}, a, p', q'_i)$ on the other
run, for some state $q'_i$ and look-ahead state $p'$ such that either 
$p\neq p'$ or $q_i\neq q'_i$. Since both runs are accepting, the
suffix $s_{i+1}\dots s_n$ is in $L(A_p)\cap L(A_{p'})$, which is
impossible by the mutual-exclusiveness of look-aheads. Therefore $p =
p'$, but in that case, $q_i = q'_i$ since $\delta$ is a function. This
leads to a contradiction.

\textbf{Variable Flow and Transition Monoid for \sstla{}}. 
Let $Q_A$ represent the states of the (deterministic) lookahead automaton $A$, and $Q$ denote states of the 
  \sstla{}.

The transition monoid of an \sst{} with look-ahead depends on its configurations
and variables. 
It extends the notion of transition monoid for \sst{} with look-ahead states
components but is defined only on {\it useful} configurations $(q,P)$. 
A configuration $(q,P)$ is {\it useful} iff it is {\it accessible} and 
{\it co-accessible} : that is, $(q,P)$ is reachable from the initial
configuration $(q_0,\varnothing)$ and some accepting configuration $(q_f,P)\in Q_f\times 2^{P_f}$ is reachable from $(q,P)$.

Note that given two useful configurations $(q,P)$, $(q',P')$ and a
string $s\in\Sigma^*$, there exists at most one run from $(q,P)$ to
$(q',P')$ on $s$. Indeed, since $(q,P)$ and $(q',P')$ are both useful, 
there exists $s_1,s_2\in\Sigma^*$ such that 
$(q_0,\varnothing)\flows^{s_1} (q,P)$ and $(q',P') \flows^{s_2}
(q_f,R_f)$ where $(q_f,R_f)$ is accepting. If there are two runs from
$(q,P)$ to $(q',P')$ on $s$, then we can construct two accepting runs
on $s_1ss_2$, which contradicts the fact that accepting runs are
unique. We can even strengthen this result by showing that the
sequence of transitions associated with the unique run from $(q,P)$
to $(q',P')$ on $s$ is as well unique. We denote by 
$\textsf{useful}(T,A)$ the useful configurations of $(T,A)$.

Thanks to the uniqueness of the sequence of transitions associated
with the run of an \sstla{} from and to useful
configurations on a given string, one can extend the notion of
variable flow naturally by considering, as for \sst{}, the composition
of the variable updates along the run.

A string $s\in \Sigma^*$ maps to a square  matrix $M_s$ of dimension
$|Q\times 2^{Q_A}|\cdot |\varsst|$ and is defined by  $M_{s}[(q, P), X][(q',P'),
X']=n$ if there exists a run $r$ from $(q,P)$ to
$(q',P')$ on $s$ such that $n$ copies of $X$ flows to $X'$ over the
run $r$, and $(q,P)$ and $(q',P')$ are both useful  (which implies
that the sequence of transitions of $r$ from $(q,P)$ to $(q',P')$ is
unique, as seen before), otherwise $M_{s}[(q, P), X][(q',P'), X']=\bot$.

\subsection{Proof of  Lemma \ref{lem:aperiodicSSTLA}}
\label{sec:sstla2sst}

\begin{proof}
Let $(T,A)$ be an $\sst_\la$, with  $A = (Q_A, \Sigma, \delta_A, P_f)$  a deterministic lookahead automaton, \\
and $T = (\Sigma, \Gamma, Q, q_0, Q_f, \delta, \varsst, \rho,
F)$. Without loss of generality, we make the following assumption
$$
\text{\textbf{Assumption} } \star:\ \forall q,q',q''\in Q,~\forall p,p'\in
Q_A,~\forall a\in\Sigma,~\quad p\neq p'\wedge \delta(q,a,p)=q'\wedge
\delta(q,a,p')=q''\implies q'\neq q''
$$
This is indeed wlog: if $(T,A)$ does not satisfy this assumption, then
we can have as many copies of states $Q$ as states of $Q_A$ (i.e. the new set of
states of $T$ is $Q\times Q_A$) and transform the transitions
accordingly to maintain uniqueness of the successor states w.r.t. to
input symbols and look-ahead states. Moreover, it is easy to show that
this transformation preserves aperiodicity.

\vspace{2mm}
\textbf{Construction of $T'$} We construct an aperiodic and
1-bounded \sst{} $T'$ equivalent to $T$.
As explained in definition \ref{sst-LA}, the unique  run of a string $s$ on $(T,A)$ is not only a sequence of $Q$-states, but also a collection of the look ahead states $2^{Q_A}$.
At any time, the current state of $Q$, and collection of look-ahead states $P \subseteq Q_A$ is a configuration. 
A configuration $(q_1,P_1)$, on reading $a$, evolves into 
$(q_2,P_2 \cup\{p_2\})$, where $\delta(q_1,a,p_2)=q_2$ is a transition  in the $\sstla$ and  
$\delta_A(P_1,a)=P_2$, where $\delta_A$ is the transition function of the look ahead automaton $A$.
Note that the transition monoid of the $\sstla$ is aperiodic and 1-bounded by assumption. We now show 
how to remove the look-ahead, resulting in an equivalent $\sst{}$ $T'$ whose transition monoid 
is aperiodic and 1-bounded. 

While defining $T'$, we ``collect'' together all 
the states resulting from transitions of the form $(q,a,p,q')$ and $(q,a,p',q'')$ in the $\sstla$.  
We define $T'=(\Sigma, \Gamma,Q',q'_0, \delta', \varsst', \rho',Q_f')$
with: 
\begin{itemize}
\item $Q' = 2^{\textsf{useful}(T,A)}$ where $\textsf{useful}(T,A)$ are the
useful configurations of $(T,A)$ ($\textsf{useful}(T,A)$ is computable in exponential
time from $(T,A)$),

\item $q'_0=\{(q_0,\emptyset)\}$ (wlog we assume that $(T,A)$ accepts
  at least one input therefore $(q_0,\emptyset)$ is useful),

\item $Q'_f$, the set of accepting states, is defined by 
$\{ S\in Q'\ |\ \exists (q,P)\in S,\ q\in Q_f\wedge P\subseteq
P_f\}$. 

\item $\varsst'=\{X_{q'} \mid X \in \varsst, q' \in
  \textsf{useful}(T,A)\}$,

\item   The transitions are defined as follows:
$\delta'(S,a)=\bigcup_{(q,P)\in S}\Delta((q,P),a)$ where \\
$\Delta((q,P),a)=\{(q',P' \cup\{p'\}) \mid (q,a,p',q') \in \delta$ and  
$\delta_A(P,a)=P'\}\cap \useful(T,A)$.

\end{itemize}

Before defining the update function, we first assume a total ordering
$\preceq_{\useful(T,A)}$ on $\useful(T,A)$. 
For all $(p,P) \in Q \times 2^{Q_A}$, we define the substitution $\sigma_{(p,P)}$ as  $X \in
\varsst \mapsto X_{(p,P)}$.  Let $(S,a,S')$ be a transition of $T'$. Given a
state $(q',P') \in S'$, there might be several predecessor states $(q_1,P_1),\dots,(q_k,P_k)$
in $S$ on reading $a$. The set $\{(q_1,P_k),\dots,(q_k,P_k)\} \subseteq S$ is denoted by 
$Pre_S((q',P'),a)$. Formally, it is defined by $\{ (q,P)\in S\ |\
(q',P')\in \Delta((q,P),a)\}$.

We consider only the variable update of the
transition from the  minimal predecessor state.
Indeed, since any string has at most one accepting run in the $\sstla$
$(T,A)$ (and at most one associated sequence of transitions), if two runs reach the same state at
some point, they will anyway define the same output and therefore we
can drop one of the variable update, as shown in
\cite{FiliotTrivediLics12}. Formally, the variable update
$\rho'(S,a,S')(X_{(q',P')})$, for all $X_{(q',P')}\in\varsst'$ is
defined by $\epsilon$ if $(q',P') \notin S'$, and by
$\sigma_{(q,P)}\circ \rho(q,a,p,q')(X)$, where $(q,P) = \text{min}\ \{
(r,R)\in S\ |\ (q',P')\in \Delta((r,R),a)\}$, and $\delta(q,a,p) = q'$
(by Assumption $\star$ the look-ahead state $p$ is unique). 
It is shown in \cite{FiliotTrivediLics12} that indeed $T'$ is
equivalent to $T$. We show here that the transition monoid of $T'$ is aperiodic and
$1$-bounded.

For all $S\in Q'$, let us define $\Delta^*(S,s)=\{ (q',P') \mid \exists
(q,P) \in S$ such that $(q,P) \rightsquigarrow^s_{T,A}
(q',P')\}\cap\useful(T,A)$.

\vspace{2mm}
\textbf{Claim} Let $M_{T'}$ be the transition monoid of $T'$ and
$M_{T,A}$ the transition monoid of $(T,A)$. Let
$S_1,S_2\in Q'$, $X_{q,P},Y_{q',P'}\in \varsst'$ and $s\in\Sigma^*$. Then
one has $M_{T',s}[S_1,X_{(q,P)}][S_2,Y_{(q',P')}]=i\geq 0$
iff $S_2 = \Delta^*(S_1,s)$ and one of the 
following hold:
\begin{enumerate}
\item either $i=0$ and, $(q,P)\not\in S_1$ or $(q',P')\not\in S_2$, or

\item $(q,P) \in S_1$, $(q',P')\in S_2$, $(q,P)$ is the minimal
  ancestor in $S_1$ of $(q',P')$ (i.e. $(q,P) = \text{min}\ \{ (r,R)\in S_1\ |\ (q',P')\in
  \Delta^*((r,R),s)\}$), and $M_{(T,A),s}[(q,P),X][(q',P'),Y] = i$. 
\end{enumerate}

\vspace{2mm}
\textit{Proof of Claim.} It is easily shown that 
$M_{T',s}[S_1,X_{(q,P)}][S_2,Y_{(q',P')}]\geq 0$ iff $S_2 =
\Delta^*(S_1,s)$. Let us show the two other conditions. Assume that 
$M_{T',s}[S_1,X_{(q,P)}][S_2,Y_{(q',P')}] = i\geq 0$. The variable
update function is defined in such a way that after reading $s$ from
$S_1$, all the variables $Z_{(r,R)}$ such that $(r,R)\not\in S_2$ have
just been reset to $\epsilon$ (and therefore no variable can flow from
$S_1$ to them). In particular, if  $(q',P')\not\in S_2$, then no
variable can flow in $Y_{(q',P')}$ and $i=0$.

Now, assume that $(q',P')\in S_2$, and consider the sequence of states 
$S_1,S'_1,S'_2,\dots,S'_k,S_2$ of $T'$ on reading $s$. By definition
of the variable update, the variables that are used to update
$Y_{(q',P')}$ on reading the last symbol of $s$ from $S'_k$ 
are copies of the form $Z_{(r,R)}$ such that $(r,R)$ is the minimal
predecessor in $S'_k$ of $(q',P')$ (by $\Delta$). By induction, it is
easily shown that if some variable $Z_{(r,R)}$ flows to $Y_{(q',P')}$ 
from $S_1$ to $S_2$ on reading $s$, then $(r,R)$ is necessarily the
minimal ancestor (by $\Delta^*$) of $(q',P')$ on reading $s$. 
In particular if $(q,P)\not\in S_1$, then $i=0$. 

Finally, if $i>0$, then necessarily $(q,P)$ is the minimal ancestor in
$S_1$ of $(q',P')$ on reading $s$, from $S_1$ to $S_2$, and since $T'$
mimics the variable update of $(T,A)$ on the copies, we get 
that $M_{(T,A),s}[(q,P),X][(q',P'),Y]=i$.

The converse is shown similarly.  \hfill \textit{End of Proof
  of Claim}.

\vspace{2mm}
\textbf{1-boundedness and aperiodicity of $\mathbf{T'}$} 
1-boundedness is an obvious consequence of the claim and the fact that $(T,A)$ is
$1$-bounded. Let us show that $M_{T'}$ is aperiodic. We know that
$M_{T,A}$ is aperiodic. Therefore there exists $n\in\mathbb{N}$ such
that for all strings $s\in\Sigma^*$,  $M_{(T,A),s}^n =
M_{(T,A),s}^{n+1}$.

Let us first show that for all $S_1,S_2\in Q'$, and all strings $s\in\Sigma^*$,
$\Delta^*(S_1,s^n) = S_2$ iff $\Delta^*(S_1,s^{n+1}) = S_2$.
  Indeed, 
  \begin{itemize}
\item $S_2 = \Delta^*(S_1,s^n)$, iff 
$S_2 = \{ (q',P')\in\textsf{useful}(T,A)\ |\ \exists (q,P)\in S_1,\
(q,P)\flows^{s^n}_{T,A},(q',P')\}$, iff
\item $S_2 = \{ (q',P')\in\textsf{useful}(T,A)\ |\ \exists (q,P)\in S_1,\
M_{(T,A),s}^n[(q,P),X][(q',P'),Y]\geq 0\text{ for some
}X,Y\in\varsst\}$, iff
\item by aperiodicity of $M_{T,A}$, $S_2 = \{ (q',P')\in\textsf{useful}(T,A)\ |\ \exists (q,P)\in S_1,\
M_{(T,A),s}^{n+1}[(q,P),X][(q',P'),Y]\geq 0 \\ \text{ for some
}X,Y\in\varsst\}$, iff 
\item $S_2 = \Delta^*(S_1,s^{n+1})$.
\end{itemize}

Let $S_1,S_2\in Q'$ and $X_{(q,P)}, Y_{(q',P')}\in \varsst$. Let also
$s\in\Sigma^*$. We study condition $(1)$ of the claim and show that 

\begin{quote}
$M_{T',s}^n[S_1,X_{(q,P)}][S_2, Y_{(q',P')}] = i$ and condition $(1)$
of the claim holds, iff 
$M_{T',s}^{n+1}[S_1,X_{(q,P)}][S_2,
Y_{(q',P')}] = i$ and condition $(1)$ of the claim holds. 
\end{quote}

\begin{itemize}
\item Indeed, $M_{T',s}^n[S_1,X_{(q,P)}][S_2, Y_{(q',P')}] = 0$ and,
$(q,P)\not\in S_1$ or $(q',P')\not\in S_2$ iff (by the claim) 
$\Delta^*(S_1,s^n) = S_2$, and $(q,P)\not\in S_1$ or $(q',P')\not\in
S_2$, iff by what we just showed, $\Delta^*(S_1,s^{n+1}) = S_2$, and $(q,P)\not\in S_1$ or $(q',P')\not\in
S_2$, iff (by the claim) $M_{T',s}^{n+1}[S_1,X_{(q,P)}][S_2,
Y_{(q',P')}] = 0$ and condition $(1)$ of the claim holds. 
\end{itemize}

Let us now treat condition $(2)$ of the claim, and  show that 
\begin{quote}
$M_{T',s}^{n}[S_1,X_{(q,P)}][S_2,
Y_{(q',P')}] = i$ and condition $(2)$ of the claim holds, iff 
$M_{T',s}^{n+1}[S_1,X_{(q,P)}][S_2, Y_{(q',P')}] = i$ and condition
$(2)$ of the claim holds.
\end{quote}

\begin{itemize}
\item   We only show one direction, the other being
proved exactly similarly. Suppose that
$M_{T',s}^{n}[S_1,X_{(q,P)}][S_2, Y_{(q',P')}] = i$ and $(q,P)\in
S_1$, $(q',P')\in S_2$, and $(q,P)$ is the minimal ancestor in $S_1$
of $(q',P')$, and $M_{(T,A),s}^n[(q,P),X][(q',P'),Y] = i$. It implies,
by the claim, that $\Delta^*(S_1,s^n) = S_2$, and therefore 
$\Delta^*(S_1,s^{n+1}) = S_2$. Now, we have
$(q,P) = \text{min}\ \{ (r,R)\in S_1\ |\ (q',P')\in
\Delta^*((r,R),s^n)\}$. Since $\Delta^*((r,R),s^n) =
\Delta^*((r,R),s^{n+1})$ for all $(r,R)\in S_1$, we have
$(q,P) = \text{min}\ \{ (r,R)\in S_1\ |\ (q',P')\in
\Delta^*((r,R),s^{n+1})\}$. Finally, $M_{T,A,s}^{n+1}[q,P,X][q',P',Y] = 
M_{T,A,s}^{n}[q,P,X][q',P',Y] = i$ (by aperiodicity of $(T,A)$). 
By the claim, it implies that $M_{T',s}^{n+1}[S_1,X_{(q,P)}][S_2,
Y_{(q',P')}] = i$ and condition $(2)$ of the claim is satisfied. 
\end{itemize}

Since by the claim we can be only in case $(1)$ or $(2)$, it implies
that $M_{T'}$ is aperiodic.

\end{proof}

\section{Proofs from Section \ref{fot:sstla}}

\subsection{Proof of Lemma \ref{lem:addr}}
\label{app:fot-sstla}
\begin{proof}
Intuitively, if $j_1\neq
j_2$ and $j_1^c$, $j_2^c$ are both $i$-heads, then the string $s$
can be decomposed as in the following figure:
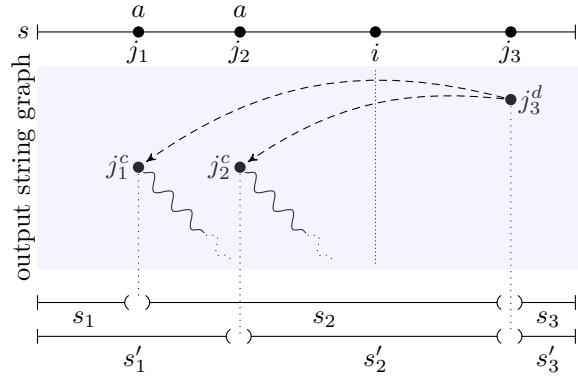
\begin{center}
\begin{tikzpicture}[->,>=stealth',shorten >=1pt,auto,scale=0.9]

\tikzstyle{graphnode}=[circle,fill=black,thick,inner
sep=0pt,minimum size=1.5mm]

\node [draw=none] (input) at (-0.2,0) {$s$} ;

\node [graphnode] (j1) at (1.5,0) {} ;
\node [graphnode] (j2) at (3,0) {} ;
\node [graphnode] (i) at (5,0) {} ;
\node [graphnode] (j3) at (7,0) {} ;

\node  (j1) at (1.5,-0.3) {$j_1$} ;
\node  (j2) at (3,-0.3) {$j_2$} ;
\node  (i) at (5,-0.3) {$i$} ;
\node  (j3) at (7,-0.3) {$j_3$} ;

\node  (j1) at (1.5,0.3) {$a$} ;
\node  (j2) at (3,0.3) {$a$} ;

\draw[|-|] (0,0) -- (8,0) ;

\node [graphnode] (j1c) at (1.5,-2) {} ;
\node [graphnode] (j2c) at (3,-2) {} ;
\node [graphnode] (j3d) at (7,-1) {} ;

\node  (j1ct) at (1.2,-2) {$j_1^c$} ;
\node  (j2ct) at (2.7,-2) {$j_2^c$} ;
\node  (j3dt) at (7.3,-1) {$j_3^d$} ;

\draw [-,snake=snake] (j1c) -- (2.5,-3);
\draw [-,snake=snake] (j2c) -- (4,-3);
\draw [-,dotted,snake=snake] (2.5,-3) -- (2.9,-3.4);
\draw [-,dotted,snake=snake] (4,-3) -- (4.4, -3.4);

\draw [-,densely dotted] (i) -- (5,-3.5);

\draw [-,dotted] (j1c) -- (1.5,-4) ;

\draw [-,dotted] (j2c) -- (3,-4.5) ;

\draw [-,dotted] (j3d) -- (7,-4.5) ;

\draw [|-(] (0,-4) -- node[below] {$s_1$} (1.4, -4) ;

\draw [)-(] (1.6,-4) -- node[below] {$s_2$} (6.9, -4) ;

\draw [)-|] (7.1,-4) -- node[below] {$s_3$} (8, -4) ;

\draw [|-(] (0,-4.5) -- node[below] {$s'_1$} (2.9, -4.5) ;

\draw [)-(] (3.1,-4.5) -- node[below] {$s'_2$} (6.9, -4.5) ;

\draw [)-|] (7.1,-4.5) -- node[below] {$s'_3$} (8, -4.5) ;

\fill[blue!20,fill opacity=0.2] (0,-0.5) rectangle (8,-3.5);

\node [rotate=90] at (-0.2, -2) {output string graph};

\path[densely dashed,->] (j3d) edge [bend right=20]  (j2c) ;
\path[densely dashed,->] (j3d) edge [bend right=25]  (j1c) ;

\end{tikzpicture}
\end{center}

Since the output is a string, there is necessarily some edge from
a position $j_3^d$ such that $j_3>i$, to $j_1^c$ or $j_2^c$. It can be
easily shown that  the
existence of such an edge is FO-definable by a formula with two-free
variables of quantifier rank at most
     $k$. Since the two decompositions are indistinguishable by
formulas with two-free variables of quantifier rank at most $k$, by Proposition
\ref{prop:ktype}.\ref{prop:indis}, one gets that an edge from $j_3^c$ to the other
considered $i$-head also exist, which contradicts the fact that
the output is a string.

    We formally prove the result now. Suppose that there exist $j_1\neq j_2$ that both satisfy the
    preconditions and suppose that $j_1^c$ and $j_2^c$ are both
    $i$-heads. We exhibit a contradiction.

    By definition of $i$-heads, $j_1^c$ and $j_2^c$ are
    alive, and therefore both contribute to the output $T(s)$. 
    Since $T(s)$ is a string (i.e. a unique directed path), there is
    necessarily some incoming edge 
    to $j_1^c$ or $j_2^c$ in $T(s)$, say $j_1^c$. Formally, there
    exists a position $j_3$ and a copy $d\in C$ such that 
    $(j_3^d, j_1^c)$ is an edge of $T(s)$, i.e. $s\models
    \phi^{d,c}_\suc(j_3,j_1)$. Since $j_1^c$ is an $i$-head, 
    it is necessarily the case that $j_3>i$. We claim that 
    $s\models \phi^{d,c}_\suc(j_3,j_2)$, i.e. there exists an edge in
    $T(s)$ from $j_3^d$ to $j_2^c$, which contradicts the fact that
    $T(s)$ is a string.

    Indeed, let decompose the input string $s$ as 
    $$
    \begin{array}{llllllllllllll}
    s_1 & = & s[1{:}j_1) &  \quad &   s'_1 & = & s[1{:}j_2) \\
 s_2 & =& s(j_1{:}j_3) & \quad & s'_2 & =& s(j_2{:}j_3) \\
 s_3 & =&  s(j_3{:}|s|] & \quad &  s'_3 & =& s(j_3{:}|s|] \\
    \end{array}
    $$

    We show that the conditions of Proposition~\ref{prop:ktype}.\ref{prop:indis} are
    satisfied by this decomposition.  Clearly, $s = s_1s[j_1] s_2 s[j_3] s_3 = s'_1 s[j_2] s'_2 s[j_3]
    s'_3$. Moreover, $s_1 \equiv_{k+2} s'_1$ by hypothesis, and, $s_3
    \equiv_{k+2} s'_3$ since $s_3 = s'_3$.  We also have
    $s_2 = s(j_1{:}i]s(i{:}j_3)$ and $s'_2 =
    s(j_2{:}i]s(i{:}j_3)$ and by hypothesis, 
    $s(j_1{:}i] \equiv_{k+2} s(j_2{:}i]$. Hence,  by Proposition 
    \ref{prop:ktype}.\ref{prop:compo} one gets $s \equiv_{k+2} s'$. 
   Since $s \models \phi^{d,c}_\suc(j_3,j_1)$, 
  and  $s \equiv_{k+2} s'$,  using 
      Proposition~\ref{prop:ktype}.\ref{prop:indis} we get 
      $s\models \phi^{d,c}_\suc(j_3,j_1)$ iff 
    $s \models \phi^{d,c}_\suc(j_3,j_2)$ (Recall  that by definition of
    quantifier rank $k$ of $T$, $\phi^{d,c}_\suc$ has quantifier rank at most $k$). Since $s\models
    \phi^{d,c}_\suc(j_3,j_1)$, one gets $s\models
    \phi^{d,c}_\suc(j_2,j_1)$, which leads to the contradiction
    mentioned earlier. 
    The proof is the same when assuming that $j_1^c$ and $j_2^c$ are both $i$-tails.
\end{proof}

\subsection{Proof of Lemma \ref{lem:addrfo}}
\label{app:addrfo}
\begin{proof}
    Let $\alpha\in\Addr_T$ and $c\in C$. Let us first prove the Lemma
    for the heads. Let $x,y$ be two variables (intended to capture positions $i$ and
    $j$ respectively).

    The condition that $\tau_1(\alpha) = \ktype{s[1{:}j)}{k+2}$ can be
    expressed, thanks to Proposition \ref{prop:hintikka},  by the formula 
  $\phi_2(x,y)$ of quantifier rank at most $k+2$ obtained by guarding all the quantifications of any
  variable $z$ in $\phi_{\tau_1(\addr)}$ by $z\prec y$.

  The condition $a(\addr) = s[j]$ is expressed by the formula
  $\phi_3(y) = L_{a(\addr)}(y)$.

  The condition $\tau_2(\addr) = \ktype{s(j{:}i]}{k+2}$ is defined, again
  by using Proposition
  \ref{prop:hintikka}, by the formula $\phi_4(x,y)$ of quantifier rank
  at most $k+2$, obtained by
  guarding all the quantifications of any variable $z$ in
  $\phi_{\tau_2(\addr)}$ by $y\prec z\preceq x$.

  Finally, the formula $\Phi_{\absaddr(\addr)}^c(x,y)$ is defined by

  $$
  \Phi_{\absaddr(\addr)}^c(x,y)\ \equiv\ \text{head}_c(x,y)\wedge
  \phi_2(x,y)\wedge \phi_3(y)\wedge \phi_4(x,y)
  $$

  The formula $\text{head}_c(x,y)$ has quantifier rank at most $k+2$,
  therefore $\Phi_{\absaddr(\addr)}^c(x,y)$ has quantifier rank at
  most $k+2$.

  The formula $\Phi_{\tailaddr(\addr)}^c(x, y)$ is defined by

  $$
  \Phi_{\tailaddr(\addr)}^c(x, y)\ \equiv\ \exists z.\bigvee_{c'\in C}
  \Phi_{\absaddr(\addr)}^{c'}(x,z)\wedge \phi_{\preceq}^{c', c}(z,
  y)\wedge \forall z'>x.\neg \bigvee_{c''\in C}(\phi_{\preceq}^{c', c''}(z,z')\wedge \phi_{\preceq}^{c'', c'}(z',y))
  $$

  This formula has quantifier rank at most $k+3$. 
\end{proof}

\section{Proofs from Section \ref{sec:aperiodicSSTLA}}

\subsection{$T_\la$ admits exactly one accepting runs per string $s\in
dom(T_\la)$}
\label{onerun}
\begin{proof}
    For any two transitions $(\tau_1, a, p_\tau, \tau'_1), (\tau_2, a,
    p_{\tau'}, \tau'_2)$ of $T_\la$, if $\tau_1 = \tau_2$, then on suffix
    $u\in \Sigma^*$, at most one of the two transitions can be
    triggered, because $u$ cannot satisfy both types $\tau$ and
    $\tau'$, since $k$-types partition  $\Sigma^*$.
\end{proof}

\subsection{Proof of Lemma \ref{lem:stateflowb}}
\label{stateflowb}

\begin{proof}
First, recall that the look-ahead automaton has transitions of the
form $(\tau,\tau')\xrightarrow{a} (\tau, \tau'.\ktype{a}{k+2})$ and
accepting state of the form $(\tau,\tau)$ for all $(k+2)$-types $\tau$.

Since we assume that $s\in\dom(T_\la)$, given an integer
$j\in\{0,\dots,n\}$, we can precisely define the
$j$-th configuration $(q_j,P_j)$ of the unique accepting run $r$ of
$T_\la$ on $s$. By definition of $T_\la$ and its look-ahead automaton,
we indeed have:

\begin{enumerate}
  \item $q_j = \ktype{s[1{:}j)]}{k+2}$ (recall that $q_j$ is a
    $k+2$-type)

  \item $R_j = \{ (\ktype{s[{\ell{+}1{:}n}]}{k+2},\ktype{s[\ell{+}1{:}j]}{k+2})\
    |\ 1\leq \ell\leq j\}$
    \end{enumerate}

Notice that $q_0$ is indeed equal to $\ktype{\epsilon}{k+2} =
\ktype{s[1{:}0]}{k+2}$ and $R_0 = \varnothing = 
\{ (\ktype{s[{\ell{+}1{:}n}]}{k+2},\ktype{s[\ell{+}1{:}j]}{k+2})\
    |\ 1\leq \ell\leq 0\}$. Let us express  equalities 1. and 2. in FO.

We construct a formula $\Phi_{q,R}(x)$ such that $s\models
\Phi_{q,R}(j)$ iff $(q,R) = (q_j,R_j)$, for all positions
$j\in\dom(s)$. It is defined by:

$$
\Phi_{q,R}(x)\ \equiv\ \Phi_q^{\cdot\preceq x}(x)\wedge \Phi_R(x)
$$

where $\Phi_q^{\cdot\preceq x}(x)$ expresses the fact that the prefix up
to position $x$ has type $q$, and is obtained by 
guarding all the quantifiers of $\Phi_q$ (the Hintikka formula
corresponding to type $q$, see Proposition \ref{prop:hintikka})
by $\preceq x$. The formula $\Phi_R(x)$ expresses the fact that the
look-ahead states after reading position $x$ are $R$:

The second property is expressed by the conjunction of the two
following formulas $\Phi_R^2(x)$ and $\Phi_R^3(x)$, where
$$
\Phi_R^2(x) = \forall z\cdot (1\preceq z\preceq x\rightarrow
\bigvee_{(\tau,\tau')\in R} \Phi_\tau^{z\prec\cdot}\wedge
\Phi_{\tau'}^{z\prec \cdot \preceq x})
$$

$$
\Phi_R^3(x) = \bigwedge_{(\tau,\tau')\in R} \exists z\cdot 1\preceq
z\preceq x \wedge \Phi_\tau^{z\prec \cdot}\wedge
\Phi_{\tau'}^{z\prec \cdot \preceq x}
$$

where the superscript $z\prec \cdot$ and $z\prec \cdot\preceq x$
indicates the guards applied to the quantifiers.

Finally, the formula $\text{sflow}_{q,q',R,R'}(x,y)$ is defined by
distinguishing among the cases $x=y=1$, $x=1\prec y$ and $1\prec
x\preceq y$:
$$
\begin{array}{llllllll}
\text{sflow}_{q,q',R,R'}(x,y)\ \equiv\  x\preceq y\wedge \\
(x=y=1\ \wedge\ \Psi_{q=q_0,R=\varnothing}\wedge \Psi_{q'=q_0,R'=\varnothing}) \vee
(x=1\prec y\ \wedge\ \Psi_{q=q_0,R=\varnothing}\wedge \Phi_{q',R'}(y)) \vee 
(x>1 \wedge \Phi_{q,R}(x-1)\wedge \Phi_{q',R'}(y))
\end{array}
$$
where $\Psi_{q=q_0,R=\varnothing}\equiv \top$ if $q=q_0$ and
$R=\varnothing$, otherwise $\perp$, and similarly for
$\Psi_{q'=q_0,R'=\varnothing}$.
The formula $\text{sflow}_{q,q',R,R'}(x,y)$
has a quantifier rank at most $k+3$. 
\end{proof}

\subsection{Proof of Lemma \ref{lem:relativeflow}}
\label{relativeflow}
\begin{proof}
We define two different formulas, depending on whether $m=0$ or $m\geq
1$. 

Suppose first that $m\geq 1$. 
We show how to define the formula $\text{rflow}_{t}(x,y)$ in
  FO by expressing the conditions of Lemma \ref{var-add} and taking
  the resulting formula in conjunction with the formula
  $\text{sflow}_{q,q',R,R'}(x,y)$ obtained from Lemma
  \ref{lem:stateflowb}. One uses two free variables $x'$ and $y'$ to extract
  the $x-$ and $y-$heads corresponding to adresses $\alpha$ and
  $\alpha'$, thanks to Lemma \ref{lem:addrfo}. The whole formula is
  defined by:

$$
\text{rflow}_{t}(x,y)\ \equiv\
\text{sflow}_{q,q',R,R'}(x,y)\wedge \bigvee_{c,c'\in C} \exists x'\exists y'\
\Phi_{\absaddr(\addr)}^c(x,x')\wedge
\Phi_{\absaddr(\addr')}^{c'}(y,y')\wedge 
\Phi_\preceq^{c',c}(y',x')\wedge$$
$$~~~~~~~~~~~~~~~~\neg \exists z[ z>y'\wedge
\bigvee_{c''\in C} \Phi_\preceq^{c',c''}(y',z)\wedge \Phi_{\preceq}^{c'',c}(z,x')]
$$
This formula has quantifier rank at most $k+4$.

If $m=0$, then the formula $\text{rflow}_t(x,y)$ is obtained by
taking the conjunction of the negation of the previous formula with
the formula $\text{sflow}_{q,q',R,R'}(x,y)$. 
\end{proof}


\subsection{Proof of Lemma \ref{lem:flowFO}}
\label{app:flowFO}
 
\begin{proof}
    We have to distinguish two cases, depending on whether $|s|=0$ or
    $|s|>0$. For these two cases, we construct two formula $\flow_t^0$
    and $\flow_t^{>0}$, and then define $\flow_t$ by
    $$
    \flow_t\ \equiv\ \flow_t^0 \wedge \flow_t^{>0}
    $$

    In case $|s|=0$, it should be true that $q=q'$, $R=R'$, 
    $X_\alpha = X_{\alpha'}$ and $m=1$. It is defined by the formula
    $$
    \text{flow}_t^0\equiv (\neg \exists x.\top)\rightarrow B_{q=q',R=R',\alpha=\alpha',m=1}
    $$
    where $B_{q=q',R=R',\alpha=\alpha',m=1} = \top$ if indeed $q=q'$,
    $R=R'$, $\alpha=\alpha'$ and $m=1$, and $\perp$ otherwise.

    Then, we consider the case $|s|>0$ and construct the formula
    $\flow_t^{>0}$ as follows. We first transform the formula $\text{rflow}_t(x,y)$ into a
    sentence $\text{rflow}_t^{xy}$ on the FO-signature whose alphabet
    is extended with pairs of Boolean values that indicate the
    positions of $x$ and $y$ respectively, so that 
    $s\models \text{rflow}_t(i,j)$ iff $(s,i,j)\models
    \text{rflow}_t^{xy}$, where $(s,i,j)$ is the string $s$ extended
    with the pair $(0,1)$ at position $i$, the pair $(1,0)$ at
    position $j$, and the pairs $(0,0)$ elsewhere. 

    The formula $\text{rflow}_t^{xy}$ is defined by
    $\exists x\exists y[\text{rflow}_t'(x,y)\wedge \bigvee_{a,b\in\Sigma}
    L_{(a,0,1)}(x)\wedge L_{(b,1,0)}(y)]$
    where $\text{rflow}_t'(x,y)$ is obtained by replacing all atoms of the
    form $L_{a}(z)$ by $\bigvee_{b_1,b_2\in\{0,1\}} L_{a,b_1,b_2}(z)$
    in $\text{rflow}_t(x,y)$.

    Since $\text{rflow}_t^{xy}$ is an FO-formula, there exists an
    aperiodic automaton over the alphabet $\Sigma\times \{0,1\}^2$
    that defines the same language. We intersect this automaton with 
    an (aperiodic) automaton that checks that the sequence of Boolean 
    pairs belongs to $(0,0)^*(0,1)(0,0)^*(1,0)(0,0)^*$. Let 
    $L^b$ ($b$ for Boolean) denote the aperiodic language defined by this
    automaton.

    Let us define the language $L$ of strings $u$ over $\Sigma\times
    \{0,1\}^2$ whose sequence of Boolean pairs is in
    $(0,1)(0,0)^*(1,0)$ and such there exists $u_1,u_2\in
    (\Sigma\times \{(0,0)\})^*$ such that $u_1uu_2\in L^b$.
    The language $L$ can be easily defined by some aperiodic
    automaton obtained from any aperiodic automaton defining $L^b$.
    We now define the language $\pi(L)$ obtained by projecting 
    $L$ on the component $\Sigma$, i.e. 
    $\pi(L)$ is the set of strings $s$ such that $s$ can be extended
    with Boolean pairs into a string $u$ such that $u\in L$. 
    The language $\pi(L)$ is aperiodic. Indeed, there is a bijection
    between the strings $s$ of $\pi(L)$ to the strings $L$, defined by
    extending the first symbol of $s$ with $(0,1)$, its last symbol by
    $(1,0)$, and the symbols in between by $(0,0)$. Aperiodic
    languages are not closed by projection in general, but they are
    preserved by bijective renaming \cite{dg08SIWT}. Therefore $\pi(L)$
    is aperiodic, and definable by some FO formula $\Phi_{\pi(L)}$. We
    let $\flow_t^{>0} = (\exists z.\top)\rightarrow \Phi_{\pi(L)}$.

    Let us prove the correctness of $\flow_t^{>0}$. Suppose that $s\models
    \flow_t^{>0}$ and $|s|>0$. Therefore there exists an extension $u$ of $s$ on the
    alphabet $\Sigma\times \{0,1\}^2$ such that $u\in L$,
    i.e. $u\in L$. By definition of $L$, the Boolean part of $u$ is
    necessarily of the form $(0,1)(0,0)^*(1,0)$, and 
    there exist $u_1,u_2\in (\Sigma\times \{(0,0)\})^*$ such that
    $u_1uu_2\models L^b$. By definition of $L^b$, we get $u\models
    \text{rflow}_t^{xy}$, i.e. $s_1ss_2\models \text{rflow}_t(i,j)$, where
    $s_1,s_2$ are the projections of $u_1,u_2$ on $\Sigma$, $i$ is
    the starting position of $s$ and $j$ its ending position. 
    In other words,
    $(q,R,X_\addr)\flows^{(s_1ss_2)[i{:}j]}_m(q',R',X_{\addr'})$ and in
    particular, $(q,R,X_\addr)\flows^s_m(q',R',X_{\addr'})$.

    Conversely, suppose that $(q,R,X_\addr)\flows^s
    (q',R',X_{\addr'})$ with $|s|>0$. 
    Since $(q,R)$ and $(q',R')$ are useful, there
    exists $s_1,s_2$ such that there exists a run from the initial
    pair $(q_0,R_0)$ to $(q,R)$ on $s_1$, and there exists an
    accepting run from $(q',R')$ to an accepting pair on $s_2$. 
    In particular $s_1ss_2\in dom(T_\la)$. Therefore
    $(q,R,X_\addr)\flows^{(s_1ss_2)[i{:}j]}_m(q',R',X_{\addr'})$, where $i$
    and $j$ are respectively the starting and ending position of $s$
    in $s_1ss_2$. Therefore $s_1ss_2\models \text{rflow}_t(i,j)$. If one
    extends $s_1$ with Boolean pairs $(0,0)$, $s$ with
    $(0,1)(0,0)^{n-2}(1,0)$, where $n = |s|$, and $s_2$ with the
    Boolean pairs $(0,0)$, one gets three strings $u_1,u,u_2$ such that 
    $u_1uu_2\models \text{rflow}_t^{xy}$, i.e. $u_1uu_2\in L^b$. By
    definition of $L$, we also get $u\in L$ and thus $u\in L$ and clearly, $s$ (the projection of $u$ on
    $\Sigma$) satisfies $\flow_t^{>0}$. 
\end{proof}

\subsection{Proof of Corollary \ref{coro:aperiodicity-sstla}}

\begin{proof}
From Lemma \ref{lem:flowFO}, it is clear that $T_\la$ is $1$-bounded. 
We show that the transition monoid $M$ of $T_\la$ is aperiodic. Let
$s\in\Sigma^*$ and let $t = (q,q',R,R',X_\addr,X_\addr',m)\in Q\times Q\times
2^{Q_A}\times 2^{Q_A}\times X\times X\times \mathbb{N}$.

If $(q,R)$ is not useful or $(q',R')$ is not useful, then for all $m$, 
$M_{s^m}[q,R][q',R'] = \bot$.

Now suppose that $(q,R)$ and $(q',R')$ are both useful. 
By Lemma \ref{lem:flowFO}, there exists an FO-sentence
$\text{flow}_t$ such that $s\models
\text{flow}_t$ iff $(q,R,\addr)\rightsquigarrow^s_n (q',R',\addr')$.
Let $b$ be the maximal quantifier rank of all the formulas
$\text{flow}_t$. By Proposition~\ref{prop:ktype}.\ref{prop:ape} there exists $n_0$ such that 
$s^{n_0} \equiv_b s^{n_0+1}$. Therefore there exists $n_0$ such that 
for all tuples $t$, $s^{n_0}\models \text{flow}_t$ iff $s^{n_0+1}\models
\text{flow}_t$, i.e. $(q,R,\addr)\rightsquigarrow^{s^{n_0}}_m
(q',R',\addr')$
iff $(q,R,\addr)\rightsquigarrow^{s^{n_0+1}}_m
(q',R',\addr')$. In other words $M_{s^{n_0}}[(q,R,\addr)][(q',R',\addr')]=m$ iff
$M_{s^{n_0+1}}[(q,R,\addr)][(q',R',\addr')]=m$. Therefore the transition monoid of
$T_\la$ is aperiodic.

\end{proof}

\section{$f_\halve$ is not FO-definable}
\label{halve}
\begin{proof}
    Let assume that it is FO-definable by some FO-transducer $T$ that outputs
    strings over a signature that does not contain the transitive
    closure of the successor relation. We show a contradiction, which
    will therefore imply the non FO-definability of $f_\halve$ by an
    FO-transducer that, additionally, must output the transitive
    closure of the successor relation.

    Let $k = qr(T)$ be the quantifier rank of $T$, and $C>0$ be the number of copies of $T$. We know by Proposition \ref{prop:ktype}.\ref{prop:ape} that for all $n\geq 2^{k+2}$,
    $a^n \equiv_{k+2} a^{n+1}$. Take such an $n$ and consider the string
    $s := a^{8nC}$.

    Clearly, $f_\halve(s) = a^{4nC}$. Therefore the output graph of
    $T(s)$ contains $4nC-1$ edges. Suppose that $s\models
    \phi_{\suc}^{c,d}(i,j)$ for some copies $c,d$ of $T$ and some
    input positions $i \leq j$ (the case $j\leq i$ is symmetric).

    Suppose that $j-i> n$ and $i>n$.  Therefore $s[1{:}i-1) = a^{i{-}1}\equiv_{k+2} a^i =
    s[1{:}i-1]$ and $s[i,j] = a^{j-i+1} \equiv_{k+2} a^{j-i} = s[i+1,j]$. Since $s\models
    \phi_\suc^{c,d}(i,j)$ and the quantifier rank of
    $\phi_\suc^{c,d}$ is at most $k$, by Proposition
    \ref{prop:ktype}.\ref{prop:indis}, it
    is also the case that $s\models \phi_{\suc}^{c,d}(i-1,j)$. It is a
    contradiction since it that case, there would be two incoming
    edges to the output node $j^d$, and the output would not be a
    string.

    A similar contradiction being obtained symmetrically for the case
    $j-i>n$ and $j < 8nC-n$, it is implies that necessarily, if  
    $j-i > n$, then $i\leq n$ and $j\geq 8nC-n$. In other
    words, either the edge $(i^c, j^d)$ is ``local'' or one of its 
    element is close from the extremities of $s$. In both cases, we
    show again a contradiction.

    Now, there exist necessarily two positions $i',j'$ and two copies $c',d'$ such that 
    $s\models \phi_{\suc}^{c',d'}(i',j')$ such that $n < i'$ and $j' < 8nC-n$. If it was
    not the case, then since the input nodes in $[1,n]$ and $[8nC-n, 8nC]$
    contribute to at most $2nC$ edges (otherwise the output would
    not be a string as two edges would have either same target or same
    source), there would not be a sufficient number of edges
    to define the output.

    Therefore, since $n<i'$ and $j'<8nC-n$, we have just shown that necessarily, it is the case that
    $|i'-j'|\leq n$. Since $s\models \phi_\suc^{c',d'}(i',j')$, by a similar reasoning as before (in particular by
    applying Proposition \ref{prop:ktype}.\ref{prop:indis}), we can
    show that many other edges can be obtained by shifting the edge
    $(i'^{c'}, j'^{d'})$ left or right. More precisely, for all $\ell\in\mathbb{Z}$ such that $max (i'+\ell, j'+\ell)\leq 7n$ and 
    $min (i'+\ell, j'+\ell)\geq n$, it is the case that 
    $s\models \Phi_{\suc}^{c',d'}(i'+\ell, j'+\ell)$. 
Since there exist
    at least $8nC-3n$ such $\ell$ (because $|i'-j'|\leq n$ and $n<
    i'$ and $j'<8nC-n$), it means that the output graph of $s$ by $T$
    contains at least $8nC-3n$ edges, which is a
    contradiction. Indeed, we know that the output contains exactly
    $4nC-1$ edges, and $8nC-3n-4nC+1 = 4nC-3n + 1 > 0$. 
\end{proof}


\begin{thebibliography}{10}

\bibitem{AC10}
R.~Alur and P.~{\v C}ern\'y.
\newblock {Expressiveness of streaming string transducers}.
\newblock In {\em FSTTCS}, volume~8, pages 1--12, 2010.

\bibitem{AC11}
R.~Alur and P.~{\v C}ern\'y.
\newblock Streaming transducers for algorithmic verification of single-pass
  list-processing programs.
\newblock In {\em POPL}, pages 599--610, 2011.

\bibitem{AD12}
R.~Alur and L.~D'Antoni.
\newblock Streaming tree transducers.
\newblock In {\em ICALP (2)}, pages 42--53, 2012.

\bibitem{ADD+11}
R.~Alur, L.~D'Antoni, J.~V. Deshmukh, M.~Raghothaman, and Y.~Yuan.
\newblock Regular functions and cost register automata.
\newblock In {\em LICS}, 2013.

\bibitem{ADT13}
R.~Alur, A.~Durand-Gasselin, and A.~Trivedi.
\newblock From monadic second-order definable string transformations to
  transducers.
\newblock In {\em LICS}, pages 458--467, 2013.

\bibitem{FiliotTrivediLics12}
R.~Alur, E.~Filiot, and A.~Trivedi.
\newblock Regular transformations of infinite strings.
\newblock In {\em LICS}, pages 65--74, 2012.

\bibitem{DBLP:journals/corr/Bojanczyk13}
M.~Bojanczyk.
\newblock Transducers with origin information.
\newblock In {\em ICALP}, 2014.
\newblock To appear.

\bibitem{Bu60}
J.~R. B\"uchi.
\newblock Weak second-order arithmetic and finite automata.
\newblock {\em Zeitschrift f\"ur Mathematische Logik und Grundlagen der
  Mathematik}, 6(1--6):66--92, 1960.

\bibitem{Carton13}
O.~Carton and L.~Dartois.
\newblock Aperiodic two-way transducers.
\newblock In {\em Highlights of Logic, Automata and Games}, 2013.
\newblock Oral communication, slides available at
  \url{http://highlights-conference.org/pub/3-1-Dartois.pdf}.

\bibitem{cho-dung}
S.~Cho and D.~T. Huynh.
\newblock Finite state automaton aperiodicity is pspace-complete.
\newblock {\em Theoretical Computer Science}, 88:99--116, 1991.

\bibitem{Cour94}
B.~Courcelle.
\newblock Monadic second-order definable graph transductions: a survey.
\newblock {\em Theoretical Computer Science}, 126(1):53--75, 1994.

\bibitem{dg08SIWT}
V.~Diekert and P.~Gastin.
\newblock First-order definable languages.
\newblock In {\em Logic and Automata: History and Perspectives}, Texts in Logic
  and Games, pages 261--306. Amsterdam University Press, 2008.

\bibitem{Elg61}
C.~C. Elgot.
\newblock Decision problems of finite automata design and related arithmetics.
\newblock {\em In Transactions of the American Mathematical Society},
  98(1):21--51, 1961.

\bibitem{EH01}
J.~Engelfriet and H.~J. Hoogeboom.
\newblock {MSO} definable string transductions and two-way finite-state
  transducers.
\newblock {\em ACM Trans. Comput. Logic}, 2:216--254, 2001.

\bibitem{Engelfriet03}
J.~Engelfriet and S.~Maneth.
\newblock Macro tree translations of linear size increase are {MSO} definable.
\newblock {\em SIAM Journal on Computing}, 32:950--1006, 2003.

\bibitem{McKenzieEtAl06}
P.~McKenzie, T.~Schwentick, D.~Therien, and H.~Vollmer.
\newblock The many faces of a translation.
\newblock {\em JCSS}, 72, 2006.

\bibitem{stern}
J.~Stern.
\newblock Complexity of some problems from the theory of automata.
\newblock {\em Information and Control}, 66:163--176, 1985.

\bibitem{Strau94}
H.~Straubing.
\newblock {\em Finite Automata, Formal Logic, and Circuit Complexity}.
\newblock Birkh{\"a}user, Boston, 1994.

\bibitem{Tho96}
W.~Thomas.
\newblock Languages, automata, and logic.
\newblock In {\em Handbook of Formal Languages}, pages 389--455. Springer,
  1996.

\bibitem{Tra62}
B.~A. Trakhtenbrot.
\newblock Finite automata and monadic second order logic.
\newblock {\em Siberian Mathematical Journal}, 3:101--131, 1962.

\end{thebibliography}
\end{document}